\documentclass{article}
\usepackage{fullpage}

\usepackage[utf8]{inputenc} %
\usepackage[T1]{fontenc}    %
\usepackage[colorlinks,citecolor=teal]{hyperref}       %

\usepackage{url}            %
\usepackage{booktabs}       %
\usepackage{amsfonts}       %
\usepackage{nicefrac}       %
\usepackage{microtype}      %
\usepackage[table,xcdraw]{xcolor}         %
\usepackage[numbers]{natbib} %
\usepackage{authblk}
\usepackage{comment} %
\usepackage{xspace}
\usepackage{setspace}
\usepackage{amssymb}
\usepackage{amsmath}
\usepackage{amsthm, thmtools, thm-restate}
\usepackage{mathtools}
\usepackage{array}
\usepackage{physics}
\usepackage{dsfont}
\usepackage{cases}
\usepackage[shortlabels]{enumitem} %
\usepackage{algorithm}
\usepackage[noend]{algpseudocode}
\usepackage{tikz}
\usetikzlibrary{calc, positioning, external, arrows.meta} %
\usepackage{forest}
\usepackage{bm}
\usepackage{complexity}
\usepackage{lipsum}
\usepackage{multirow}
\usepackage{siunitx} %
\usepackage{adjustbox}
\usepackage{graphicx}

\usepackage{subfiles}
\usepackage{subcaption}
\usepackage{caption}
\usepackage[capitalize]{cleveref}
\usepackage{autonum} %

\input{defs}
\tikzset{
        every path/.style={-},
        every node/.style={draw},
        infoset/.style={-, densely dotted, ultra thick},
        infoset1/.style={infoset, color=p1color},
        infoset2/.style={infoset, color=p2color},
        terminal/.style={},
        extraedge/.style={parent anchor=south, child anchor=north},
      ilabel/.style={fill=white, inner sep=0pt, draw=none},
      pubnode/.style={draw=red, inner sep=0.1cm},
}
\forestset{
        default preamble={for tree={
        parent anchor=south, child anchor=north
}},
  p1/.style={
      regular polygon,
      regular polygon sides=3,
      inner sep=0pt, fill=p1color, thick, text=white, draw=p1color, font=\bfseries\sffamily\footnotesize, text height=0.6em,text depth=0.1em},
  p2/.style={p1, shape border rotate=180, fill=p2color, text=white, draw=p2color},
  parent/.style={no edge,tikz={\draw (#1.parent anchor) to (!.child anchor);}},
  parentd/.style n args={2}{no edge,tikz={\draw[ultra thick, p#2color] (#1.parent anchor) to (!.child anchor); }},
  nat/.style={font=\bfseries\sffamily\footnotesize},
  terminal/.style={draw=none, inner sep=2pt},
  action/.style n args={1}{edge label={node[midway, fill=white, inner sep=1pt, draw=none, font=\sf\footnotesize] {#1}}},
  d/.style={edge={ultra thick, draw={p#1color}}},
  comment/.style={no edge, draw=none, align=center, font=\tiny\sf},
}
\bracketset{action character=@}

\begin{document}

\title{Efficient representations for team and imperfect-recall \\ equilibrium computation}

\author[1]{Luca Carminati\thanks{Equal contribution; author order randomized}}
\author[2]{Brian Hu Zhang$^*$}
\author[1]{Federico Cacciamani}
\author[3,4]{Junkang Li}
\author[2]{Gabriele Farina}
\author[1]{Nicola Gatti}
\author[5,6]{Tuomas Sandholm}
\affil[1]{\small DEIB, Politecnico di Milano}

\affil[2]{\small MIT}

\affil[3]{\small NukkAI}

\affil[4]{\small GREYC, Université Caen Normandie}

\affil[5]{\small Computer Science Department, Carnegie Mellon University}
            
\affil[6]{\small Additional affiliations: Strategy Robot, Inc., Strategic Machine, Inc., Optimized Markets, Inc.}
\date{}
\maketitle

\begin{abstract}

Equilibrium finding in two-{\em player} zero-sum games with perfect recall is a well-studied topic that has led to many breakthroughs in computational game theory. This paper aims to generalize such techniques to (timeable) two-player zero-sum games with imperfect recall, or equivalently to two-{\em team} zero-sum games. In this setting, the problem of computing a {\em mixed-strategy Nash equilibrium} (or, equivalently, a {\em team maxmin equilibrium with correlation}) is known to be \NP-hard.
We connect the imperfect-recall setting with its perfect-recall counterpart through a novel construction we call the {\em belief game}.
This is a perfect-recall game {\em equivalent} to a given (timeable) two-player zero-sum game with imperfect recall. The belief game may be exponentially larger than the original game but can be solved using any standard method. We then show that the strategy spaces of the two players in the belief game can be {\em directly} represented as a DAG, leading to a possibly exponential speedup. We call this the {\em team belief DAG} (TB-DAG).
The TB-DAG simultaneously enjoys essentially optimal parameterized complexity bounds and the advantages of efficient regret minimization techniques. 
Along the way, we show $\Delta_2^\P$-completeness and $\Sigma_2^\P$-completeness of finding Nash equilibria in both mixed and behavioral strategies for the class of games we consider.
Experimentally, we show that the TB-DAG, when paired with existing learning techniques, yields state-of-the-art performance on a wide variety of benchmark team games.
\end{abstract}

\section{Introduction}\label{sec:introduction}

Two-player zero-sum imperfect-information extensive-form games are a classic setting within the game theory literature with a rich history of notable achievements. They represent sequential interactions between two adversaries where any utility gained by one player implies a loss for the other. A crucial component is that players may have \emph{imperfect information}, that is, they may not be able to distinguish the current node of the game from others.
In this setting, \emph{Nash equilibria in mixed strategies} are the most natural solution concept for modeling rational value-maximizing players. Mixed strategies specify the behavior of a player as a distribution over \emph{pure (deterministic) strategies}.
However, the exponential number of such strategies makes the computation of Nash equilibria potentially inefficient.
A key assumption to circumvent this issue is \emph{perfect recall}. In a perfect-recall game, the players never forget previously received information or played actions. 
When this assumption is satisfied,
\begin{enumerate}
    \item \emph{Kuhn's theorem}~\cite{Kuhn50:Extensive} states that \emph{mixed strategies} are equivalent to \emph{behavioral strategies}, which are the strategies expressible as a product of distributions over actions at each decision point. %
    \item The \emph{sequence-form representation}~\cite{Stengel96:Efficient,Romanovskii62:Reduction} of the strategy spaces enables efficient computation of Nash equilibria via a wide variety of different methods.
    In particular, uncoupled learning dynamics such as \emph{counterfactual regret minimization} (CFR)~\cite{Zinkevich07:Regret} converge to a Nash equilibrium by employing a regret minimizer at each decision point of the strategy tree.
\end{enumerate}
There have been significant recent speed improvements to CFR-based techniques~\cite{Tammelin15:Solving,Brown19:Solving,Farina21:Faster,Zhang24:Faster}, and other techniques have been built on top of CFR-based techniques, for example, abstraction algorithms~\cite{Sandholm15:Abstraction,Sandholm15:Solving}, subgame solving~\cite{Gilpin06:Competitive,Ganzfried15:Endgame,Moravcik16:Refining,Brown17:Safe,Moravvcik17:DeepStack,Brown18:Superhuman,Brown19:Superhuman}, further enhancing scalability.
Notable results on large-scale games include poker~\cite{Bowling15:Heads,Moravvcik17:DeepStack,Brown18:Superhuman,Brown19:Superhuman}, Stratego~\cite{Perolat22:Mastering}, and Diplomacy~\cite{FAIR22:Human}.

This work seeks to extend these techniques beyond the perfect-recall two-player zero-sum setting. In particular, we focus on computing mixed Nash equilibria in the two equivalent settings of \emph{imperfect-recall games} and \emph{adversarial team games}\footnote{This equivalence is formalized in \Cref{sec:adversarial_team_games}.}, for which it is known that computing a Nash equilibrium is NP-hard \cite{Koller92:Complexity}.

Two-player zero-sum imperfect-recall games are characterized by players who may forget information at some point in the game. 
In this case, a mixed strategy corresponds to a distribution over pure strategies, while a behavioral strategy corresponds to a distribution that performs an independent sampling procedure at each decision point. Unlike for perfect-recall games, Kuhn's theorem does not apply in imperfect-recall games: mixed strategies can in general be more expressive than behavioral strategies.
Imperfect-recall games have been employed in the literature to compress a game representation through forgetfulness (this is the case of some abstraction techniques \cite{Waugh09:Abstraction,Lanctot12:No,Kroer16:Imperfect}), or by considering human-like agents with imperfect memories \cite{Camerer03:Behavioral}.

Adversarial team games portray two teams of agents facing each other adversarially. Each team member has utilities identical to their teammates and opposite to members of the opposing team. Effective team coordination is a non-trivial challenge in this setting because team members may have different imperfect information about the current node and no communication channels are available during the game. 
Intuitively, a player cannot distinguish nodes that are different due to private information revealed to a teammate (such as private cards revealed to them solely).
In this case, mixed strategies correspond to strategies coordinated \emph{before the start of the game} through \emph{ex-ante coordination}, while behavioral strategies represent strategies that are not coordinated, in the sense that each agent samples their actions independently from other teammates.
Recreational and non-recreational examples of team games include Bridge, security games with multiple defenders and attackers \cite{Jiang13:Defender}, and poker with colluding agents.

Overall, team games are a more common application setting than imperfect-recall games; they have many competing works in the equilibrium computation literature and allow a more intuitive game description. On the other hand, imperfect-recall games yield a cleaner formalism. As these two perspectives are equivalent for our purposes, we choose to adopt an imperfect-recall perspective throughout the rest of the paper to simplify the notation, while using team games to make more intuitive examples for some of the notions introduced.

The main objective of this paper is to propose a novel representation for team/imperfect-recall games by constructing an equivalent \emph{perfect-recall} two-player zero-sum game. This enables the use of all the solving techniques previously developed for perfect-recall two-player zero-sum games.

We now summarize the contributions of the paper.
In \Cref{sec:beliefs_obs,sec:auxgame}, we present an algorithm that converts any two-player zero-sum imperfect-recall game into a strategically-equivalent {\em perfect-recall game} which we call the \emph{belief game}. In \Cref{sec:auxgame_strategic_equivalence} we formally prove the equivalence between the two games, and in \Cref{sec:worst-case-size-aux} we show worst-case bounds on the size of the belief game in terms of the number of histories of the original game. 
In particular, we show that in the worst case the number of histories of the belief game is $O(b^{dk})$, where $b$ is the maximum branching factor of the original game, $d$ is its depth, and $k$ is a parameter we introduce called the \emph{information complexity}, which intuitively measures the amount of information that can be forgotten by the player---or, in the case of team games, the amount of {\em information asymmetry} between players on the team.

In \Cref{sec:dag-decision-problems}, we introduce a notion of DAG-form decision-making that we use to generalize counterfactual regret minimization (CFR) beyond tree-form games. While we introduce it for the purpose of applying it to imperfect-recall games, we believe it to be of independent interest as well.

In \Cref{sec:tbdag}, we use DAG-form decision problems to efficiently represent each player's strategy space in the belief game through a construction we call the \emph{team-belief DAG} (TB-DAG). We show that the TB-DAG representation of a game with imperfect recall can be exponentially smaller than the size of the belief game and that it can be constructed directly from the original game without first constructing the belief game, thus leading to exponentially faster algorithms in the worst case. %
This construction improves the worst-case efficiency\footnote{By {\em efficiency} here we mean the size of the representation of the strategy spaces of the players. Algorithms such as CFR have per-iteration complexity that scales linearly in this size.} of our technique to $O(|\H|(b + 1)^{k+1})$, where $|\H|$ is the number of nodes in the original game. We also show that this bound is essentially optimal: under reasonable computational assumptions (namely, the exponential time hypothesis), we show that there cannot exist an algorithm for solving even single-player games of imperfect recall the runtime of which is of the form $f(k) \poly(|\H|)$ \changed{(that is, with no dependence on $b$)}, for {\em any} function $f$.

In \Cref{sec:complexity} we investigate the computational complexity of computing mixed Nash equilibria with imperfect recall. We prove that computing a max-min strategy in mixed or behavioral strategies in games where both players have imperfect recall is $\Delta_2^\P$-complete and $\Sigma_2^\P$-complete respectively.

In \Cref{sec:experiments}, we evaluate our methods empirically by benchmarking our construction on a standard testbed of imperfect-information games, compared to state-of-the-art baselines. We find that our technique allows much faster equilibrium computation when the information complexity $k$ of the game is low.

\section{Preliminaries}\label{sec:prelims}
In the following, we introduce the concepts of \emph{game}, \emph{strategy}, and \emph{equilibrium} with a focus on classic results regarding strategy representation. We remark that the paper's main contribution is to extend these results from the setting of \emph{perfect-recall} games to imperfect-recall and adversarial team games.

\subsection{Extensive-form games}
A canonical representation for sequential games is the one of \emph{extensive-form games}. In the following, we introduce the notation adopted to describe their different components. 
A crucial aspect is that players have imperfect information about the state of the game during gameplay. That is, they may not know the exact point of the game that has been reached; instead, they have to consider a set of possible candidates that fit the information they know. These sets are called \emph{information sets} (in short \emph{infosets}).

\begin{definition}
    [Extensive-form game]
    An \emph{extensive-form game} $G$ is a tuple $(\N, \H, \Z, \A, \I, \vec p, u)$ that describes a sequential game, where
    \begin{itemize}
        \item \N is the set of \emph{players} in the game. A special symbol $\nature \in \N$ is used to indicate \emph{Nature} player (also called \emph{chance}). Nature encodes all the stochasticity of the environment known a priori, and therefore it picks actions according to a fixed probability distribution. We will commonly use a letter $i$ to indicate any player $i \in \N$ and $-i$ to indicate $\mc N\setminus i$.
        \item \H is a rooted tree of \emph{nodes} (also called \emph{histories}) of the game. Terminal nodes $\Z \subset \H$ represent possible game endings after which nobody plays anymore. At every non-terminal node in $\H \setminus \Z$, there is exactly one \emph{active} player who picks an action. $\H \setminus \Z$ is partitioned into sets $\{\H_i\}_{i \in \N}$, where $\H_i$ contains all nodes at which player $i$ is the active player. We use $h$ to indicate a generic node $h \in \H$,  $z$ to indicate a generic terminal node $z \in \Z$, and $\Root$ to indicate the root node of the game tree. Moreover, we use the symbol $\H_{-i}$ to indicate the nodes at which player $i$ is \emph{inactive}, that is, nodes at which another player $j \in \N \setminus i$ is playing.
        \item \A is the set of actions in the game. A non-empty set of actions $A_h \subseteq \A$ is available to $i$ at any node $h \in \H_i$. By assumption, each player knows the actions available at each node. Playing an action $a$ at $h$ leads to a node $h'$ which we denote as $ha \equiv h'$.
        \item \I is the set of all information sets. An information set $I \subseteq \H_i$ is a set of nodes that are indistinguishable by the active player $i$. It is assumed that players know the actions available to them whenever they are active, so it is required that for each $I \in \I$, all $h \in I$ have identical $A_h$. This common action set is denoted by $A_I$. $\I$ is partitioned into $\{\I_i\}_{i \in \N\setminus\nature}$, where $\I_i$ is the set of information sets of player $i$. Every node $h\in\H_i$ where player $i$ is active is contained in exactly one information set $I_h\in\I_i$.
        \item $\vec p \in [0,1]^\Z$ is a vector encoding Nature's fixed strategy. $\vec p[z]$ denotes the probability that Nature plays all the actions leading to $z \in \Z$.
        \item For every player $i \in \N \setminus \nature$, $u_i: \Z \to \R$ is the \emph{utility} function of player $i$, which associates each terminal node to the payoff that player $i$ receives when the game reaches this terminal node.
    \end{itemize}
\end{definition}
In the rest of this paper, we will use the following notation.
\begin{itemize}
    \item %
    The symbol $|\cdot|$ denotes the length of a node $h$, that is, the depth of the node from the root. The root node $\Root$ has length $0$ and for each $h\in\H\setminus\Z$ and $a\in\A_h$, it holds that $|ha| = |h|+1$. \changed{The {\em depth} $d$ of $G$ is the length of the longest history in $G$. The {\em branching factor} $b$ of a game is the maximum number of actions available at any information set: $b = \max_{h \in \H \setminus \Z} |A_h|$.}
    \item A set of nodes $S$ (usually $S$ will be an information set) precedes another set $S'$, denoted by $S \preceq S'$, if there are nodes $h \in S, h' \in S'$ such that $h'$ is a descendant of $h$ (or $h'=h$). If one of $S$ and $S'$ is a singleton, we omit the braces; for example, $h \preceq I$ means there is a path from $h$ to some node $h' \in I$.
\end{itemize}
The analysis of the information structure of an extensive-form game requires the introduction of specific concepts called \emph{sequences} and \emph{connectedness}.
Sequences characterize the structure of the information revealed to a player at any node of the game.
\begin{definition}[Sequence] \label{def:sequence} The \emph{sequence} $\sigma_i(h)$ of player $i$ at node $h$  is the ordered list of pairs of infosets reached and actions played, by player $i$, on the $\Root \to h$ path, excluding the information set of $h$. A special sequence $\Root$ corresponding to an empty list is also considered for each player.
\end{definition}

The set containing all sequences of all players and of a single player are respectively denoted by $\Sigma$ and $\Sigma_i$, with $\Sigma \equiv \bigcup_{i \in \N\setminus\nature} \Sigma_i$.
We use the notation $\sigma_i + (I,a)$ to indicate the sequence obtained by extending $\sigma_i$ by adding an (infoset, action) pair $(I,a)$.

This paper focuses on extensive-form games that satisfy some regularity properties.

\begin{definition}
    [Timeability] \label{def:timeability}
    An extensive-form game $G$ is \emph{timeable} if any path from the root to any node in the same infoset has the same length (equivalently, all histories belonging to the same infoset have the same depth). Formally, $G$ is timeable if for every infoset $I \in \mc I$ and every $h, h' \in I$, we have $|h| = |h'|$.
\end{definition}
As an assumption, we will consider timeable games throughout the paper. Intuitively, this corresponds to assuming that all players know the number of actions played in the game whenever they play.

\begin{definition}
    [Perfect recall]\label{def:perfect_recall}
    A player $i$ has \emph{perfect recall} if all histories in the same infoset of that player have the same sequence. Formally, player $i$ has perfect recall if for every infoset $I \in \mc I_i$ and every $h, h' \in I$, we have $\sigma_i(h) = \sigma_i(h')$.
    The common sequence is denoted by $\sigma_I$. 
\end{definition}
When a player has perfect recall, their nonempty sequences are uniquely identified by the last infoset-action pair contained within them. Thus, when player $i$ has perfect recall, we will use $Ia \in \Sigma_i$ to indicate the sequence of player $i$ that ends with playing action $a$ at infoset $I$. 
Intuitively, assuming that all the players have perfect recall corresponds to assuming that each of them never forgets information they previously acquired during the game's unfolding.
A game is said to be \emph{perfect-recall} when all players in the game have perfect recall.
We say that a player or a game is \emph{imperfect-recall} whenever the perfect recall condition does not hold.

We remark that the timeability assumption excludes \emph{absentmindedness}. That is, players never have nodes on the same path belonging to an identical information set. Intuitively, knowing the number of actions played in the game allows the player to distinguish nodes on the same path.

\subsection{Strategies}

Extensive-form games represent the sequential interaction among players that are required to take an action given a specific information state. The description of a possibly stochastic behavior of a player in a game is called a \emph{strategy}. \changed{ In this section, we will define the various kinds of strategy that we need for this paper.

\begin{definition}[Pure strategy]\label{def:pure_strategy}
    A {\em pure strategy} $\vec\pi$ for player $i$ is a vector $\vec\pi \in \bigtimes_{I \in \I_i} \A_I$. The {\em realization form} of a pure strategy is the vector $\xvec \in \{0,1\}^\Z$ where $\xvec[z] = 1$ if and only if the pure strategy prescribes all player actions on the path $\treepath{\Root}{z}$. 
\end{definition}

The definition of a pure strategy includes some redundancy. Suppose that $I$ is an information set such that $\vx[h] = 0$ for every $h \in I$, that is, no node in $I$ is played by $\vec\pi$. Then the value of $\vec\pi[I]$ does not matter, that is, the realization form does not depend on $\vec\pi[I]$. We are thus motivated to define the reduced pure strategy.

\begin{definition}\label{def:played_sequence}
    A sequence $\sigma_i(h)$ is \emph{played} or {\em reached} by pure strategy $\vec\pi$ for player $i$ if for all $(I', a') \in \sigma_i(h)$ we have $\vec\pi[I'] = a'$. A node $h \in \H_i$ is played if $\sigma_i(h)$ is played. An infoset $I$ is played if some $h \in I$ is played.
\end{definition}
\begin{definition}
    The {\em reduced form} of a pure strategy is the vector $\vec\pi \in \bigtimes_{I \in \I_i} \A_I \cup \{\bot\}$, for which $\vec\pi[I]$ is replaced by $\bot$ if and only $I$ is not played by $\vec\pi$.
\end{definition}
Two pure strategies with the same reduced form also have the same realization form, and are thus equivalent for all purposes in this paper. Thus, for the rest of the paper, we will use only the reduced form.

\begin{definition}[Mixed strategy] A \emph{mixed strategy} is a distribution over pure strategies. The realization form of a mixed strategy is the vector $\xvec \in [0,1]^\Z$ obtained through the convex combination of realization forms of pure strategies according to the probabilities in the mixed strategy.
\end{definition}

It is to be noted that a realization-form strategy, in general, is not a probability distribution over $\Z$. \changed{Instead, $\vx_i[z]$ is only player $i$'s contribution to this probability distribution. The product $\prod_{i\in\N} \xvec_i[z]$, however, {\em is} a probability distribution representing the probability that a terminal node $z$ is reached when each player plays strategy $\xvec_i$.}

\begin{definition}
    [Behavioral strategy] A \emph{behavioral strategy} is a mixed strategy for which the actions selected at different information sets are mutually independent.
\end{definition}
We will use the following notation: for each player $i$, $\Pi_i, \X_i$, and $\hat\X_i$ denote the sets of pure, mixed, and behavioral strategies respectively, all in realization form. Similarly, $\Pi := \prod_{i \in \N \setminus 0} \Pi_i, \X = \prod_{i \in \N \setminus 0}$, and $\hat\X = \prod_{i \in \N \setminus 0} \hat\X_i$ denote the sets of pure, mixed, and behavioral strategy {\em profiles}, respectively.

}

A classic question is to characterize the strategy spaces given their different definitions and to find conditions to guarantee the equivalence of these different notions.
In particular, the local decomposability at each infoset of behavioral strategies offers efficiency advantages over mixed strategies for describing strategies in an extensive-form game. 
Luckily, a sufficient condition for their equivalence is that the considered game is perfect-recall, as stated in the following classic theorem.
\begin{theorem}
    [Kuhn's theorem, \citep{Kuhn53:Extensive}]\label{th:kuhn}
    In every extensive-form game with perfect recall, $\hat \X_i = \X_i$, that is, for each mixed strategy, there is a behavioral strategy with the same realization form, and vice versa.
\end{theorem}

\citet{Kuhn53:Extensive} also showed that perfect recall is a necessary condition for $\hat \X_i = \X_i$; in imperfect-recall games (but without absentmindedness; cf. \Cref{sec:assumptions}), only $\hat \X_i \subseteq \X_i$ holds.

\subsection{Tree-form decision problems and representation of realization-form strategies}\label{sec:tree-form-decision-problems}
This subsection describes \emph{tree-form decision problems}, originally introduced by \citet{Farina19:Regret}. They are highly efficient representations of the strategy space of a specific player in the game.
In perfect-recall games, tree-form decision problems lead to the \emph{sequence-form} representation of the strategy space, first noted by \citet{Romanovskii62:Reduction} and later rediscovered by \citet{Stengel96:Efficient}. The tree-form decision problem will later be generalized to DAG structures arising from imperfect-recall games.

\begin{definition}[Tree-form decision problem]
    A {\em tree-form decision problem} is a rooted tree of nodes representing an interaction between a {\em player} and an {\em environment}. Each node $s$ is one of
    the following types: 
    \begin{itemize}
        \item {\em decision points} $s \in \mc D$, at which the player takes an action, and
        \item {\em observation points} $s \in \mc S$, at which the environment selects the next state.
    \end{itemize}
\end{definition}

We will assume WLOG that the root node ($\Root$) and all terminal nodes are observation points and that decision points and observation points alternate. At a decision point $s$, the set of legal actions will be denoted by $\A_s$. Each action $a \in \A_s$ leads to a different observation node, denoted by $sa$. The parent observation node of $s$ is denoted by $p_s$. For notational simplicity, when $\vec x \in \R^{\mc S}$ is any vector indexed by observation points and $s$ is a decision point, we will use $\vec x[s*] \in \R^{\mc A_s}$ to denote the subvector of $\vec x$ indexed only by the children of $s$.

Similar definitions of strategy notions as those from extensive-form games hold in tree-form decision problems.
A {\em pure strategy} is an assignment of one action to each decision node. Pure strategies are associated with vectors $\vec x \in \{0, 1\}^{\mc S}$ in the following way: for a node $s \in {\mc S}$, $\vec x[s]=1$ if and only if the player plays every action on the path from the root to $s$. We will call such vectors {\em tree-form (pure) strategies}. A {\em tree-form mixed strategy} $\vec x \in \mc Q$ is a convex combination of tree-form pure strategies. 

The set of tree-form mixed strategies has a natural description as a set of linear constraints. Namely: 
\begin{proposition}\label{pr:seq_form_constraints}
    $\mc Q$ is the set of nonnegative vectors $\vec q \in \R^{\mc S}$ satisfying the following constraints:
    \begin{subequations}
        \begin{numcases}{}
            \vec q[\Root] = 1 \label{eq:seq_strategy_root}\\
            \vec q[p_s] = \sum_{a \in \A_s} \vec q[s a] \label{eq:seq_strategy_split} & for all $s \in \mc D$.
        \end{numcases}
    \end{subequations}
\end{proposition}

The tree-form constraints are \emph{probability flow} constraints: the probability of reaching a node starts at $1$ at the root (\ref{eq:seq_strategy_root}), and is split at each decision point on the available actions (\ref{eq:seq_strategy_split}).

Whenever a player $i$ has perfect recall in an extensive-form game, a tree-form decision problem can be extracted from the game by considering the point of view of player $i$.
In this tree-form decision problem, we have $\mc D = \mc I_i$ and $\mc S = \Sigma_i$, i.e.\ decision points are given by information sets of player $i$, and every observation point corresponds to a sequence of player $i$.
In this case, tree-form mixed strategies correspond to the so-called \emph{sequence-form} strategies, which are in natural correspondence with realization-form strategies.

\begin{definition}
    [Sequence-form strategies]
    In a perfect-recall game, a \emph{sequence-form strategy} of player $i$ is a nonnegative vector $\vec q \in \mc Q$, where $\mc Q$ is defined as in \Cref{pr:seq_form_constraints} with $\mc D = \mc I_i$ and $\mc S = \Sigma_i$.
\end{definition}

By definition, a sequence-form strategy $\vec q \in \mc Q$ of player $i$ associates to each sequence $\sigma \in \Sigma_i$ the probability that the player plays all the actions in $\sigma$. In particular, the realization-form strategy of the sequence-form strategy $\vec q$ is $\vec x_i \in \X_i$ where $\vec x_i[z] = \vec q[\sigma_i(z)]$.

This yields a much more compact representation of mixed strategies than the naive one offered by their definition. Indeed, mixed strategies are a probability distribution over an exponential number of pure strategies, while their corresponding sequence-form strategies have one element per sequence (which is uniquely identified by its last infoset-action pair), similar to behavioral strategies.

\subsection{Two-player zero-sum games and max-min strategies}

A (mixed) {\em strategy profile} is a tuple $\vec x = (\vec x_1, \dots, \vec x_n) \in \X_1 \times \dots \times \X_n$ consisting of one mixed strategy in realization form for each player. The {\em expected utility} of player $i$ is $u_i(\vec x) := \E_{z \sim \vec x} u_i(z)$, where $z \sim \vec x$ denotes sampling a terminal node $z$ by following the profile $\vec x$, that is, $u_i(\vec x) \coloneqq \sum_{z \in \Z} u_i(z) \vec p[z]\prod_{j \in \N \setminus \nature} \xvec_j[z]$.
A crucial property of the realization form, which we will use repeatedly, is that the utility functions $u_i : \X_1 \times \dots \X_n \to \R$ are linear in each $\vec x_i$.

After deciding the most suitable representation for strategies in a given problem, an objective to be optimized to find the optimal strategy is required. This is provided by \emph{solution concepts}, which describe the characteristics that make a strategy optimal in a context.
One solution concept relevant for the paper is the {\em Nash equilibrium}:

\begin{definition}
    [Nash equilibrium] A \emph{Nash equilibrium} is a mixed strategy profile $\vec x$ such that no player has a profitable deviation. That is, for every player $i$ and every $\vec x_i' \in \X_i$, we have $u_i(\vec x_i', \vec x_{-i}) \le u_i(\vec x)$.
\end{definition}

 The most studied type of extensive-form games are \emph{two-player zero-sum games}. In this class of games, two players face each other in a setting where a payoff for one is a cost for the other. In this section, we define two-player zero-sum games and explore some basic solution concepts for them.
 
 \begin{definition}
     [Two-player zero-sum game]
     An extensive-form game is two-player zero-sum if
     \begin{itemize}[noitemsep,topsep=0pt,parsep=0pt,partopsep=0pt]
         \item the player set contains Nature and two players, called \pmax and \pmin, namely $\N = \{\nature, \pmax, \pmin \}$, and
         \item the utilities of the two players are opposite, namely  $u_\pmax = -u_\pmin$.
     \end{itemize}

 \end{definition}
In this case, we will use a single utility function $u = u_\pmax = -u_\pmin$ and use $\X$ and $\Y$ to represent the two players' set of realization-form mixed strategies, respectively.
In two-player zero-sum games, Nash equilibria $(\vec x, \vec y) \in \X \times \Y$ are the saddle-point solutions to the bilinear saddle-point problem  \begin{align}
        \max_{\xvec \in \X} \min_{\yvec \in \Y} \xvec^\top \vec U \yvec = \min_{\yvec \in \Y} \max_{\xvec \in \X}  \xvec^\top \vec U \yvec \label{eq:maxmin}
    \end{align}
where $\vec U = \textrm{diag}(\{u(z) \vec p[z]\}_{z \in \Z})$  so that $\vec x^\top \vec U \vec y = \E_{z \sim (\vec x, \vec y)} u(z)$. In addition, in all Nash equilibria of a two-player zero-sum game, \pmax has the same expected utility \cite{vonNeumann28:Zur}, which is called the {\em Nash value} of the game. As discussed in the previous subsection, in games with {\em perfect recall}, mixed and behavioral strategies are equivalent by \Cref{th:kuhn}. In games without perfect recall, the above definition of Nash equilibrium is still valid; however, it will also be useful to define the behavioral max-min strategy:
\begin{definition}[Behavioral max-min strategy]
    In a two-player zero-sum game, a \emph{behavioral max-min strategy} $\xvec \in \hat\X$ is a solution to the optimization problem
    \[
        \max_{\xvec \in \hat\X} \min_{\yvec \in \hat\Y} \xvec^\top \vec U \yvec.
    \]
\end{definition}
The {\em behavioral max-min value} is the optimal value of the above problem.
Since $\hat \X$ and $\hat\Y$ are not necessarily convex sets, the minimax theorem does not apply, so the maximization and minimization cannot be swapped. Therefore, the behavioral max-min strategy is not an {\em equilibrium}. Further, in games with imperfect recall, the tree-form decision problem is not a valid representation of the set of realization-form strategies. Therefore, we will need different techniques to tackle such games.

\subsection{Online learning in tree-form decision problems and extensive-form games}

In this work, we will leverage the \emph{online convex optimization}~\cite{Zinkevich03:Online} (OCO) framework for modeling repeated interactions between a player and an arbitrary environment. According to this model, the player iteratively interacts with the environment by selecting at each timestep $t$ a strategy $\vec x^t$ from a convex, compact set $\mc Q\subseteq \R^m$ and observes a linear utility function chosen (possibly adversarially) by the environment. These linear utility functions will be represented as $u^t(x) = \ip{\vu^t, \vx}$, where $\vec u^t \in [-1,+1]^m$. Considering an interaction repeated over $T$ steps, the objective of the player is to minimize the \emph{cumulative regret}, which is defined as
\[
    R^T_{\mc Q} \coloneqq \max_{\vec x \in\mc Q} \sum_{t=1}^T  \ip{\vu^t, \vx - \vx^t}
\]
We say that an algorithm is a \emph{regret minimizer} if it guarantees that the cumulative regret grows sublinearly with $T$, that is, $R^T=o(T)$.

Regret minimization and equilibrium computation in two-player zero-sum games are tightly connected. In particular, consider the scenario of a repeated two-player zero-sum game in which \pmax and \pmin sequentially interact for $T$ timesteps. Both players use regret minimizers over their strategy sets $\X$ and $\Y$ respectively. Let, at each timestep $t$, the utility vectors observed by the two players be $\vec u^t_\pmax = \vec U \vec y^t$ for \pmax and $\vec u^t_\pmin = - \vec U\vec x^t$ for \pmin. \changed{Then, it follows directly by summing the regrets and dividing by $T$ that the equilibrium gap computed w.r.t. the average strategies $\bar{\vec x} :=\sum_{t=1}^T\vec x^t/T$, $\bar{\vec y} := \sum_{t=1}^T\vec y^t/T$ is equal to the total regret:}
\[
    0 \leq \max_{\vec x\in\X} \vec x^\top \vec U \bar{\vec y} - \min_{\vec y\in\Y} \bar{\vec x}^\top \vec U \vec y = \frac{R^T_\pmax + R^T_\pmin}{T}.
\]
Thus, if the two players pick their strategies according to a no-regret algorithm, this property guarantees the convergence of the average strategies $\bar{\vec x}$ and $\bar{\vec y}$ to a Nash Equilibrium as $T\to\infty$. 

We now show how to construct regret minimizers for arbitrary tree-form decision problems (and thus, in particular, also for perfect-recall zero-sum games). For this problem, we use the {\em counterfactual regret minimization} (\ref{al:cfr})  framework \cite{Zinkevich07:Regret,Farina19:Regret}. Intuitively, \ref{al:cfr} allows one to {\em build} a regret minimizer on a tree-form strategy set $\mc Q$ by running {\em local} regret minimizers at each decision point and combining them in a clever way. The guarantee given by \ref{al:cfr} can be expressed as follows. Call a subset $P \subseteq \mc D$ {\em playable} if there is a pure strategy that reaches every decision point in $P$, that is, there is a pure strategy $\vec x \in \mc Q$ such that $\vec x[p_s]=1$ for every $s \in P$. 

\begin{algorithm}[t]\namealg{CFR}
\caption{Counterfactual regret minimization on tree-form decision problems $\mc Q$. For each decision point $s$, $\mc R_s$ is a regret minimizer on $\Delta(\A_s)$.}\label{al:cfr}
\begin{algorithmic}[1]
\Procedure{NextStrategy}{}
    \State $\vec x^t[\Root] \gets 1$ 
    \For{each decision point $s$, in top-down order} 
        \State $\vec r_s^t \gets \mc R_s$.{\sc NextStrategy()}
        \State $\vec x^t[s*] \gets \vec x^t[p_s] \vec r_s^t$
    \EndFor
    \State \Return{$\vec x^t$}
\EndProcedure
\Procedure{ObserveUtility}{$\vec u^t$}
    \State $\vec v^t \gets \vec u^t$
    \For{each decision point $s$, in bottom-up order} 
        \State $\mc R_s$.{\sc ObserveUtility}($\vec v^t[s*]$) 
        \State $\vec v^t[p_s] \gets \vec v^t[p_s] + \ip{\vec r_s^t, \vec v^t[s*]}$
    \EndFor
    \State $t \gets t + 1$
\EndProcedure
\end{algorithmic}
\end{algorithm}
\begin{theorem}[\cite{Zinkevich07:Regret, Farina19:Regret}]\label{th:cfr}
    After the execution of \ref{al:cfr} for $T$ timesteps on the mixed strategy set $\mc Q$, the cumulative regret is at most $\max_P \sum_{s \in P} [R_s^T]_+$, where $R_s^T$ is the regret of the local regret minimizer $\mc R_s$ and the maximum is taken over all playable sets $P$.
\end{theorem}
Local regret minimizers on simplices $\Delta(m)$ the regret of which scales as $O(\sqrt{mT})$ or even $O(\sqrt{T \log m})$ are well known; two examples are {\em regret matching}~\cite{Hart00:Simple}, which has the former guarantee, and {\em multiplicative weights}, which has the latter guarantee. Picking either of these algorithms (or, indeed, any of many regret minimizers with similar guarantees), we have:
\begin{corollary}
    If both players in an extensive-form zero-sum game with perfect recall use \ref{al:cfr} with any efficient local regret minimizer, then after $T$ timesteps, the average strategy profile $(\bar{\vec x}, \bar{\vec y})$ will be an $O(|\Sigma_\pmax| + |\Sigma_\pmin|)/\sqrt{T}$-approximate Nash equilibrium.
\end{corollary}

We refer the interested reader to \cite{Brown19:Solving, Farina21:Faster} for further reading on practical state-of-the-art no-regret-based solvers for two-player zero-sum games.

\subsection{Equivalence across games}

The contributions presented in this paper will rely on auxiliary games to represent strategy optimization problems.
In order for the results obtained in the auxiliary game to map to the original one we want to solve, we need to define what it means for two games to be equivalent. Let $G = (\N, \H, \Z, \A, \I, \vec p, u)$ and $G' = (\N, \H', \Z', \A', \I', \vec p', u')$ be extensive-form games with the same set of players. Let $\Pi_i$ and $\Pi_i'$ be player $i$'s pure strategy set in $G$ and $G'$, respectively, and similarly let $u_i$ and $u_i'$ be player $i$'s utility function in $G$ and $G'$, respectively.

\begin{definition}
    [Strategic Equivalence]\label{def:strategic_equivalence}
    Two games $G$ and $G'$ are {\em strategically equivalent} if there are bijective {\em strategy maps} $\rho_i : \Pi_i \to \Pi_i'$ for each player $i$ such that, for every profile $\vec\pi \in \Pi$ and every player $i$ we have $u_i(\vec\pi) = u'_i(\rho(\vec\pi))$ where $\rho(\vec\pi) := (\rho_i(\vec\pi_i))_{i\in \N \setminus 0}$.
\end{definition}
This definition is a very strong notion of equivalence: if two extensive-form games are equivalent in the above sense, then every strategy $\vec\pi_i$ in one game is equivalent to some strategy $\rho(\vec\pi_i)$ in the other game. Thus, in particular, a solution to one game will give a solution to the other game. 

\subsection{Adversarial team games}\label{sec:adversarial_team_games}
The general framework of adversarial team games has first been studied by \citet{vonStengel97:Team} in the context of normal-form games, while \citet{Celli18:Computational} first addressed them in an extensive-form setting.
Adversarial team games describe situations where multiple agents are organized in two teams, receiving zero-sum payoffs. Like the above papers, ours 
focuses on the setting in which no extra communication channel is available to the players during the game, but they are allowed to communicate freely before the start of the game. This means that the only form of coordination across players' strategies available is \emph{preplay coordination}; that is, any coordination has to be prepared before the start of the game.

Adversarial team games can be modeled as extensive-form games as follows.
\begin{definition}
    [Adversarial team game]
    An extensive-form, perfect-recall game $(\N, \H, \Z, \A, \I, \vec p, u)$ is said to be an \emph{adversarial team game} (ATG), or \emph{two-team zero-sum game} if
    \begin{itemize}[noitemsep,topsep=0pt,parsep=0pt,partopsep=0pt]
        \item the player set is partitioned in two sets called \emph{teams}, symbolized by \tmax and \tmin. Formally, $\N = \tmax \cup \tmin \cup \{\nature\}$, and
        \item the utilities of the players belonging to the same team are identical, and the total utilities of the two team are opposites. Formally,
        \begin{gather}
            u_i = u_j \quad \text{for all $i,j \in \tmax$} \\
            u_i = u_j \quad \text{for all $i,j \in \tmin$} \\
            \sum_{i \in \tmax} u_i = -\sum_{j \in \tmin}u_j.
        \end{gather}
    \end{itemize}
\end{definition}
In adversarial team games, the Nash equilibrium fails to consider that teams can coordinate among themselves. Indeed, it is possible for there to be a Nash equilibrium in which two teammates could profit by {\em jointly} switching strategies, but no {\em individual} player can profit from a unilateral deviation. To consider these joint deviations, it is most natural to reformulate an adversarial team game as a two-player zero-sum game of imperfect recall, in which a {\em team coordinator} plays on behalf of all team members. In this manner, deviations of the team coordinator correspond to {\em simultaneous, joint} deviations of all team members. We now formalize this conversion.
\begin{definition}
    [Coordinator game]\label{def:coordinator_game}
    Let $G = (\N, \H, \Z, \A, \I, \vec p, u)$ be an adversarial team game. The {\em coordinator game} $G'$ corresponding to $G$ is the two-player zero-sum imperfect-recall game defined $G' = (\{ \nature, \pmax, \pmin\}, \H, \Z, \A, \I', \vec p, u')$, where
    \begin{gather}
        \I'_\pmax = \bigcup_{i \in \tmax} \I_i, \qq{} \I'_\pmin = \bigcup_{i \in \tmin} \I_i, \qq{} u'_\pmax = \sum_{i \in \tmax} u_i, \qq{and} u'_\pmin = \sum_{i \in \tmin} u_i.
    \end{gather}
\end{definition}
The coordinator game merges all members of a team (\tmax or \tmin) into a coordinator (\pmax or \pmin). Therefore:
\begin{itemize}
    \item Pure strategies of a coordinator correspond to {\em pure profiles} of the team.
    \item Behavioral strategies of a coordinator correspond to {\em behavioral profiles} of the members of the team. Since behavioral strategies enforce actions at different infosets to be independently sampled, this means that team members can {\em privately} sample randomness for their own personal use but cannot share that randomness with teammates.
    \item Mixed strategies of a coordinator correspond to {\em correlated} strategy profiles of the members of the team. In a correlated profile, team members may {\em jointly} sample randomness that they use to correlate their actions.
\end{itemize}
We remark on the role that preplay coordination has in allowing the coordination capabilities modeled by the coordinator game. In fact, before starting the game, players are allowed to jointly sample a pure profile from their coordinator's mixed strategy and then individually play the specified actions at the infosets in which they play. This allows the team to play any randomized strategy of the coordinator effectively.

The coordinator game allows us to define notions of equilibrium specialized for team games:
\begin{definition}\label{def:tmecor}
    A {\em team max-min equilibrium with correlation}  (TMECor) of an ATG $G$ is a mixed-strategy Nash equilibrium of $G'$.
\end{definition}
\begin{definition}\label{def:tme}
    A {\em team max-min equilibrium}  (TME) of an ATG $G$ is a behavioral max-min strategy of $G'$.
\end{definition}
The {\em TMECor value} and {\em TME value} are defined analogous to the Nash value and behavioral max-min value. As discussed before, behavioral max-min strategies in $G'$ are not equilibria in $G'$, so one may wonder about the name ``team max-min equilibrium''. However, there {\em is} a sense in which TMEs are equilibria: \citet{vonStengel97:Team} showed that, at least in the case where $|\tmin| = 1$, the TMEs are precisely the Nash equilibria of the team in $G$ that maximize the utility of team $\tmax$. 

\changed{
\begin{figure*}[t]
\centering
\begin{forest}
[,nat
  [1,p1,name=1L
    [2,p1,name=2A,action={$\ell$}
        [3,p2,name=3A
            [\util0{+1},terminal]
            [\util0{--1},terminal]
        ]
        [\util0{--1},terminal]
    ]
    [2,p1,parent=1R,name=2B,action={L}
        [\util0{--1},terminal]
        [3,p2,name=3B
            [\util0{--1},terminal]
            [\util0{+1},terminal]
        ]
    ]
  ]
  [1,p1,name=1R
    [2,p1,parent=1L,name=2C,action={r}
        [3,p2,name=3C
            [\util0{+1},terminal]
            [\util0{--1},terminal]
        ]
        [\util0{--1},terminal]
    ]
    [2,p1,name=2D,action={R}
        [\util0{--1},terminal]
        [3,p2,name=3D
            [\util0{--1},terminal]
            [\util0{+1},terminal]
        ]
    ]
  ]
]
\draw[infoset1] (2A) to (2B);
\draw[infoset1] (2C) to (2D);
\draw[infoset2, bend right=30] (3A) to (3B);
\draw[infoset2, bend right=30] (3C) to (3D);
\end{forest}
\caption{
An example of an adversarial team game. There are three players: (1) and (2) are on team \tmax, and (3) is on team \tmin. Dotted lines connect nodes in the same information set. The (total) utility of \tmax is listed on each terminal node. The root node is a chance node, at which Nature selects uniformly at random. Actions at the first level nodes are labeled for easier reference.
}\label{fig:example-atg}
\end{figure*}
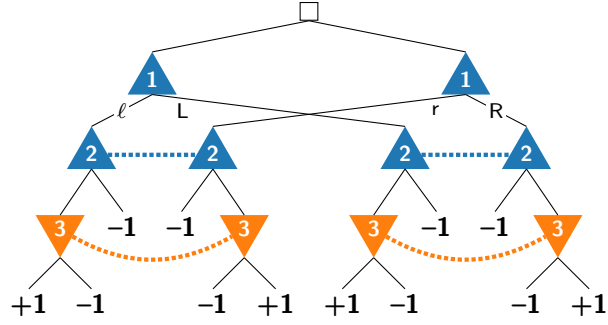

}

An example adversarial team game in which the difference between TME and TMECor is relevant can be found in \Cref{fig:example-atg}. The coordinator game is constructed simply by erasing the player labels, creating a two-player zero-sum game.  \changed{The coordinator game is indeed an imperfect-recall game: at the P2-infosets, the coordinator forgets the Nature bit, which they knew at the P1-infosets.} This game is a simple signaling game: Nature selects a bit, which is privately revealed to P1. P1 then communicates a single bit, which is publicly revealed. Then P2 and P3 both attempt to guess Nature's selected bit, and \tmax wins if and only if P2's guess is correct \changed{and P3's is wrong.} Therefore, the goal of P1 and P2 is for P1 to ``securely'' communicate the bit to P2 without also revealing it to P3. With a behavioral profile, this is impossible since P1 and P2 cannot correlate their strategies; therefore, the TME value is $-1/2$. However, if P1 and P2 are allowed to correlate their strategies, they can do the following: jointly flip a coin. If that coin landed heads, P1 communicates the true bit, and P2 plays what P1 communicates. If that coin landed tails, P1 communicates the opposite of the true bit, and P2 plays the opposite of what P1 communicates. In this way, P2 will always play the true bit, but P3 (who does not know the outcome of the correlating coin flip) does not learn any information. Therefore, the value of this strategy for \tmax is $0$ (since P3 wins half the time by randomly guessing the bit).

\changed{We have defined equilibrium concepts for team games by constructing a coordinator game that is a two-player zero-sum imperfect-recall game. It turns out that the ``converse'' also holds: 
\begin{proposition}
    Every two-player zero-sum imperfect-recall game $G'$ has an ATG $G$ whose coordinator game is $G'$.
\end{proposition}
\begin{proof}
    Consider the ATG $G$ constructed by starting with $G'$, dividing each player's utility by the number of information sets of that player, and then assigning every information set to a different player. By construction, the coordinator game of $G$ is $G'$.
\end{proof}
}
 Therefore, team games and imperfect-recall games are equivalent in a very strong sense. All of the results of this paper, unless otherwise stated, therefore apply equally to team games and to two-player zero-sum imperfect-recall games. 

In the remainder of the paper, we consider the point of view of two-player zero-sum games with imperfect recall. A summary of the different equivalent terms that are used in the two settings can be found in \Cref{tab:terms_mapping}.

\begin{table}
    \centering
    \definecolor{newcolor}{RGB}{192,255,192}
\begin{tabular}{c|c}
\textbf{Adversarial Team Games} & \textbf{Imperfect-Recall Games}
\\
\hline
Team \tmax \tmin & Player \pmax \pmin
\\
Correlated team strategy & Mixed strategy
\\
Uncorrelated team strategy & Behavioral strategy
\\
TMECor & Mixed-strategy Nash Equilibrium
\\
TME & Behavioral max-min strategy
\end{tabular}

\captionof{table}{Translation table between terms commonly employed in the adversarial team games and two-player imperfect-recall games. The translation happens through the introduction of coordinator games (\Cref{def:coordinator_game}).}
\label{tab:terms_mapping}
\end{table}

\section{Related research}\label{sec:related_works}

This section describes the most important related works on computing Nash equilibria in imperfect-recall games and team games.
We start with a review of the results in the literature on imperfect-recall games.
To the best of our knowledge, there are no previous works concerned specifically with efficient algorithms for computing mixed-strategy Nash equilibria in imperfect-recall and timeable games, which is the main objective of this paper. A similar approach to ours is taken by \citet{Tewolde23:Computational}, investigating imperfect-recall games without the timeability assumption; that is, they include the absent-minded case; however, their solution concept is weaker than ours, and their setting is more difficult. %

There is a stream of works studying the computational complexity of finding Nash equilibria in imperfect-recall games.
\citet{Koller92:Complexity} study the NP-hardness of computing mixed Nash equilibria in imperfect-recall games. In particular, they proved that computing Nash equilibria in imperfect-recall games is NP-hard and that computing a max-min equilibrium in pure strategies is $\Sigma_2^\P$-complete.
Our results from \Cref{sec:complexity} refine those results proving $\Delta_2^\P$-completeness and $\Sigma_2^\P$-completeness of computing Nash equilibria in mixed or behavioral strategies, respectively.

\citet{Gimbert20:Bridge} study imperfect-recall games with or without absentmindedness and one or two players. In particular, they connect the problem of computing optimal values in behavioral strategies with polynomial optimization. They also make use of their results on imperfect-recall games to characterize the complexity of the bidding phase of the \emph{team} game of Bridge, acknowledging the equivalence of team games and imperfect-recall games.
Our complexity results from \Cref{sec:complexity} focus on the computation of mixed-strategy Nash equilibria in two-player imperfect recall without absentmindedness and can be thought of as a refined (more powerful) version of the results of \citet{Gimbert20:Bridge}.

Similarly, previous uses of imperfect recall for abstractions in two-player zero-sum games make use of behavioral strategies with the aim of computing an approximate Nash equilibrium in the original perfect-recall game \cite{Waugh09:Abstraction,Lanctot12:No,Kroer16:Imperfect,Cermak18:Approximating}. In this case, the imperfect-recall game is used as a compressed representation to compute behavioral strategies that are then mapped to a hopefully good strategy in the original game. 
Our work is distinguished from theirs as our approach targets imperfect-recall games per se, and we are not interested in mapping optimal strategies to a perfect-recall game from which the game is generated.
Some works, in particular, are interested in A-loss imperfect-recall games. The A-loss condition was originally introduced by \citet{Kaneko95:Behavior} and is related to imperfect recall due to a player forgetting their own actions played in the past. 
Our results generalize theirs in the sense that, since A-loss games are those that inflate to games with perfect recall, \Cref{pr:inflate1} guarantees that our algorithm will run in polynomial time for games with A-loss recall.

On the other hand, the adversarial team game literature presents many works interested in finding TMECors (\Cref{def:tmecor}), corresponding to mixed-strategy Nash equilibria in imperfect-recall games.
Past works typically overcome the issue of coordination among team members by directly optimizing over the space of coordinated strategies, that is, the space of mixed strategies of the coordinator game.
The algorithm commonly adopted is that of \emph{column generation (CG)}, which for games essentially amounts to the \emph{double oracle} \cite{McMahan03:Planning} algorithm. CG is an iterative approach that considers only a subset of pure strategies for the team at a time and progressively expands it in order to best approximate the support of the TMECor. Each CG iteration is composed of two steps: firstly, a TMECor restricted to the current set of pure strategies is computed; secondly, a joint best response to the current equilibrium is computed for each team, and these strategies are added to the sets of available strategies.
Different variants of this algorithm have been developed in the literature. \citet{Celli18:Computational} first defined TMECors in the extensive-form games setting and proposed a CG algorithm based on integer linear programming for the best response computation and linear programming for the equilibrium computation.
Successive works \citep{Farina18:ExAnte,Zhang21:Computing,Farina21:Faster,Zhang22:Optimal} further refined the results in terms of execution speed and scalability while still operating in the CG framework. 
Similarly, \citet{Zhang20:Converging,Zhang21:Computing} focused on the computation of TMEs (\Cref{def:tme}), which are behavioral-strategy equilibria for the settings in which no coordination signals are shared among the team.

Contrary to the direction followed by previous works, the purpose of this paper is to enable a sequence-form-like description of the strategy space of each team, thus bridging the gap in the literature between two-player zero-sum games and adversarial team games and empirically evaluating the tradeoff imposed by the solution proposed.

A technique based on a belief-prescription scheme was originally proposed by \citet{Nayyar13:Decentralized} in the field of decentralized stochastic control in a cooperative setting.
The term \emph{prescriptions} was later proposed by \citet{Foerster19:Bayesian} in the setting of cooperative reinforcement learning to indicate the meta-actions played when coordinating multiple agents.
Our construction of the auxiliary game in \Cref{sec:auxgame} can be seen as a generalization of those approaches to a two-team adversarial setting. In addition, the adoption of observations (defined in \Cref{sec:beliefs_obs}) in place of public states and the DAG approach from \Cref{sec:tbdag} are novel contributions that refine the efficiency of our methods compared to theirs.
A more detailed comparison between observations and public states is given in \Cref{sec:discussion:obs}.

\changed{Lastly, \citet{Lanctot12:No} introduced the idea of a {\em well-formed game} of imperfect recall, which in essence is a game for which CFR will return the same sequence of iterates on the imperfect-recall game as it would on the perfect-recall refinement; therefore CFR can be safely applied to the imperfect-recall game to find an equilibrium. However, not all imperfect-recall games are well formed; as such, their technique is restricted to a small class of games, whereas ours applies to an arbitrary imperfect-recall game. Indeed, this can be clearly seen because well-formed games can be solved in polynomial time, whereas general imperfect-recall games, even against one player, are NP-hard \cite{Koller92:Complexity}. Nonetheless, there is an interesting connection between well-formed games and our techniques: 
the idea that infosets on ``identical'' paths of play can be merged while maintaining equivalence of the game is the fundamental idea behind well-formed games.
We take the idea one step further, compressing all {\em paths} resulting from identical infosets to create a DAG structure, and giving an algorithm that creates this DAG structure without even building the entire tree beforehand. 
In other words, the DAG-form decision problem that we construct can be thought of as a well-formed abstraction of the belief game for each player. The main advantage of our approach is that we directly construct this DAG-form decision problem from the original (much smaller) team game, without needing the intermediate step of creating the belief game, which can be much larger than the DAG-form decision problems in both theory and practice.} %

\section{Mixed Nash computation through beliefs and observations}

We introduce the paper's fundamental contribution: a novel technique for computing a mixed Nash equilibrium in two-player zero-sum imperfect-recall games (or, equivalently, for computing a TMECor in adversarial team games) based on the construction of an equivalent two-player zero-sum game with perfect recall.

Our technique attains the perfect recall condition by suitably changing the information available to the players and their action sets.
The main intuition behind the belief game is to consider the point of view of a perfect recall player in place of the imperfect recall one. Differently from the imperfect-recall player, this player reasons only using information the player would never forget due to imperfect recall and chooses an action for every possible information set the imperfect-recall player may be in. The game then transitions by applying the action corresponding to the information set of the current node. Crucially, the perfect-recall player can strategically refine the set of reached nodes over time by carefully considering reachable nodes given the played strategy and the perfect-recall results of their actions.

After introducing the main concepts and the construction algorithm, we prove that the original and the belief games are strategically equivalent. This means that the perfect-recall player we introduce is an equivalent representation of both the imperfect-recall player and the corresponding preplay coordinated team (thanks to the considerations from \Cref{sec:adversarial_team_games}).

\subsection{Beliefs and Observations}\label{sec:beliefs_obs}
The main purpose of this section is to formally define \emph{beliefs}, which are sets of nodes $B \subseteq \H$ derived from information sets $\I_i$ of player $i \in \{\pmax, \pmin\}$. Informally, beliefs are the ``information sets'' that $i$ would have if they could not distinguish nodes that cannot be distinguished using information from a later stage. This notion is formalized by putting in the same belief any two nodes that have descendent nodes in the same information set (even if they belong to different information sets).\footnote{We remark that this is a condition of imperfect recall equivalent to \Cref{def:perfect_recall}, because information sets $I,I'$ at the same level of the game having descendent nodes in the same information set $J$ implies that $J$ has nodes with different sequences.}
Similarly to information sets:
\begin{enumerate}[i)]
    \item nodes in the same belief would be indistinguishable to $i$,
    \item one action is chosen at each belief, and then this action is followed in all nodes in the belief, and
    \item \changed{if the player knows that the current node of the game $h$ lies in a set $H$, and the player observes that their current belief is $B$, then the player knows that the current node lies in $B \cap H$ (that is, beliefs correspond to observations over the state of the game).}
\end{enumerate}
Crucially, beliefs can be organized in the tree-like structure needed by algorithms finding Nash equilibria in two-player zero-sum games, as we will see in \Cref{sec:auxgame,sec:regret-minimization-team-games}. This is thanks to the guarantee that once a group of nodes is split among two different beliefs, then any group of descendent nodes from one belief will be distinguishable from any group of descendants from the other.

In the following, we formalize the notion of beliefs and observations. We consider a two-player zero-sum game with imperfect recall $G$.

\changed{
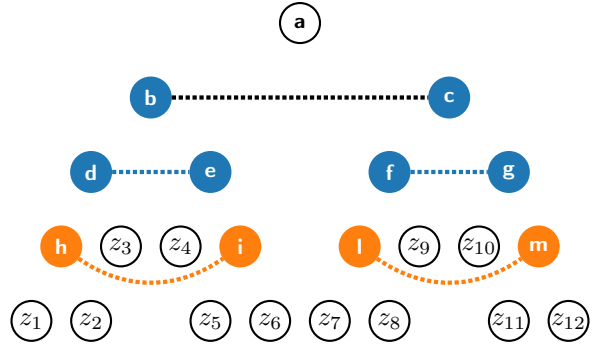
\begin{figure*}[t]
\centering
\begin{forest}
p1/.style={text=white, fill=p1color, minimum size=15pt, inner sep=0pt, draw=p1color, font=\bfseries\sffamily\footnotesize},
    terminal/.style={draw=black, minimum size=15pt, inner sep=0pt},
    for tree={no edge, circle, fill=white, thick, anchor=center}
[a,nat
  [b,p1,name=1L
    [d,p1,name=2A
        [h,p2,name=3A
            [$z_1$,terminal]
            [$z_2$,terminal]
        ]
        [$z_3$,terminal]
    ]
    [e,p1,name=2B
        [$z_4$,terminal]
        [i,p2,name=3B
            [$z_5$,terminal]
            [$z_6$,terminal]
        ]
    ]
  ]
  [c,p1,name=1R
    [f,p1,name=2C
        [l,p2,name=3C
            [$z_7$,terminal]
            [$z_8$,terminal]
        ]
        [$z_9$,terminal]
    ]
    [g,p1,name=2D
        [$z_{10}$,terminal]
        [m,p2,name=3D
            [$z_{11}$,terminal]
            [$z_{12}$,terminal]
        ]
    ]
  ]
]
\draw[infoset1, draw=black] (1L) to (1R);
\draw[infoset1] (2A) to (2B);
\draw[infoset1] (2C) to (2D);
\draw[infoset2, bend right=35] (3A) to (3B);
\draw[infoset2, bend right=35] (3C) to (3D);
\end{forest}
\caption{
The connectivity graph for player \pmax and \pmin in the example game from \Cref{fig:example-atg}. The nodes in the two figures correspond in position, and letter identifiers are added to each node for easier reference.
We highlight using the black color the connection between nodes $b$ and $c$, which is due to the infoset connecting $d$ and $e$.
}\label{fig:example-connectivity}
\end{figure*}

}

\paragraph{Connectivity graph} We say that two nodes $h$ and $h'$ are \emph{unforgettably distinguished} by $i$ if they do not belong to the same infoset and no pair of descendants of those two nodes belong to the same infoset, that is, $i$ will never be in an information set where these two ancestors are both possible.
This condition guarantees that if the set of candidates is $H = \{h,h'\}$, then the player can discern $h$ from $h'$ and will never forget which of the two nodes has been reached in the next steps of the game.%
\footnote{In particular, any team member in the corresponding team game can always recall, after reaching $h$ or $h'$, whether $h$ or $h'$ has been reached.}

For the purpose of our definitions, we are concerned with pairs of nodes that are \emph{not} distinguishable. This can be represented through a \emph{connectivity graph} over \H as follows.
\begin{definition}
    [Connectivity graph]
    The \emph{connectivity graph} $\mathcal G_i = (\H, \mc E_i)$ for player $i\in\{\pmax,\pmin\}$ is the graph with nodes $\H$ and edges $\mc E_i$, where $(h, h')\in\mc E_i$ if $h$ and $h'$ are at the same depth in $G$ and there exists $I \in \I_i$ such that $h\preceq I$ and $h'\preceq I$.
\end{definition}

\changed{\Cref{fig:example-connectivity} is an example of a connectivity graph in adversarial team games. Blue edges connect nodes of player \pmax at the same depth and in the same infoset, while black edges connect nodes preceding an identical infoset at deeper depth. For example, the black edge $b$-$c$ in \Cref{fig:example-connectivity} is due to $d$ and $e$ belonging to the same infoset. Note that in \Cref{fig:example-atg}, $d$ is a child of $b$ and $e$ is a child of $c$.}

\paragraph{Beliefs} Consider now a set $H$ of nodes such that the induced subgraph $\mathcal G_i[H]$ is connected. Player $i$ cannot distinguish any subset of $H$ from the others because no node can be distinguished from its neighbors.
\emph{Beliefs} are defined as these sets of indistinguishable nodes.%

\begin{definition}
    [Belief]\label{def:belief}
    A set of nodes $B \subseteq \H$ is a \emph{belief} for player $i$ if the induced subgraph $\mathcal G_i[B]$ is connected.
\end{definition}

Since only nodes of the same depth can be connected via edges in $\mathcal G_i$, all nodes in the same belief have the same depth.
A direct consequence of the definition of beliefs is that $\{\Root\}$ and $\{z\}$ for $z\in\Z$ are singleton beliefs for both players.

\paragraph{Observations} Consider a set $H$ of nodes such that the induced subgraph $\mathcal G_i[H]$ has multiple connected components. In this case, player $i$ can distinguish those components one from the other, thus partitioning $H$ into multiple beliefs. Intuitively, the unforgettable information is enough to distinguish every node in a component from any node in other components. The player can, therefore, exclude nodes from components that are distinguishable from the current reached node.
We say that upon reaching a node $h$ among possible candidates $H$, player $i$ \emph{observes belief $B \subseteq H$}, meaning that player $i$ uses the newly acquired unforgettable information acquired in $h$ to refine their imperfect information from $H$ to $B$.
We formalize this notion of observation through the function \splitb{i}:\footnote{%
In team games, the belief returned by \splitb{i} is the team-common knowledge update happening when reaching $h$ among a set of candidates $H$. \changed{That is, upon \splitb{i} returning a belief $B$, it is team-common knowledge that the game reached a node in $B$.}}
\begin{definition}[Observation]\label{def:splitb}
    The \emph{observation} for player $i \in \{\pmax, \pmin\}$ when reaching node $h \in H$ of candidate nodes is:
    \[
    \splitb{i}(H, h) \coloneqq \text{the connected component of $\mc G_i[H]$ containing $h$}.
    \]
    \changed{The set of all possible observations at $H$ is thus the set of all connected components of $\mc G_i[H]$, and will be denoted by $\mc B_H$.}\footnote{%
    We remark that the belief-based constructions employed by the paper would also work when
    considering a coarser set of beliefs, that is, pretending that more edges existed in $\mc G_i[H]$. For example, to align with the framework of {\em factored-observation games}~\cite{Kovavrik22:Rethinking}, one could define $\splitb{i}$ using the explicitly-given public observations. However, since the efficiency of the proposed algorithms depends on the size of the beliefs employed, we opt not to allow the use of beliefs larger than needed.  As we show in \Cref{sec:worst-case-size-aux}, any reduction in the size of the beliefs in a game brings exponential benefits in the size of the belief game obtained.}
\end{definition}

\changed{
An example of observations can be given by considering the team game depicted in \Cref{fig:example-connectivity} and a candidate set $H=\{d,e,f\}$, which is possible when player 1 plays a mixed strategy that excludes $g$ from being reached. 
What team-public information can player 2 use to distinguish the current node of the game? Player 2 knows that the current node $h$ in the game is in $H$ because they know the strategy of player 1. However, player 2 also observes their current information set $I = \{d,e\}$ if $h \in \{d,e\}$, or $J =\{f,g\}$ if $h \in \{f,g\}$. This observation can be used to split the team-public belief $H$ since there is no future infoset of the team containing both descendants of $\{d,e\}$ and $\{f\}$; indeed, there are no future infosets for \pmax at all in this game. Thus, $d$ and $e$ are unforgettably distinguished from $f$ by \pmax.

The reasoning from this example is generalized by \Cref{def:splitb}: we have 
\begin{gather}
    \splitb{\pmax}(H,d) = \splitb{\pmax}(H,e) = \{d,e\} \\ \qq{and} \splitb{\pmax}(H,f) = \{f\}.
\end{gather}}%
\paragraph{Team public states} We compare our notion of beliefs with \emph{public states}, an alternative customarily used in the related literature. A public state $P$ for player $i$ is a connected component of the connectivity graph $\mathcal G_i$. The set of all public states of $i$ is denoted by $\mathcal P_i$. 

Public states identify sets of nodes that are distinguishable to a player without considering a possibly pruned subgraph of $\mathcal G_i$ as instead done for team observations. Therefore, every belief is contained in a public state. 
In \Cref{fig:example-connectivity} we have that $\mathcal P_\pmax = \{\{a\}, \{b,c\}, \{d,e\}, \{f,g\}, \{h\}, \{i\}, \{l\}, \{m\}\} \cup \{\{z\} : z\in\Z\}$.

Public states are the customarily adopted alternative to observations when partitioning a set $H$ of candidates in beliefs by splitting $H$ in $\{H \cap P : P \in \mc P_i\}$. However, public states may return a coarser partition than the one returned by observations, as the absence of specific nodes from $H$ may disconnect components in $\mc G$. We will, therefore, use observations in place of public states whenever possible.
An example illustrating the difference between the two definitions is available in \Cref{sec:discussion:obs}.

\paragraph{Prescriptions} Restricting the information available to player $i$ to their beliefs also affects the set of actions available. In fact, multiple infosets may intersect a given belief. \changed{In the belief game, the player only knows the current belief, and in particular does not know which infoset within the current belief is the true infoset. Therefore, the player cannot just specify a single action at the true infoset.}

We overcome this issue by associating to each belief $B$ a set of meta-actions $A_B^i$ such that an action is specified for every infoset that intersects the belief. We call such structured meta-actions \emph{prescriptions} and use a symbol $\vec a$ to indicate them. 
The concept is formally defined as follows.
\begin{definition}
    [Prescription]\label{def:prescription}
    Consider a belief $B$ of a player $i\in\{\pmax,\pmin\}$. A \emph{prescription} $\vec a$ is a selection of one action at each infoset having a nonempty intersection with $B$:
    \[
        \vec a \in \bigtimes_{I \in \I_i[B]} A_I \qq{where} \mc I_i[B] = \{ {I \in \I_i : I \cap B \not = \emptyset}\}.
    \]
\end{definition}

Given a prescription $\vec a $ for a belief $B$ and an infoset $I$ such that $I\cap B\neq \varnothing$, we denote as $\vec a[I]$ the action relative to infoset $I$ which is specified by prescription $\vec a$.
Note that we have \emph{empty prescriptions} at beliefs containing no active nodes for a player.

As we will see in the next section, our equivalent belief game introduces one perfect-recall player per team, with information sets associated with beliefs corresponding to the perfect-recall part of the information available to this unique player. Prescriptions will allow this player to have an identical expressive power in terms of actions without accessing the exact information set of the player, which is their imperfect-recall information. 
Moreover, specifying a prescription at each reached belief for $i$ incrementally defines a pure strategy %
of player $i$. This allows us to consider a reduced set of candidate nodes $H$ for the reached node $h$ from which the belief is observed, as a non-played action implies that all the descendant nodes are not reached and, therefore, excluded from the candidates.

\changed{We now give a complete example of the concepts introduced so far.} Consider a 3-player poker instance where two players collude to form a team.
At any time of the game, we can consider the point of view of a team coordinator, who acts as the single imperfect recall player. We can imagine this coordinator as sitting at the same table as the players, and therefore, they cannot access the private cards given to the players but can access the same public information as the players, that is, the bet, fold, and check actions of the players. Their belief at any point regards the private cards that each team member has. At the start of the game, this belief is uniform over all pairs of cards, as no information regarding these cards is available from an external point of view. The coordinator emits prescriptions for the players to follow as the game progresses. Since the coordinator does not know the card held by a player, they have to prescribe an action for each possible card the current player may hold. The player receives this prescription and follows the part of it that matches the private card. By observing the action played by the player, the coordinator can exclude from their belief the cards for which they prescribed different actions.
While there are no means of communicating prescriptions during play at the poker table, this mechanism can be implemented \emph{ex ante}; that is, each team of players jointly samples a pure strategy of this coordinator before the start of the game, and each member simulates the coordinator locally.

\paragraph{Information complexity} We quantify the number of information sets reaching a belief through the notion of \emph{information complexity} $k$. This quantity will allow us to bound the size of the belief games in \Cref{sec:worst-case-size-aux,sec:worst_case_size_tbdag}.

We first characterize the notion of \emph{remembered} information sets and the set of \emph{last-infosets} at a node.
Intuitively, an infoset $J$ remembers another infoset $I$ if reaching a node in $J$ implies traversing a node in $I$ and picking a specific action. Therefore, knowing to be at a node in $J$ allows the player to recall having traversed $I$ and have played action $a$ there.
The last-infosets of player $i$ at $h$ are the information sets traversed by $i$ on the path to $h$ and not remembered by any following information set of the player up to $h$.\footnote{In an adversarial team game, the only possible last-infosets at a node $h$ for a given team $t$ are the most recent infosets of each player in $t$: every other infoset belongs to one of these players, and therefore, by perfect recall, is remembered by the most recent infoset of that player.} This set quantifies the knowledge lost by the player at a node due to imperfect recall. 

\changed{Intuitively, the set of last-infosets at a node $h$ completely characterizes all the information ever presented to the imperfect-recall player: if $I_h$ remembers $J$, then keeping $J$ in the set of last-infosets would be redundant. If $K \ne I_h$ is a last-infoset at $h$, that means that player $i$ at $h$ has forgotten that they have traversed $K$ in the past. For the team setting, the set of last-infosets completely characterizes the common knowledge of the team at node $h$.}

\begin{definition}
    An infoset $J$ {\em remembers} another infoset $I$ if there exists an action $a \in A_{I}$ such that, for every $h \in J$, we have $h' a \preceq h$ for some $h' \in I$.
\end{definition}%

\begin{definition}
    The set of {\em last-infosets} at node $h$ for player $i$ is the set of infosets $I \in \I_i$ such that $I \preceq h$ and there is no other infoset $J \in \I_i$ such that $J \preceq h$ and $J$ remembers $I$.
\end{definition}

We will use $\LI_i(h)$ to denote the set of last-infosets at $h$ for player $i$. Note that if $h \in \H_i$ then $I_h \in \LI_i(h)$.

Now define the {\em information complexity} $k$ of a two-player game $G$ as follows.
    \[
        k = \max_{\substack{i \in \{ \pmax, \pmin \}, \\P \in \mathcal P_i} }\abs{\bigcup_{h \in P} \LI_i(h)}
    \]
Intuitively, $k$ is a representation of {\em how much information can be worst-case forgotten} by player $i$. In the team game interpretation, $k$ is a representation of how {\em asymmetric} the information is among team members: \changed{if $k$ is small, then all team members have approximately the same information; in particular if $k=1$ then all team members have {\em the same} information, that is, the ``imperfect-recall'' team coordinator actually has perfect recall.} 

The information complexity characterizes both the number of beliefs in a public state $P$, and the number of prescriptions that are available at such beliefs. In fact, the actions played at information sets in $\bigcup_{h \in P} \LI_i(h)$ determine which nodes in $P$ are reached (that is, a belief $B \subseteq P$)\changed{; these actions upper-bound the number of worst-case information sets available, which can be used to bound the worst-case amount of prescriptions at a belief. In fact, the number of prescriptions at a public state $P$ is upper-bounded by $(b+1)^{|\bigcup_{h \in P} |\LI_i(h)|} \le (b+1)^k$, where $b$ is the \emph{branching factor}, which is the worst-case number of actions at an information set: for each possible information set, either that information set is not reached, or it is reached and one of the $b$ available actions is played. Specifying one of these $b+1$ possibilities at each of the (at most) $k$ last-infosets uniquely determines the prescription. This analysis will be formalized in the proof of  \Cref{th:tb-dag-size}}.

\changed{Since $\LI_i(h) \subseteq \I_i$, a trivial bound on $k$ is $k \le \max_i |\I_i| \le |\Z| \le b^d$, where $b$ is the branching factor and $d$ is the depth. Later in the paper we will show that our algorithms have exponential dependence on $k$, which in the worst case according to this bound will lead to exponential time in $|\Z|$. This exponential dependence is inevitable: as we will discuss in depth in \Cref{sec:complexity}, there cannot exist efficient algorithms for representing general imperfect-recall decision problems under reasonable complexity assumptions. However, there are many classes of game in which $k \ll \max_i |\I_i|$. As an example, consider four-player Texas hold'em (THE) poker, where the four players are partitioned into two teams of two. Then a public state is defined by the betting history and any public cards. Given a public state, each player only has at most $\binom{52}{2} = 1326$ different possible last-infosets, one for each possible choice of hole cards. Thus, we have $k \le 2 \cdot 1326 = 2652$, which is much smaller than the number of histories.\footnote{For example, in two-player no-limit THE, $|\Z| > 10^{160}$~\cite{Johanson13:Measuring}. This example applies more generally to games in which player actions are public, and the only private information is due to one of $t$ ``private types'' ($t=1326$ for THE) being assigned to each player. In such games, we have $k \le nt$, where $n$ is the size of the larger team.} In some sense (which we will make precise in \Cref{sec:complexity}), information complexity precisely characterizes the hardness of representing the strategy set of a player with imperfect recall.
}

\changed{
As an example, consider the game from \Cref{fig:example-connectivity}, whose public states for \pmax are:
\[
    \mathcal P_\pmax = \{\{a\}, \{b,c\}, \{d,e\}, \{f,g\}, \{h\}, \{i\}, \{l\}, \{m\}\} \cup \{\{z\} : z\in\Z\}.
\]
The information complexity in this case is $k=3$. Indeed, for public state $P = \{d,e\}$ of player \pmax, we have
\[
    \bigcup_{h' \in \{d,e\}} \LI_\pmax(h') = \LI_\pmax(d) \cup \LI_\pmax(e) = \{I_b, I_d\} \cup \{I_c, I_e\} = \{I_b, I_c, I_d\}.
\]
This information complexity can be used to bound the number of possible prescriptions for player \pmax given a public state. At $P = \{d,e\}$, the actions played at the three information sets $I_b, I_c, I_d$ are enough to characterize which belief $B \subseteq P$ is reached during the game and which prescription was played at $B$.
In fact, the action at $I_b$ decides whether $d$ is reached, and the action at $I_c$ decides whether $e$ is reached.
Overall, we bound the number of prescriptions in the public state as $(b+1)^k = (2+1)^3 = 27$.  The bound of $(b+1)^k$ is not tight in this example because all three information sets are always played by $\pmax$---indeed, the true number of prescriptions at $P$ is $b^k = 8$.
}

\subsection{Belief game construction}\label{sec:auxgame}

\begin{algorithm}[!tbp]\namealg{MakeBeliefGame}
\caption{Belief game construction}\label{al:auxgame}
\begin{algorithmic}[1]
\Procedure{MakeBeliefGame}{$G$} 
    \State \algoname{MakeNode}$_\pmax(\Root, \{\Root\}, \{\Root\}, \emptyset, \emptyset)$
\EndProcedure
\Procedure{MakeNode$_\pmax$}{$h, B_\pmax, B_\pmin$, $\tilde \sigma_\pmax$, $\tilde \sigma_\pmin$}
    \State create node $\tilde h \in \tilde\H_\pmax$
    \State add $\tilde h$ to infoset labeled $(\tilde\sigma_\pmax,B_\pmax)$
    \For{each prescription $\vec a_\pmax \in \mc A_{B_\pmax}^{\pmax}$}
        \State $\tilde h \vec a_\pmax \gets \algoname{MakeNode}_\pmin(h, B_\pmax, B_\pmin, \tilde\sigma_\pmax, \tilde\sigma_\pmin,  a_\pmax)$
    \EndFor
    \State \Return $\tilde h$
\EndProcedure
\Procedure{MakeNode$_\pmin$}{$h, B_\pmax, B_\pmin$, $\tilde \sigma_\pmax$, $\tilde \sigma_\pmin$, $\vec a_\pmax$}
    \State create node $\tilde h\vec a_\pmax \in \tilde\H_\pmin$
    \State add $\tilde h\vec a_\pmax$ to infoset labeled $(\tilde\sigma_\pmin,B_\pmin)$
    \For{each prescription $\vec a_\pmin \in \mc A_{B_\pmin}^{\pmin}$}
        \State $\tilde h \vec a_\pmax  \vec a_\pmin \gets \algoname{MakeNode}_\nature(h, B_\pmax, B_\pmin, \tilde\sigma_\pmax, \tilde\sigma_\pmin, \vec a_\pmax, \vec a_\pmin)$
    \EndFor
    \State \Return $\tilde h\vec a_\pmax$
\EndProcedure
\Procedure{$\algoname{MakeNode}_\nature$}{$h, B_\pmax, B_\pmin$, $\tilde\sigma_\pmax$, $\tilde\sigma_\pmin$, $\vec a_\pmax$, $\vec a_\pmin$}
    
    \If{$h$ is terminal node}
        \State create new terminal node $\tilde h\vec a_\pmax \vec a_\pmin \in \tilde\Z$
        \State $u_i(\tilde h\vec a_\pmax \vec a_\pmin) \gets u_i(h)$ for each player $i$
        \State $\vec p[\tilde h\vec a_\pmax \vec a_\pmin] \gets \vec p[h]$
        \State \Return $\tilde h\vec a_\pmax \vec a_\pmin$
    \EndIf
    \State create new chance node $\tilde h \vec a_\pmax  \vec a_\pmin \in \H_\nature$
    \If{$h$ is a chance node} $S \gets \{ ha : a \in A_h \}$ \label{al:line:chance_actions}
    \Else{} $S \gets \{ h\vec a_i[I_h]  \}$ where $h \in \H_i$ \EndIf
    \For{each node $ha \in S$}
        \State $ B_i' \gets \splitb{i}(B_i \vec a_i, ha)$ for each player $i$
        \State $\tilde h\vec a_\pmax \vec a_\pmin a \gets \textsc{MakeNode}_\pmax(ha, B_\pmax', B_\pmin', \tilde \sigma_\pmax + (\tilde\sigma_\pmax, a_\pmax), \tilde \sigma_\pmin + (\tilde\sigma_\pmin, a_\pmin))$
    \EndFor
    \State \Return $\tilde h\vec a_\pmax \vec a_\pmin$
\EndProcedure
\end{algorithmic}
\end{algorithm}

We now introduce an algorithm that explicitly constructs a belief game given any two-player game.
We will use $\tilde G$ to denote the belief game and distinguish components of the original game $G$ from components of the belief game by writing tildes: for example, a generic history is $\tilde h \in \tilde{\mc H}$, a generic information set is $\tilde I_i \in \tilde{\mc I}_i$, and so on.

A node $\tilde h \in \tilde\H$ in the belief game is identified by a tuple $(h, B_\pmax, B_\pmin)$ such that $h \in B_\pmax \cap B_\pmin$, where $h \in \H$ is the corresponding node in the original game describing the underlying state of the game, $B_\pmax, B_\pmin$ are the current beliefs of \pmax and \pmin respectively. At $\tilde h$, each player $i \in \{ \pmax, \pmin \}$ has a (possibly empty) collection of infosets, $\mc I[B_i]$, at which they needs to prescribe an action. The two players simultaneously submit prescriptions $\vec a_i \in A_{B_i}^i$. The next belief game node is $(ha, B'_\pmax, B'_\pmin)$, where the action $a$ and the next beliefs $(B'_\pmax, B'_\pmin)$ are defined as follows.
\begin{enumerate}[label=(\roman*)]
    \item The action $a$ is the one taken by the player at $h$: if $h$ is chance node, then $a$ is sampled from chance's action distribution at $h$; otherwise, $a = \vec a_i[I_h]$.
    \item The beliefs evolve as follows. For each player $i$, the set of candidate next histories in the original game compatible with $i$'s current belief $B_i$ and their prescription $\vec a_i$ is given by:
    {
    \[
        B_i\vec a_i \coloneqq \underbrace{\{ha : h\in B_i \cap \mc H_i,  a= \vec a [I_h]\}}_\text{
    \scalebox{0.6}{\setstretch{1.0}\begin{tabular}{c}when player $i$ acts, \\ it must be according to the prescription\end{tabular}}} \cup \underbrace{\{ha : h\in B_i \setminus \mc H_i, a \in A_h\}}_\text{\scalebox{0.6}{
    \setstretch{1.0}\begin{tabular}{c}when player $i$ does not act, \\ $i$ does not know what action is taken\end{tabular}}} ,
    \]
    }
    Next, player $i$ observes the information revealed by the next history $ha$, thus arriving at belief:
    \begin{align}
        \tilde B'_i := \splitb{i}(B_i \vec a_i, ha).
    \end{align}
\end{enumerate}

We remark some characteristics of $\tilde G \coloneqq \ref{al:auxgame}(G)$. 
\begin{itemize}
\item \changed{The belief game is a game with simultaneous moves: \pmax, \pmin, and chance simultaneously chooses prescriptions and chance actions. \Cref{al:auxgame} describes the procedure that constructs an extensive-form game (without simultaneous moves) that is equivalent to it.\footnote{Following the usual technique for converting games with simultaneous moves into extensive-form games without simultaneous moves, we implement simultaneous actions by representing each step in the game as a sequence of one node per player \pmax, \pmin, \nature where everyone acts; the effects of the actions taken are applied at the end.}}
\item The depth of the belief game is exactly three times that of the original game. That is, each move in the original game becomes three moves in the belief game: one where \pmax chooses their prescription (this node is trivial if it is not \pmax's turn in the original game), one where \pmin chooses their prescription (again, trivial if it is not \pmin's turn), and one node where Nature chooses its action. 
    \item Multiple different tree nodes $\tilde h$ can correspond to the same $(h, B_\pmax, B_\pmin)$ tuple. In particular, for each terminal node $z \in \Z$ there is only one possible $(z, \{z\}, \{z\})$, but yet many different tree nodes can lead to the same $z$. 
    \item Information sets in $\tilde G$ are associated to sequences of beliefs and prescriptions. In particular, such infosets can be described by tuples of the form $(B^1_i=\{\Root\}, \vec a^1_i, B^2_i, \vec a^2_i, \dots, B^L_i)$, where $\vec a^\ell_i  \in A_{B_i^\ell}$ and $B^{\ell+1}_i = \splitb{i}(B_i^\ell \vec a_i^\ell, h)$ for some $h \in B_i^\ell \vec a_i^\ell$.
    \item By construction of \ref{al:auxgame} we have that $\tilde G$ is a perfect-recall game. In fact, nodes with different sequences are associated to different information sets thanks to including sequences in each information set's label;
    \item When $h$ is terminal, the belief game does not stop until both players have observed the trivial belief $\{ h \}$ at $h$ and then submitted their empty prescriptions at that belief. This is for notational convenience: it ensures that terminal sequences for a player $i$ will always end with singleton beliefs, which will make the later analysis cleaner.
    \item Modulo trivial reformulations (namely, the insertion of nodes with a single child), if $G$ is perfect recall, then $\tilde G$ is identical to $G$.
\end{itemize}

Given a pure strategy $\tilde{\vec \pi}_i \in \tilde \Pi_i$, we say that $\tilde{\vec \pi}_i$ plays to a belief $B_i$ of player $i$ if $\tilde{\vec \pi}_i$ plays to some node corresponding to $(h, B_i, B_{-i})$. 

\changed{
\begin{figure*}[t]
\centering
\newcommand{\noded}{[k,p2,[$z_1$][$z_2$]][$z_3$]}
\newcommand{\nodee}{[$z_4$][i,p2,[$z_5$][$z_6$]]}
\newcommand{\nodef}{[l,p2,[$z_7$][$z_8$]][$z_9$]}
\newcommand{\nodeg}{[$z_{10}$][m,p2,[$z_{11}$][$z_{12}$]]}
\adjustbox{center}{
\begin{forest}
for tree={
if n children=0{
  terminal
}{}
}
[,nat
  [b,p1,name=b
    [d,p1,action={$\ell$L},@\noded]
    [d,p1,action={$\ell$R},@\noded]
    [e,p1,action={rL},@\nodee]
    [e,p1,action={rR},@\nodee]
  ]
  [c,p1,name=c
    [f,p1,action={$\ell$L},@\nodef]
    [g,p1,action={$\ell$R},@\nodeg]
    [f,p1,action={rL},@\nodef]
    [g,p1,action={rR},@\nodeg]
  ]
]
\end{forest}
}
\caption{
The belief game constructed by applying \ref{al:auxgame} to the adversarial team game in \Cref{fig:example-atg}. To improve clarity, nodes with single actions have been removed. Note how the imperfect recall between nodes $b$ and $c$ from \Cref{fig:example-connectivity} is solved by merging them and using prescription to express an action for each node. Prescriptions have been labeled with the actions prescribed for each node in the belief.
}\label{fig:example-belief-game}
\end{figure*}
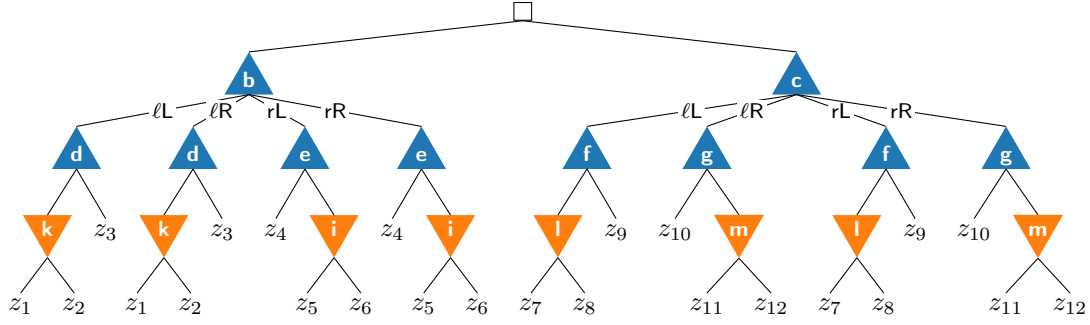

}

\changed{As an example, the belief game for the game in  \Cref{fig:example-connectivity} can be found in \Cref{fig:example-belief-game}.}

\subsection{Strategic equivalence}\label{sec:auxgame_strategic_equivalence}

The main result of this section is the following.

\begin{theorem}\label{th:strategic-equivalence}
    Let $G$ be any two-player imperfect-recall extensive-form game, and $\tilde G$ be the belief game constructed by \ref{al:auxgame}. $G$ and $\tilde G$ are strategically equivalent.
\end{theorem}

We use the remainder of the subsection to prove this result. The requirements for strategic equivalence as per \Cref{def:strategic_equivalence} that we will show are: i) there exists a function $\rho$ mapping pure strategies in the two games ii) $\rho$ is a bijective function iii) $\rho$ is value-preserving, that is, the mapped strategies have the same expected utilities.

We first construct the strategy maps $\rho_i : \Pi_i \to \tilde \Pi_i$. 
A pure strategy $\vec\pi_i \in \Pi_i$ in $G$ assigns one action to each information set.\footnote{At infosets not reached by $\vec\pi_i$, actions can be selected arbitrarily.} From such a strategy we construct a strategy $\tilde{\vec\pi}_i = \rho_i(\vpi_i)$ which plays prescriptions consistent with $\vec \pi_i$. At belief $B_i$, $\tilde{\vec\pi}_i$ plays the prescription $\vec a_i$ given by $\vec a_i[I] = \vec\pi_i[I]$ for each $I \in \mc I_{B_i}$ reached by $\tilde{\vec\pi}$.

We now show that $\rho_i$ is injective. This will follow from the following lemma.
\begin{lemma}\label{lm:terminal_reach}
    Let $\tilde\vpi_i = \rho_i(\vpi_i) \in \tilde\Pi_i$. Then for every $z \in \Z$, $\tilde\vpi_i$ plays to belief $\{z\}$ if and only if $\vpi_i$ plays to $z$.
\end{lemma}
\begin{proof}
    First suppose $\tilde\vpi_i = \rho_i(\vpi_i) \in \tilde\Pi_i$ plays to $\{z\}$. Thus $\tilde\vpi_i$ plays some sequence of beliefs and prescriptions $(B^1_i=\{\Root\}, \vec a^1_i, \dots, B^2_i, \vec a^2_i, \dots, B^L_i=\{z\})$. But then, by construction, every ancestor $h \prec z$ is included in one of the $B^\ell_i$s, and %
    for $ha \preceq z$ to appear in $B^\ell_i \vec a^\ell_i$, if $h \in \H_i$, it must be the case that $\vpi_i$ plays $a$. Thus $\vpi_i$ plays to $z$.

    Conversely suppose $\vpi_i$ plays to $z$. Then, construct a sequence of beliefs and prescriptions as follows. Let $\vec a^\ell_i$ be the prescrption played by $\tilde\vpi_i$ at belief $B_i^\ell$, and $B_i^{\ell+1} = \splitb i (B_i^\ell \vec a_i^\ell, h)$ where $h \in B_i^\ell\vec a_i^\ell$ and $h \preceq z$. (The fact that $\vpi_i$ plays to $z$ ensures that such $h$ must exist). Then, by induction, $\tilde\vpi_i$ plays to this sequence, and the sequence must eventually terminate at $\{z\}$ because it always contains at least one predecessor of $z$. Thus $\tilde\vpi_i$ plays to $\{z\}$.
\end{proof}
Since different pure strategies (by definition) play to different sets of terminal nodes, this immediately shows that $\rho$ is injective. We now show that $\rho_i$ is a surjection (and hence a bijection), that is, any $\tilde{\vec\pi}_i \in \tilde\Pi_i$ is the image of some $\vec\pi_i \in \Pi_i$. We remark that a pure strategy $\tilde{\vec\pi}_i \in \tilde\Pi_i$ assigns one prescription to every reached infoset of player $i$ in $\tilde G$. %

We will require the following lemma. Informally, it states that no player can play to two nodes with intersecting beliefs.
\begin{lemma}\label{lm:belief_common_ancestor}
    Let $\tilde{\vec\pi}_i \in \tilde\Pi_i$ be any pure strategy. Let $\tilde I, \tilde I' \in \tilde{\mc I}_i$ be two distinct infosets of player $i$ that are simultaneously reached by $\tilde{\vec\pi}_i$. Let $B_i$ and $B_i'$ be the beliefs for player $i$ at $\tilde I$ and $\tilde I'$ respectively. Then $B_i$ and $B_i'$ are not connected and do not intersect. That is, there do not exist nodes $h \in B_i, h' \in B_i'$ with $(h, h') \in \mc E_i$  or  $h = h'$. 
\end{lemma}
\begin{proof}
    Consider the tree-form decision problem $\tilde{\mathcal T}$ of player $i$ in $\tilde G$. Since $\tilde G$ is perfect-recall, $\tilde{\mathcal T}$ is indeed a valid representation of player $i$'s strategy set in $\tilde G$. Let $s$ be the lowest common ancestor of $\tilde I$ and $\tilde I'$ in $\tilde{\mathcal T}$. %
    
    Since $\tilde{\vec\pi}_i$ plays to both $\tilde I$ and $\tilde I'$, the node $s$ must be an observation point, as a pure strategy plays a single action at each decision point. At observation points, the next observation made by player $i$ is the next belief. Let $C_i$ and $C_i'$ be the different observed beliefs at node $s$ that lead to $\tilde I$ and $\tilde I'$, respectively. Then, by construction of \splitb{i}, $C_i$ and $C_i'$ are disconnected and disjoint. Since every node in $B_i$ is a descendant of some node in $C_i$ (and the same for $C_i'$), it follows that $B_i$ and $B_i'$ are also disconnected and disjoint.
\end{proof}

Thus, in particular, for any infoset $I \in \mc I_i$ in the original game, $\tilde{\vec\pi}_i$ can only play to one infoset $\tilde I \in \tilde{\mc I}_i$ whose belief $B_i$ overlaps $I$. At $B_i$, the prescription chosen by $\tilde{\vec\pi}_i$ includes an action $\vec a[I]$ at $I$ (by construction of prescriptions at a node). Thus, consider the strategy $\vec\pi_i$ defined such that $\vec\pi_i[I] = \vec a[I]$ for every infoset $I$ such that $\vpi_i$ plays to a belief $B_i$ overlapping $I$.

\begin{lemma}
    $\vpi_i$ is a well-defined strategy. That is, if $\vpi_i$ plays to a node $h \in \H_i$, then $\tilde\vpi_i$ plays to a belief $B_i \ni h$, and hence, if $h \in I \in \I_i$ then  $\vpi_i[I]$ is defined.
\end{lemma}
\begin{proof}
    By induction on the length of the history $h$. For $h = \Root$ this is trivial. Now let $h = h'a'$ be a non-root node, and suppose $\vpi_i$ plays to $h$. Then by inductive hypothesis, $\tilde\vpi_i$ plays to a belief $B_i \ni h'$. Let $\vec a$ be the prescription played by $\tilde\vpi_i$ at $B_i$. 
    \begin{itemize}
        \item If $h' \in I \in \I_i$, then by construction of $\vpi_i$, it must be the case that $\tilde\vpi_i$'s prescription $\vec a$ at $B_i$ satisfies $\vec a[I] = a'$, so $B_i \vec a \ni h$.
        \item If $h' \notin \H_i$, then for {\em any} prescription $\vec a$ at $B_i$ we have that $h \in B_i \vec a$. 
    \end{itemize}
    In either case, we have $B_i \vec a \ni h$. Thus, $\tilde \vpi_i$ must also play to the belief $\splitb i (B_i \vec a, h) \ni h$. 
\end{proof}

Further, from the definition of $\rho_i$ it follows immediately that $\rho_i(\vec\pi_i) = \tilde{\vec\pi}_i$. It only remains to show that $\rho_i$ is value-preserving. 
But this is easy: $\rho_i(\vec\pi_i)$ prescribes the same actions as $\vec\pi_i$. Thus, for any pure strategy profile $\vec\pi \in \Pi$, following profile $\vec\pi$ through $G$ will yield exactly the same trajectory as following profile $\rho(\vec\pi)$ through $\tilde G$. Thus, their expected utilities will also coincide, and the proof is complete.

\subsection{Worst-case dimension of the belief game}\label{sec:worst-case-size-aux}

The per-iteration time complexity of \ref{al:cfr} depends linearly on the size of the game on which the algorithm is applied. Thus, it is critical for complexity analysis to bound the size of the belief game produced by \ref{al:auxgame}.

\paragraph{Lower Bound} We first present a lower bound of the worst-case size of the belief game by describing a worst-case game whose belief game has a large number of histories.
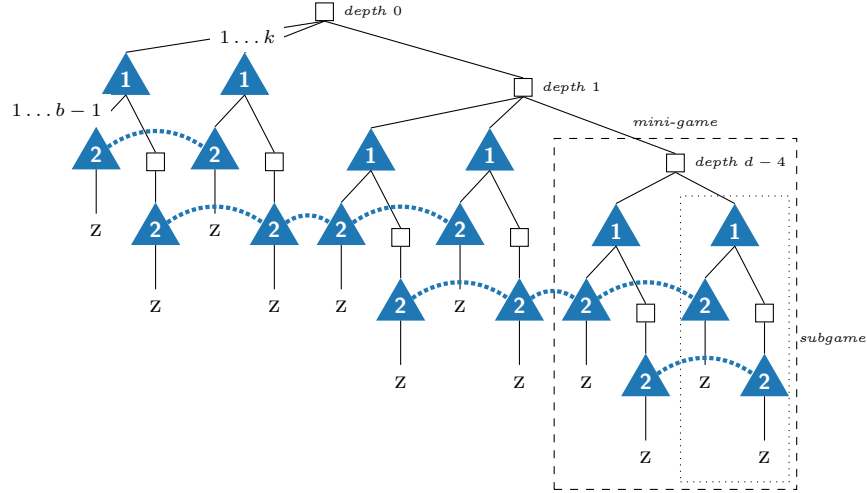
\begin{figure}[t]
\newcommand{\makeinfoset}[1]{
    #1/.style={tier=#1,
        if nodewalk valid={previous on tier=#1}%
            {tikz={\draw[infoset1, bend left=35] (!previous on tier=#1.35) to (.145);}}%
            {}%
    },
}
\centering
\begin{forest}
terminal/.style={draw=none, inner sep=4pt},
@\makeinfoset{i1}
@\makeinfoset{i2}
@\makeinfoset{i3}
@\makeinfoset{i4}
[,nat,label=right:\emph{\tiny depth $0$},
    [1, p1,
        [2,p1,i1,edge label={node[midway,draw=none,left,fill=white,font=\scriptsize]{$1\dots b-1$}}
            [z, terminal]
        ]
        [,nat,
            [2,p1,i2
                [z, terminal]
            ]
        ]
    ]
    [1, p1, edge label={node[midway,draw=none,left,fill=white,font=\scriptsize]{$1\dots k$}}
        [2,p1,i1,
            [z, terminal]
        ]
        [,nat,
            [2,p1,i2
                [z, terminal]
            ]
        ]
    ]
    [,nat,label=right:\emph{\tiny depth $1$},
        [1, p1
            [2,p1,i2
                [z, terminal]
            ]
            [,nat,
                [2,p1,i3
                    [z, terminal]
                ]
            ]
        ]
        [1, p1
            [2,p1,i2
                [z, terminal]
            ]
            [,nat,
                [2,p1,i3
                    [z, terminal]
                ]
            ]
        ]
        [,nat,tikz={\node [draw,dashed,fit=()(!next leaf)(!last leaf),inner sep=2mm,label=north:\emph{\tiny mini-game}] {};},label=right:\emph{\tiny depth $d-4$},
            [1, p1
                [2,p1,i3
                    [z, terminal]
                ]
                [,nat,
                    [2,p1,i4
                        [z, terminal]
                    ]
                ]
            ]
            [1, p1,tikz={\node [draw,dotted,fit=()(!next leaf)(!last leaf),inner sep=1mm,label=right:\emph{\tiny subgame},] {};}
                [2,p1,i3
                    [z, terminal]
                ]
                [,nat,
                    [2,p1,i4
                        [z, terminal]
                    ]
                ]
            ]
        ]
    ]
]
\end{forest}
\caption{An example of imperfect-recall game derived from a team game whose rationale is described in the proof of \Cref{th:aux-size-lower}.
Given that the terminal values of this game are not relevant, we use z to indicate a generic payoff. The boxes indicated by \emph{mini-game} and \emph{subgame} correspond to the terms used in the description.}
\label{fig:worst-case-game}
\end{figure}
\begin{theorem}\label{th:aux-size-lower}
    There exists a game $G$ with depth $d$, information complexity $k$, and maximum branching factor at a node $b$ such that the number of nodes in the belief game $\tilde G$ is $|\tilde \H| \geq b^{2k(d-4)}$.
\end{theorem}
\begin{proof}
    We construct a parametric game for depth $d \geq 4$, information complexity $k$, and branching factor $b \geq k+1$.
    
    Consider a game that consists of $d-3$ repetitions of the following \emph{mini-game}.
    There is a \emph{"root"} chance node with $k$ nodes of \pmax, $k$ nodes of \pmin and the root of the next repetition of the mini-game as children. Each of \pmax's and \pmin's nodes belong to different information sets.
    Each of those is the root of \emph{subgames} with identical structures. Their children are $b-1$ player nodes followed by a single terminal node, and there is a chance node followed by a player node and then a single terminal node. These player nodes belong to the same player as the root of the subtree, namely \pmax for the first $k$ children and \pmin for the second $k$ children of the \emph{"root"} of the mini-game.
    All nodes of \pmax and \pmin at the same level of the subgames apart from the $2k$ nodes belong to the same information set.
    
    A representation of such a game for $k=2, b=2, d=6$ is given in \Cref{fig:worst-case-game}, where nodes of \pmin have been omitted to improve clarity. Those have the same structure as the nodes of \pmax.
    The main point of this game is that at any level $l=1,\dots,d-3$, both \pmax and \pmin have a public state containing the $k$ infosets corresponding to the nodes after \nature's \emph{"root"} of the \emph{mini-game} at depth $l-1$. 
    At each of those public states, both \pmax and \pmin have a single belief containing $k$ information sets and the chance node that leads to the next \emph{mini-game}. Therefore, there are $b^{k}$ prescriptions at this belief, and each prescription played leads, among the others, to the belief containing $k$ information sets of the next \emph{mini-game}.

    Multiplying this factor for each step of the level gives $|\H'| \geq \left(b^{2k}\right)^{d-4}$.
\end{proof}

\paragraph{Upper Bound} We now present an upper bound on the number of histories of the belief game.

\begin{theorem}\label{th:aux-size-upper}
    Let $G$ be a game with depth $d$, information complexity $k$, and maximum branching factor $b$.
    The number of nodes in the belief game $\tilde G$ is $|\tilde\H| \leq b^{2kd+d}$.
\end{theorem}
\begin{proof}
    Consider the algorithm \ref{al:auxgame}. Grouping levels of the belief game $\tilde G$ three by three, we have that \pmax, \pmin, and \nature each play at a different level in each group, and we have $d$ groups. Moreover, no more than $b$ actions for the chance player and $b^k$ prescriptions are available to each player. We, therefore, have that the number of nodes in-game tree $\tilde G$ is $|\tilde \H| \leq b^{d(2k+1)}$.
\end{proof}

\paragraph{Discussion}%
The bounds presented in this section highlight the main computational limitation of \Cref{al:auxgame}, the explicit dependence on depth introduced by explicitly using sequences to distinguish information sets in the belief game.

We remark that we can replace $k$ here with the maximum number of infosets (not the last-infosets) in any public state. We opted not to introduce two different notions of information complexity to have bounds comparable with the TB-DAG ones in \Cref{sec:worst_case_size_tbdag}. We will explore the effects of introducing the different definitions of $k$ in \Cref{sec:discussion:information_complexity}.

\subsection{Regret minimization on team games}\label{sec:regret-minimization-team-games}

This section shows how to find a mixed Nash equilibrium in a generic two-player zero-sum game with imperfect recall $G$ by applying \ref{al:cfr} on the belief game $\tilde G$ obtained by running \Cref{al:auxgame} on $G$.

Let $\tilde\X$ and $\tilde\Y$ be the realization-form strategy spaces for $\pmax$ and $\pmin$ in $\tilde G$ derived from the sequence-form representation as in \Cref{sec:tree-form-decision-problems}. Specifically, vectors $\tilde{\vec x} \in \tilde\X$ are indexed by terminal sequences for \pmax in $\tilde G$ (similarly for \pmin). Such a sequence $\sigma$ can be identified by a list of beliefs and prescriptions, ending in a singleton belief $\{ z \}$ for terminal node $z \in \Z$. For any terminal node $z$, let $\Sigma^z_\pmax$ be the set of terminal sequences for \pmax that end at belief $\{z\}$. Then computing a Nash equilibrium in $\tilde G$ (and hence a mixed Nash in $G$) can be done by solving the max-min problem
\begin{align}
    \max_{\tilde{\vec x} \in \tilde\X} \min_{\tilde{\vec y} \in \tilde\Y} \sum_{z \in \Z}  u(z) \sum_{\tilde\sigma_\pmax \in \Sigma^z_\pmax} \tilde{\vec x}[\tilde\sigma_\pmax] \sum_{\tilde\sigma_\pmin \in \Sigma^z_\pmin} \tilde{\vec y}[\tilde\sigma_\pmin]. \label{eq:Gtilde-maxmin}
\end{align}
This is equivalent to the max-min problem for the coordinator game (Eq. \eqref{eq:maxmin}) by setting $\vec x[z] := \sum_{\tilde\sigma_\pmax \in \Sigma^z_\pmax} \tilde{\vec x}[\tilde\sigma_\pmax]$ (and similar for $\vec y$). That is, from an optimization perspective, what has happened is that we have constructed sets $\tilde\X$ and $\tilde\Y$ that are described by linear constraints, just like the sequence form, and {\em project} onto $\X$ and $\Y$ respectively, allowing the reformulation and equivalence of problems.

We now analyze the time complexity and regret of running CFR on $\tilde G$. Fix a player, say, $\pmax$. (The same analysis will apply to $\pmin$.) First, recall from \Cref{sec:tree-form-decision-problems} that, in a decision problem, a set $P$ of \pmax-decision points is called {\em playable} if there exists a pure strategy of \pmax that plays to all the decision points in $P$. But by \Cref{lm:belief_common_ancestor}, the size of any playable set $P$ of \pmax is at most $|\mc H|$. Further, the branching factor of $\tilde G$ is at most $b^k$, where $b$ is the branching factor of $G$ and $k$ is the information complexity (see \Cref{sec:worst-case-size-aux}). Thus, applying multiplicative weights ({\tt MWU}) as the local regret minimizer at each decision point and using \Cref{th:cfr}, we have:
\begin{theorem}\label{th:aux-tree-cfr}
    After $T$ iterations of \ref{al:cfr} on $\tilde G$ with {\tt MWU} as the local regret minimizer, the average strategy profile $(\bar{\vec x}, \bar{\vec y})$ is an $O(\eps)$-Nash equilibrium of $G$, where
    \begin{align}
        \eps = |\mc H| \sqrt{\frac{k \log b}{T}}.
    \end{align}
    The per-iteration complexity is linear in the size of $\tilde G$.
\end{theorem}

While the regret above is polynomial in $\mc H$, the per-iteration complexity depends on the size of $\tilde G$, which is worst-case exponentially larger than $G$, as shown in \Cref{sec:worst-case-size-aux}.

\section{DAG decision problems}\label{sec:dag-decision-problems}

\begin{algorithm}[t]\namealg{DAG-Generic}
\caption{Generic construction of a regret minimizer $\mc R$ on $\mc Q$ from a regret minimizer $\hat{\mc R}$ on its tree form $\hat{\mc Q}$.}\label{al:dag-rm-generic}
\begin{algorithmic}[1]
\Procedure{NextStrategy}{}
    \State $\hat{\vx}^t \gets \hat{\mc R}$.{\sc NextStrategy()}
    \State \Return{$\vec D\hat\vx^t$}
\EndProcedure
\Procedure{ObserveUtility}{$\vec u^t$}
    \State $\hat{\mc R}$.{\sc ObserveUtility($\vec D^\top {\vec u}^t$)}
\EndProcedure
\end{algorithmic}
\end{algorithm}

\begin{algorithm}[t]\namealg{DAG-CFR}
\caption{Counterfactual regret minimization on DAG-form decision problems $\mc Q$. For each decision point $s$, $\mc R_s$ is a regret minimizer on $\Delta(\A_s)$.}\label{al:dag-cfr}
\begin{algorithmic}[1]
\Procedure{NextStrategy}{}
    \State $\vx^t[\Root] \gets 1$ 
    \For{each decision point $s$, in top-down order} 
        \State $\vec r_s^t \gets \mc R_s$.{\sc NextStrategy()}
        \State $\displaystyle \vx^t[s*] \gets  \sum_{p \in P_s} \vx^t[p] \vec r_s^t$
    \EndFor
    \State \Return{$\vx^t$}
\EndProcedure
\Procedure{ObserveUtility}{$\vec u^t$}
    \State $\vec v^t \gets \vec u^t$
    \For{each decision point $s$, in bottom-up order} 
        \State $\mc R_s$.{\sc ObserveUtility}($\vec v^t[s*]$) 
        \For{$p \in P_s$} $\vec v^t[p] \gets \vec v^t[p] + \ip{\vec r_s^t, \vec v^t[s*]}$ \EndFor
    \EndFor
    \State $t \gets t + 1$
\EndProcedure
\end{algorithmic}
\end{algorithm}

In this section, we will develop a general theory of DAG-form decision problems and regret minimization on them, analogous to the tree-form theory in \Cref{sec:prelims}. Although our main interest in DAG-form decision-making is its application to two-player imperfect-recall games (which we will develop in \Cref{sec:tbdag}), the observations made in this section also have general applicability beyond this setting. For example, since the publications of earlier versions of the present paper, DAG-form decision-making has been applied toward the efficient computation of many other solution concepts, including {\em linear correlated equilibria} and {\em optimal extensive-form correlated equilibria}~\cite{Zhang22:Polynomial,Zhang22:Optimal,Zhang23:Computing,Zhang24:Efficient,Zhang24:Mediator}.

As one may expect, DAG-form decision problems are identical to tree-form decision problems except that the graph of nodes is allowed to be a DAG, albeit with some restrictions. 
\begin{definition}
    A {\em DAG-form decision problem} is a DAG with a unique source (root node) $\Root$, wherein each node $s$ is either a decision point $(s\in \mc D)$ or an observation point $(s\in \mc S)$, with the following properties.
    \begin{enumerate}
        \item Observation points other than the root have exactly one incoming edge.
        \item For any two paths $p_1$ and $p_2$ from the root that end at the same node, the last node in common between $p_1$ and $p_2$ is a decision point.
    \end{enumerate}
\end{definition}
As with tree-form decision problems, we will also assume (WLOG) that decision and observation points alternate along every path and that both the root node and all terminal nodes are observation points. A {\em pure strategy} is once again an assignment of one action to each decision point. The {\em DAG form} of a pure strategy is the vector $\vx \in \{0, 1\}^{\mc S}$, where $\vx[s]=1$ if there is {\em some} $\Root \to s$ path along which the player plays all actions. A mixed strategy $\vx \in \mc Q$ is a convex combination of pure strategies. Since decision points can now have multiple parents, we will use $P_s$ to denote the set of parents of a decision point $s$.

Like tree-form decision problems, the mixed strategy set in a DAG-form decision problem has a convenient representation using linear constraints, namely:
\begin{align}
\left\{
\begin{aligned}
            \vec x[\Root] &= 1 \\
            \sum_{p \in P_s} \vec x[p] &= \sum_{a \in \A_s} \vec x[s a]  & \qq{for all} s \in \mc D.
    \end{aligned} \right.
    \label{eq:dag_strategy_root}
    \end{align}

DAG-form decision problems and tree-form decision problems are closely related. Of course, all tree-form decision problems are DAG-form decision problems. Conversely, any DAG-form decision problem can be thought of as a ``compressed'' representation of the tree-form decision problem \changed{whose paths correspond to} all the different paths through the DAG. While this tree will generally be exponentially larger than the DAG, we will find it useful to compare the DAG and tree representations.

We now formulate a general theory of regret minimization for DAG-form decision problems. We will use hats $(\hat{\mc Q}, \hat{\mc D}, \hat{\mc S}, \hat{\vx})$ to denote components of the tree form of a generic DAG-form regret minimizer. For each tree-form observation point $s \in \hat{\mc S}$, let $\delta(s) \in \mc S$ be the corresponding observation point in $\mc S$. Note that, by construction, $\delta$ is surjective but not injective unless the DAG happens to be a tree.

We now show how tree-form strategies and utilities correspond to DAG-form strategies and utilities. Concretely, we define a matrix $\vec D \in \R^{\mc S \times \hat{\mc S}}$ by $\vec D \hat{\vx}[s] = \sum_{\hat s :\delta(\hat s) = s} \hat{\vx}[\hat s]$ for all $\hat{\vx} \in \R^{\hat{\mc S}}$.  This is the matrix of the linear map that transforms tree-form strategies to their corresponding DAG-form strategies.  That is, $\vec D : \hat \X \to \X$ is a bijection.

Dually, for DAG-form utility vectors $\vec u \in \R^{\mc S}$, the vector $\vec D^\top \vec u \in \R^{\hat{\mc S}}$ is a utility vector on the tree form, with the property that
$
    \ip{\vec D^\top \vec u, \hat{\vx}} = \ip{ \vec u,\vec D \hat{\vx}}
$
by definition of the inner product. That is, the DAG-form strategy $\vec D \hat{\vx}$ achieves the same utility against DAG-form utility vector $\vec u$ as the tree-form strategy $\hat{\vx}$ achieves against the utility $\vec D^\top{\hat{\vu}}$. 

The relationship between trees and DAGs allows us to use {\em any} regret minimizer on $\hat{\mc Q}$ to construct a regret minimizer with the same guarantee on $\mc Q$. We do this in \Cref{al:dag-rm-generic}. 
\begin{restatable}{proposition}{daggeneric}
    Let $\mc R$ and $\hat{\mc R}$ be as in \Cref{al:dag-rm-generic}. Then the regret of $\mc R$ with utility sequence $\vec u^1, \dots, \vec u^T$ is equal to the regret of $\hat{\mc R}$ with utility sequence $\vec D^\top \vec u^1, \dots, \vec D^\top \vec u^T$. 
\end{restatable}
Proofs of results in this section are deferred to \Cref{app:proof:dag-decision-problems}.

Applying the transformation \ref{al:dag-rm-generic} with \ref{al:cfr} as the tree-form regret minimizer $\hat{\mc R}$, we arrive at a DAG form of CFR, which can be simulated efficiently: \Cref{al:dag-cfr}. One can think of \ref{al:dag-cfr} as a {\em more efficient implementation} of  \ref{al:cfr} when the decision problem happens to have a DAG structure. Of course, the regret bound $O(|\mc S| \sqrt{T})$ is only a worst-case bound; in special cases (such as \Cref{th:aux-tree-cfr}), \ref{al:cfr} does much better than its worst case, and therefore so will \ref{al:dag-cfr}.

Call a utility vector $\hat{\vec u}$ {\em consistent} if it is in the image of $\vec D^\top$. That is (\changed{by} definition of $\vec D^\top$), $\hat{\vec u} \in \R^{\hat{\mc O}}$ is consistent if $\hat{\vec u}[\hat s] = \hat{\vec u}[\hat s']$ if $\delta(\hat s) = \delta(\hat s')$. In essence, a DAG-form regret minimizer is able to ``simulate'' a tree-form regret minimizer so long as the tree-form regret minimizer's utilities are always consistent. We now formalize this idea.

\begin{restatable}[DAG regret minimization via CFR]{theorem}{dagrm}\label{th:dag-cfr}
    \ref{al:dag-cfr} produces the same iterates as \ref{al:dag-rm-generic} with \ref{al:cfr} as $\hat{\mc R}$. Therefore, in particular, the regret of \ref{al:dag-cfr} with utility sequence $\vec u^1, \dots, \vec u^T$ is the same as that of \ref{al:cfr} on the tree form with utility sequence $\hat\vu^1 := \vec D^\top \vec u^1, \dots, \hat\vu^t := \vec D^\top \vec u^T$. Moreover, the per-iteration runtime of \ref{al:dag-cfr} is linear in the number of edges in the DAG. In particular, taking any reasonably efficient regret minimizer over simplices, the regret of \ref{al:dag-cfr} after $T$ iterations is at most $O(|\mc S| \sqrt{T})$.
\end{restatable}

\section{DAG decision problems in team games}\label{sec:tbdag}

\begin{algorithm}[!t]\namealg{ConstructTB-DAG}
    \caption{Constructing the TB-DAG. Inputs: imperfect-recall game $G$, player $i$}\label{al:tb-dag}
\begin{algorithmic}[1]
\Procedure{MakeDecisionPoint}{$B$}
\Comment{$B \subseteq \H$ is a belief}
    \If{a decision point $s$ with belief $B$ already exists} \Return{$s$} \EndIf
    \If{$B = \{z \}$ for $z \in \Z$} 
        \Return{new terminal node with belief $\{z\}$}
    \EndIf
    \State $s \gets{}$new decision point with belief $B$
    \For{each prescription $\vec a \in A_B^i$}
        \State add edge $s \to \algoname{MakeObservationPoint}(B \vec a)$
    \EndFor
    \State\Return{$s$}
\EndProcedure
\Procedure{MakeObservationPoint}{$H$}
    \State $s \gets{}$new observation point
    \For{each $B \in \splitb{i}(H)$} 
    \State add edge $s \to \algoname{MakeDecisionPoint}(B)$
    \EndFor
    \State\Return{$s$}
\EndProcedure
\end{algorithmic}
\end{algorithm}

In \Cref{sec:regret-minimization-team-games}, it emerged that applying the \ref{al:cfr} procedure to the belief game produced by \ref{al:auxgame} suffers from the size of the game to solve, which may grow exponentially fast as shown in \Cref{sec:worst-case-size-aux}.
In this section, we show how DAG decision problems can greatly reduce the inefficiencies caused by the previous construction.

The main observation is that \ref{al:auxgame} enforces perfect recall of the belief game by including the players' sequences in the infoset definition. 
On the other hand, the strategic aspect of the game is governed solely by the nodes contained in beliefs. Once the set of possible nodes is fixed, the exact sequence of prescriptions and observations is not relevant, as the game will evolve identically from that point onwards.
This observation leads to considering a DAG structure for the decision problems, where decision nodes are identified by beliefs.

The Nash equilibrium problem in $\tilde G$, namely \eqref{eq:Gtilde-maxmin}, indeed guarantees that both players' utility vectors will be consistent with respect to these DAG-form decision problems. We will call the resulting DAG decision problems the {\em team belief DAGs} (TB-DAGs)\footnote{We keep the name {\em team belief DAG} for continuity with previous versions of the work, even though it applies equally well in the team and imperfect recall settings.}. Therefore, using \ref{al:dag-cfr} as the regret minimizer for both players, we recover the regret guarantee of \Cref{th:aux-tree-cfr} with {\em per-iteration complexity} proportional to the total size of both DAGs.

However, this proposed algorithm still depends on the size of $\tilde G$, because, naively, to construct the DAG representations, one first constructs the augmented game $\tilde G$, and only then does the merging of decision points to create the DAGs.
We therefore describe an algorithm  \ref{al:tb-dag} that recursively constructs the team belief DAGs {\em directly from the original game $G$}, thus bypassing the construction of $\tilde G$. Therefore, we have the following result. For each player $i \in \{ \pmax, \pmin\}$, let $E_i$ be the number of edges in the TB-DAG of player $i$.

\begin{theorem}[TB-DAG and CFR]\label{th:tb-dag-cfr}
Suppose that both players run the algorithm \ref{al:tb-dag} to construct their strategy spaces $\tilde \X, \tilde \Y$, and then run \ref{al:dag-cfr}. Then their average strategy profile converges at the rate shown in \Cref{th:aux-tree-cfr}, and the per-iteration runtime complexity is $O(E_\pmax + E_\pmin)$. 
\end{theorem}

\begin{proof}
\ref{al:tb-dag} is designed to construct a DAG-form decision problem whose tree form corresponds precisely to the decision problem faced by player $i$ in the belief game. Thus, it only remains to show that \Cref{th:dag-cfr} applies. That is, we need to show that, in the belief game $\tilde G$, the utility vectors $\hat\vu$ that would be observed by \ref{al:cfr} for \pmax (the same proof would hold for \pmin) are such that $\hat\vu[\tilde z] = \hat\vu[\tilde z']$ whenever histories $\tilde z, \tilde z' \in \tilde \Z$ of the belief game represent the same history $z \in \Z$ of the original game. But indeed, for any pure strategy $\tilde\vy \in \tilde\Y$, by \Cref{th:strategic-equivalence} there is a pure strategy $\vy \in \tilde\Y$ such that $\tilde\vy[\tilde z] = \tilde\vy[\tilde z'] = \vy[z]$. Thus, if \pmin plays $\tilde\vy$, the utility observed by \pmax will be given by
\begin{align}
    \hat \vu[\tilde z] = \tilde{\vec p}[\tilde z] \tilde \vy[\tilde z] = \vec p[z] \vy[z] = \tilde{\vec p}[\tilde z'] \tilde \vy[\tilde z'] = \hat \vu[\tilde z']
\end{align}
which is indeed consistent in the required sense. 
\end{proof}

\changed{
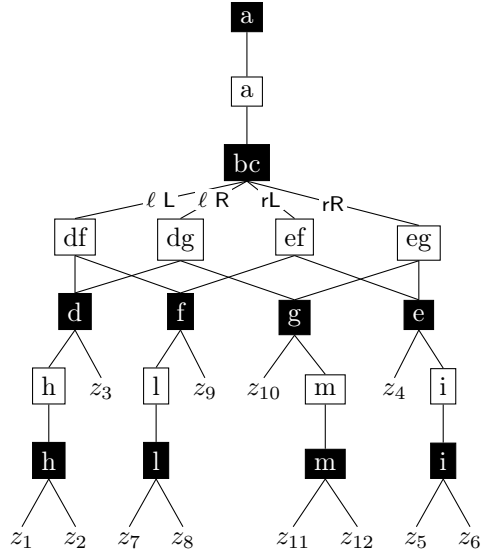
\begin{figure*}[t]
\centering
\newcommand{\noded}{d, name=d, no edge, dec,
[h,[h, dec,[$z_1$,][$z_2$]]][$z_3$]
}
\newcommand{\nodef}{f, name=f, no edge, dec,
[l,[l, dec,[$z_7$][$z_8$]]][$z_9$]
}
\newcommand{\nodee}{e, name=e, no edge, dec,
[$z_4$]
[i,[i, dec,[$z_5$][$z_6$]]]
}
\newcommand{\nodeg}{g, name=g, no edge, dec,
[$z_{10}$]
[m,[m, dec,[$z_{11}$][$z_{12}$]]]
}
\begin{forest}
dec/.style={draw, fill=black, text=white},
for tree={
if n children=0{
  terminal
}{}
}
[a, dec,
  [a
    [bc, dec,
      [df, action={$\ell$ L}, name=df, [@\noded]]
      [dg, action={$\ell$ R}, name=dg, [@\nodef]]
      [ef, action={rL}, name=ef, [@\nodeg]]
      [eg, action={rR}, name=eg, [@\nodee]]
    ]
  ]
]
\draw (df.south) to (d.north);
\draw (df.south) to (f.north);
\draw (dg.south) to (d.north);
\draw (dg.south) to (g.north);
\draw (ef.south) to (e.north);
\draw (ef.south) to (f.north);
\draw (eg.south) to (e.north);
\draw (eg.south) to (g.north);
\end{forest}
\caption{
The TB-DAG corresponding to the adversarial team game in \Cref{fig:example-atg}. Decision and observation nodes are represented in black and white respectively. Note how subgames and terminal nodes that were repeated in the belief game from \Cref{fig:example-belief-game} are now merged into a subgraph which can be reached through different trajectories of play.
}\label{fig:example-tbdag}
\end{figure*}

}

\changed{An example of TB-DAG is shown in \Cref{fig:example-tbdag}.}

\subsection{Size analysis of the TB-DAG}\label{sec:worst_case_size_tbdag}

The per-iteration runtime above depends on the number of edges in the TB-DAGs, so it is important to bound this number. We will do so now. Here, we use the same notation as in \Cref{sec:worst-case-size-aux}.

\begin{theorem}\label{th:tb-dag-size}
    For each player $i$, we have $E_i \le |\H| (b+1)^{k+1}$. 
\end{theorem}
\begin{proof}
    At any given public state $P$ of player $i$, let $\I_i(P) = \bigcup_{h \in P}\I_i(h)$ be the set of last-infosets. Consider a pure strategy $\vpi \in \Pi_i$. For each last-infoset $I \in \I_i(P)$, let $\vpi_P$ be the partial strategy defined only on infosets $I \in \mc I_i(P)$, by $\pi_P[I] = \vpi[I] \in A_I$ if $\vpi$ plays to at least one node in $I$, and $\pi_P[I] = \bot$ if it does not. There are thus at most $(b+1)^k$ such possible partial strategies, since $|\I_i(h)| \le k$ by definition. By construction, each partial strategy $\vpi_P$ completely determines which nodes in $P$ are reached by $\vpi$, as well as the actions played at any such nodes. Thus, $\vpi_P$ induces a disjoint collection of observation points $S \subseteq \mc S$ such that we have $B \vec a \in S$ for each observation point $B \subseteq P$. Now, each observation point $B \vec a$ has at most one incoming edge and at most $|B \vec a|$ outcoming edges. Since the observation points $B \vec a$ are disjoint and have a total size at most $|P|b$, the total number of edges at public state $P$ is at most $|P|(b+1)^{k+1}$. The proof finishes by summing over public states, noting that every history is in exactly one public state.
\end{proof}

Thus, from \Cref{th:tb-dag-cfr} and \Cref{th:tb-dag-size}, it follows that:
\begin{theorem}[Main theorem]\label{th:main}
    Any given imperfect-recall game $G$ can be solved by constructing the TB-DAGs using \ref{al:tb-dag} and running \ref{al:dag-cfr}. After $T$ iterations, the average strategy profile will be an $O(\eps)$-Nash equilibrium where $\eps$ is as in \Cref{th:aux-tree-cfr}. The per-iteration complexity is $O(|\H| (b+1)^{k+1})$.
\end{theorem}

Before proceeding, it is instructive to briefly compare \Cref{th:main} to \Cref{th:aux-size-upper}. The latter result gives a per-iteration complexity that is $O(b^{d(2k+1)})$. Thus, \Cref{th:main} is strictly superior: for \Cref{th:aux-size-upper} to be superior, we would need to have $b^{d(2k+1)} < |\H|(b+1)^{k+1}$, which is impossible for $d \ge 1, k \ge 1, b \ge 2$. We give a more detailed comparison between the two bounds in \Cref{sec:discussion:information_complexity}.

\subsection{Fixed-parameter hardness}

Given the above result, one may ask whether the $b$ can be removed more generally. It turns out that it cannot. Before proceeding, we need to introduce some basic concepts surrounding {\em fixed-parameter tractability}. 

\begin{definition}
    A problem is {\em fixed-parameter tractable} with respect to a parameter $k$ if it admits an algorithm whose runtime on inputs of length $N$ is $f(k) \poly(N)$, for some arbitrary function $f$. 
\end{definition}
The {\em $k$-CLIQUE} problem is to, given a graph $G$ and an integer $k$, decide whether $G$ has a $k$-clique. 
The computational assumption $\FPT \ne \W[1]$ states that $k$-CLIQUE is not fixed-parameter tractable. It is implied by the commonly-believed {\em  exponential time hypothesis} (ETH)~\cite{Chen05:Tight}.\footnote{The exponential time hypothesis states that deciding 3-SAT requires time $2^{\Omega(n)}$, where $n$ is the number of variables.}
\begin{restatable}{theorem}{fixedparameterhard}
    Assuming $\FPT \ne \W[1]$, there is no algorithm for computing the mixed Nash value of a {\em one-player} game of imperfect recall with information complexity $k$ whose runtime is $f(k) \poly(|\H|)$ for some arbitrary function $f$.
\end{restatable}
\changed{The proof via a reduction from $k$-CLIQUE is available in \Cref{app:proof:tbdag}.} Thus, it is impossible to replace the $b$ in \Cref{th:main} with any absolute constant.

\subsection{Branching factor reduction} \label{se:bf reduction proof}

Despite the worst-case hardness of removing the $b$ in \Cref{th:main}, it turns out that, for a natural class of games, we {\em can} remove $b$. In this subsection we will discuss games with {\em action recall}, and prove that in such games, it is without loss of generality to assume that the branching factor is $2$. Intuitively, a player $i$ has action recall if they remember the full sequence of actions they have taken in the past (including the timesteps at which such actions were taken). More formally:
\begin{definition}
    At a node $h \in \H$, let $(a_1, \dots, a_L) \in \A^L$ be the list of actions taken on the $\Root \to h$ path. Define the {\em action sequence} of player $i$ as the sequence $(a_1', \dots, a_L') \in (\A \sqcup \{ \bot \})^L$ where $a_\ell' = a_\ell$ if action $a_\ell$ was taken  by player $i$, and $a_\ell' = \bot$ otherwise. We say that player $i$ has {\em action recall} if, for every infoset $I$ of player $i$, every node in $I$ shares the same action sequence. 
\end{definition}
\begin{theorem}\label{th:team-public actions}
    Given a two-player imperfect-recall game $G$ where both players have perfect action recall, there exists another strategically-equivalent game $G'$ such that the branching factor of $G'$ is at most $2$ at each $h \in \H_\pmax \cup \H_\pmin$, the parameter $k$ in $G'$ is the same as it in $G$, and the size of the game has increased by a factor of $O(\log |\A|)$. 
\end{theorem}

\begin{proof}
To every action $a \in \A$ we associate a unique bitstring of length at most $\ell = O(\log |\A|)$. Assume without loss of generality, for simplicity of notation, that all such bitstrings end with a $0$. 
We will call bitstrings $\tilde a \in \{0, 1\}^{<\ell}$ ``partial actions''.

We replace every internal node $h \in \H \setminus \Z$ with a binary tree of depth $\ell$, where bitstrings that are not prefixes of any action $a \in A_h$ are pruned. If $h$ and $h'$ are in the same infoset $I$, then for every partial action $\tilde a \in \{0, 1\}^{<\ell}$ we connect $h \tilde a$ and $h' \tilde a$ in an infoset, which we will call $I \tilde a$. This creates a new game $G'$, whose parameters we must now analyze.

    For each node $h \in \H \setminus \Z$ and action $a$, $G'$ has createad at most $O(\log |\A|)$ additional nodes (namely, the nodes $h \tilde a$ where $\tilde a$ is a prefix of $a$). Thus, the size of $G'$ is at most $O(|\H| \log |\A|)$. 
    
It thus remains only to bound the information complexity of $G'$. Let $P$ be a public state of player $i$ in $G'$. By construction of action sequences, $P$ contains either only nodes in $\H_i$, or only nodes not in $\H_i$. In the latter case there is nothing to check. In the former case, we have $P \subseteq \{ h \tilde a : h \in P'\}$ for some public state $P'$ of $G$, and partial action $\tilde a \in \{0, 1\}^{<\ell}$. Now let $I$ be a last-infoset of $P'$ in $G$. Then $I$ induces at most one last-infoset in $P$: namely, if $I$ overlaps $P$, then this infoset is simply $I \tilde a$; otherwise, it is the infoset $I \tilde a'$ where $\tilde a' \in \{0, 1\}^{\ell-1}$ is the partial action such that action $\tilde a' 0$ leads to $P$ (which must be uniquely defined by definition of action recall). Thus, $P$ has as many last-infosets as $P'$, so the information complexity of $G'$ is, at most, the information complexity of $G$.
\end{proof}

\begin{corollary}\label{cor:pub-actions}
    In games where both players have action recall, \Cref{th:main} applies with the per-iteration runtime replaced with $O(3^k|\H| \log |\A|)$. 
\end{corollary}

\section{Complexity of adversarial team games}\label{sec:complexity}

\begin{table}
\centering
        \setlength\tabcolsep{1.1mm}
    \scalebox{0.9}{\begin{tabular}{c|c|c}
                   & Team vs Player                          & Team vs Team                                                                                                                          \\ \toprule
            TMECor & \NP-complete~\citep{Koller92:Complexity} & {\color{red}$\Delta_2^\P$-complete} (This paper, Theorems~\ref{th:tmecor-delta2}~and~\ref{th:tmecor-delta2-hard}) \\\midrule
            TME    & \NP-complete~\citep{Koller92:Complexity} & {\color{red}$\Sigma_2^\P$-complete} (This paper, Theorems \ref{th:tme-sigma2} and~\ref{th:tme-sigma2-hard})       \\ \bottomrule
        \end{tabular}}
        \caption{Summary of most of the complexity results shown in \Cref{sec:complexity}.}
\end{table}

Here, we state and prove several results about the {\em complexity} of finding various equilibria in timeable two-player zero-sum games of imperfect recall. 

In all cases, the goal is to solve the following promise problem: given game $G$, threshold value $v$, and error $\eps > 0$ (where all the numbers are rational), determine whether the (mixed or behavioral) value of the game is $\ge v$, or $< v - \eps$. The allowance of an exponentially-small error is to circumvent issues of bit complexity that arise due to the fact that exact behavioral max-min strategies may not have rational coefficients~\cite{Koller92:Complexity}. Throughout this section, it will often be convenient to formulate the hardness gadgets in terms of adversarial team games. We will thus freely utilize the analogy between adversarial team games and coordinator games. For mixed Nash and behavioral Nash respectively, we will refer to the problems as {\sc Mixed} and {\sc Behavioral}.

Although we do not explicitly state it in the theorem statements, all the hardness results are proven by constructing adversarial team games in which both teams have a constant number of players. \changed{Proofs of all theorems in this section are available in \Cref{app:proof:complexity}.}

\begin{restatable}[\citealp{Koller92:Complexity,Chu01:NP,Stengel08:Extensive}]{theorem}{basichardness}\label{th:basic-hardness} %
    Finding the optimal strategy in a one-player, timeable game of imperfect recall is \NP-hard, \changed{even when the desired precision $\eps$ is an absolute constant.}
\end{restatable}

\changed{For completeness, \Cref{app:proof:complexity} includes a self-contained proof. The proof of the above theorem also implies the following corollary.
\begin{restatable}{corollary}{corkdependenceisoptimal}
    Assuming the exponential time hypothesis~\cite{Chen05:Tight}, there is no $2^{o(k)} \cdot \poly(n)$-time algorithm for finding the optimal strategy in a one-player timeable game of imperfect recall. 
\end{restatable}
}

\subsection{Behavioral max-min strategies}

We first show results for {\sc BehavioralMaxMin}. In particular, we will show that it is $\Sigma_2^\P$-complete, first by showing inclusion and then constructing a gadget game to show completeness. (Recall that the inclusion will require an $\eps$-approximation because exact behavioral max-min strategies may contain irrational values.) 

\begin{restatable}{theorem}{tmesigmatwo}\label{th:tme-sigma2} %
    {\sc BehavioralMaxMin} is in $\Sigma_2^\P$. If \pmin has perfect recall, it is in \NP.
\end{restatable}

\begin{restatable}{theorem}{tmesigmatwohard}\label{th:tme-sigma2-hard}%
    {\sc BehavioralMaxMin} is $\Sigma_2^\P$-hard, even when the number of players is constant and there is no chance.
\end{restatable}

We remark that \tmax has A-loss recall in the construction with no chance. This means {\sc BehavioralMaxMin} is also $\Sigma_2^\P$-hard for team games of no chance in which \tmax has A-loss recall.

\subsection{Mixed Nash equilibria}

We now show results for {\sc MixedNash}, namely, we will show that {\sc MixedNash} is $\Delta_2^\P$-complete, again by showing inclusion first and then completeness. Unlike for {\sc BehavioralMaxMin} (\Cref{th:tme-sigma2}), here we will directly construct a separation oracle, and thus be able to recover algorithms for {\em exact} computation.

\begin{restatable}{theorem}{tmecordeltatwo}\label{th:tmecor-delta2}
    {\sc MixedNash} is in $\Delta_2^\P$, even for exact computation $(\eps = 0)$.
\end{restatable}

\begin{restatable}{theorem}{tmecordeltatwohard}\label{th:tmecor-delta2-hard}
    {\sc MixedNash} is $\Delta_2^\P$-hard, even when the number of players is constant and there is no chance.
\end{restatable}

\section{Experiments}\label{sec:experiments}

\begin{table}[!t]
    \centering
    
\noindent\adjustbox{max width=0.85\paperwidth,center}{%
\newcommand{\mymidrulegray}{\arrayrulecolor{gray}\mymidrule\arrayrulecolor{black}}
\newcommand{\mymidrule}{
\cmidrule(lr){1-1}
\cmidrule(lr){2-2}
\cmidrule(lr){3-4}
\cmidrule(lr){5-6}
\cmidrule(lr){7-8}
\cmidrule(lr){9-9}
\cmidrule(lr){10-11}
\cmidrule(lr){12-13}
\cmidrule(lr){14-15}
\cmidrule(lr){16-17}
}
\begin{tabular}{lSSSSSSS|SSSSS|SSSS}
\toprule
\multicolumn{8}{c}{\bf Original game $G$} & \multicolumn{5}{c}{\bf Belief Game $\tilde G$} & \multicolumn{2}{c}{\bf Team \pmax's DAG} & \multicolumn{2}{c}{\bf Team \pmin's DAG} \\
& {\bf Nodes} & \multicolumn{2}{c}{\bf Infosets} & \multicolumn{2}{c}{\bf Sequences} & \multicolumn{2}{c}{\bf Information} & {\bf Nodes} & \multicolumn{2}{c}{\bf Infosets} & \multicolumn{2}{c}{\bf Sequences} & {\bf Vertices} & {\bf Edges} & {\bf Vertices} & {\bf Edges} \\
\multicolumn{1}{c}{$G$} & $|\mathcal H|$ & $|{\I}_\pmax|$ & $|{\I}_\pmin|$ & $|{\Sigma}_\pmax|$ & $|{\Sigma}_\pmin|$ & $\max_{\mathcal P} |P|$ & $\bm{k}$ & $|\tilde{\H}|$ & $|\tilde{\I}_\pmax|$ & $|\tilde{\I}_\pmin|$ & $|\tilde{\Sigma}_\pmax|$ & $|\tilde{\Sigma}_\pmin|$ & $|{\mc D}_\pmax \cup \mc S_\pmax|$ & $|{E}_\pmax|$ & $|{\mc D}_\pmin \cup \mc S_\pmin|$ & $|{E}_\pmin|$ \\
\mymidrulegray$^3$K3 \teamprint{3} & 151 & 24 & 12 & 48 & 24 & 6 & 6 & 2119 & 486 & 12 & 1062 & 24 & 487 & 918 & 37 & 36 \\
$^3$K4 \teamprint{3} & 601 & 32 & 16 & 64 & 32 & 12 & 8 & 45049 & 4487 & 16 & 9800 & 32 & 2100 & 6711 & 49 & 48 \\
$^3$K6 \teamprint{3} & 3001 & 48 & 24 & 96 & 48 & 30 & 12 & 6768601 & 267184 & 24 & 574588 & 48 & 54255 & 336944 & 73 & 72 \\
$^3$K8 \teamprint{3} & 8401 & 64 & 32 & 128 & 64 & 56 & 16 & 617929873 & 13194749 & 32 & 27978704 & 64 & 1783926 & 15564765 & 97 & 96 \\
$^3$K12 \teamprint{3} & 33001 & 96 & 48 & 192 & 96 & 132 & 24 & \unk & \unk & \unk & \unk & \unk & \unk & \unk & \unk & \unk \\
$^4$K5 \teamprint{3,4} & 7801 & 80 & 80 & 160 & 160 & 20 & 10 & 577764601 & 102725 & 10385 & 221810 & 21740 & 26566 & 124875 & 4621 & 15415 \\
$^4$K5 \teamprint{4} & 7801 & 120 & 40 & 240 & 80 & 60 & 15 & 174273721 & 11739640 & 40 & 25581730 & 80 & 998471 & 4658070 & 121 & 120 \\
\mymidrulegray$^3$L133 \teamprint{3} & 12688 & 456 & 228 & 912 & 456 & 9 & 6 & 1293658 & 96115 & 228 & 208136 & 456 & 23983 & 49005 & 685 & 684 \\
$^3$L143 \teamprint{3} & 40409 & 800 & 400 & 1600 & 800 & 16 & 8 & 52745745 & 2625209 & 400 & 5736592 & 800 & 139964 & 417027 & 1201 & 1200 \\
$^3$L151 \teamprint{3} & 19981 & 1000 & 500 & 2000 & 1000 & 20 & 10 & 152692141 & 16564617 & 500 & 36016124 & 1000 & 150707 & 496196 & 1501 & 1500 \\
$^3$L153 \teamprint{3} & 98606 & 1240 & 620 & 2480 & 1240 & 25 & 10 & 1833113016 & 67400747 & 500 & 147671104 & 1240 & 855397 & 3486091 & 1861 & 1860 \\
$^3$L223 \teamprint{3} & 15659 & 1260 & 630 & 2884 & 1442 & 4 & 4 & 521285 & 47579 & 812 & 100420 & 1624 & 32750 & 45913 & 2437 & 2436 \\
$^3$L523 \teamprint{3} & 1299005 & 99168 & 49584 & 246304 & 123152 & 4 & 4 & 178141285 & 19499329 & 73568 & 40224140 & 147136 & 2911352 & 4183685 & 220705 & 220704 \\
$^4$L133 \teamprint{3,4} & 159001 & 1632 & 1632 & 3264 & 3264 & 9 & 6 & 985916371 & 475081 & 135322 & 1011500 & 292400 & 79351 & 158058 & 75157 & 155475 \\
\mymidrulegray$^3$D3 \teamprint{3} & 27622 & 1023 & 513 & 2046 & 1020 & 9 & 6 & 70704118 & 3235954 & 765 & 5501789 & 1272 & 91858 & 215967 & 1522 & 1521 \\
$^3$D4 \teamprint{3} & 524225 & 10924 & 5460 & 21840 & 10920 & 16 & 8 & \unk & \unk & \unk & \unk & \unk & 4043377 & 13749608 & 16381 & 16380 \\
$^4$D3 \teamprint{2,4} & 663472 & 6144 & 6144 & 12285 & 12285 & 9 & 6 & \unk & \unk & \unk & \unk & \unk & 514120 & 1217310 & 486442 & 1155144 \\
$^6$D2 \teamprint{2,4,6} & 524225 & 4096 & 4096 & 8190 & 8190 & 8 & 6 & 5879066753 & 1094865 & 701001 & 1869170 & 1202948 & 254758 & 457795 & 218570 & 389995 \\
$^6$D2 \teamprint{4,6} & 524225 & 5704 & 2488 & 10920 & 5460 & 16 & 8 & 4992649921 & 15032900 & 33905 & 25363692 & 57194 & 991861 & 2029546 & 46236 & 60717 \\
$^6$D2 \teamprint{6} & 524225 & 6584 & 1608 & 12922 & 3458 & 32 & 10 & 2126796737 & 126748497 & 2532 & 208964598 & 4382 & 3158364 & 7395885 & 5551 & 5550\\
\bottomrule
\end{tabular}%
}
\global\def\savethis{\thelastscalefactor}
\message{The scaling factor is \savethis}

    \caption{Game sizes of the equivalent representations proposed in the paper (that is, belief game and TB-DAG) on several standard parametric benchmark team games. See \Cref{sec:experiments} for a description of the games and for a detailed description of the meaning of each column. 
    Values denoted by `\unk' are missing due to out-of-time or out-of-memory errors.
    }
    \label{tab:game_size}
\end{table}

\begin{table}[!t]
    \centering
    
\noindent\adjustbox{scale=\savethis,center}{%
\newcommand{\mymidrulegray}{\arrayrulecolor{gray}\mymidrule\arrayrulecolor{black}}
\newcommand{\mymidrule}{
\cmidrule(lr){1-5}
\cmidrule(lr){6-8}
\cmidrule(lr){9-11}
}
\begin{tabular}{lrSSS|SSS|SSS}
\toprule
\multicolumn{5}{c|}{\bf Original game} & \multicolumn{3}{c|}{\bf CFR on TB-DAG} & \multicolumn{3}{c}{\bf Column generation} \\
\multicolumn{1}{c}{\bf Game $\{\pmin\}$} & \multicolumn{1}{c}{\bf \pmax Value} & \multicolumn{1}{c}{\bf Nodes} & \multicolumn{2}{c|}{\bf Information} & \multicolumn{3}{c|}{\bf This paper} & \multicolumn{3}{c}{\bf ZFCS22} \\
\multicolumn{1}{c}{$G$} & \multicolumn{1}{c}{$u^*$} & $|\mathcal H|$ & $\max_\mathcal P |P|$ & $\bm{k}$ & Init & $\eps=10^{-3}$ & $\eps=10^{-4}$ & Init & $\eps=10^{-3}$ & $\eps=10^{-4}$\\
\mymidrulegray
$^3$K3 \teamprint{3} & 0.000 & 151 & 6 & 6 & 0.00s & \cbox{0.7843137254901961,0.7843137254901961,1.0}{0.00s} & \cbox{0.7843137254901961,0.7843137254901961,1.0}{0.00s} & 0.00s & \cbox{0.7843137254901961,0.7843137254901961,1.0}{0.00s} & \cbox{0.7843137254901961,0.7843137254901961,1.0}{0.00s} \\
$^3$K4 \teamprint{3} & -0.042 & 601 & 12 & 8 & 0.01s & \cbox{0.7843137254901961,0.7843137254901961,1.0}{0.00s} & \cbox{0.7843137254901961,0.7843137254901961,1.0}{0.00s} & 0.00s & \cbox{0.7843137254901961,0.7843137254901961,1.0}{0.01s} & \cbox{0.823683198769704,0.823683198769704,0.9852364475201846}{0.02s} \\
$^3$K6 \teamprint{3} & -0.024 & 3001 & 30 & 12 & 1.03s & \cbox{0.7843137254901961,0.7843137254901961,1.0}{0.03s} & \cbox{0.8876585928489042,0.8876585928489042,0.9612456747404844}{0.12s} & 0.00s & \cbox{0.8962706651287966,0.8962706651287966,0.9580161476355248}{0.14s} & \cbox{0.8962706651287966,0.8962706651287966,0.9580161476355248}{0.14s} \\
$^3$K8 \teamprint{3} & -0.019 & 8401 & 56 & 16 & 1m 6s & \cbox{0.9598615916955017,0.8913494809688581,0.8913494809688581}{4.73s} & \cbox{1.0,0.7843137254901961,0.7843137254901961}{32.36s} & 0.01s & \cbox{0.7843137254901961,0.7843137254901961,1.0}{0.23s} & \cbox{0.8052287581699347,0.8052287581699347,0.9921568627450981}{0.32s} \\
$^3$K12 \teamprint{3} & -0.014 & 33001 & 132 & 24 & \unk & \cbox{1.0,0.7843137254901961,0.7843137254901961}{oom} & \cbox{1.0,0.7843137254901961,0.7843137254901961}{oom} & 0.01s & \cbox{0.7843137254901961,0.7843137254901961,1.0}{0.84s} & \cbox{0.7843137254901961,0.7843137254901961,1.0}{1.39s} \\
$^4$K5 \teamprint{3,4} & -0.037 & 7801 & 20 & 10 & 0.55s & \cbox{0.7843137254901961,0.7843137254901961,1.0}{0.03s} & \cbox{0.8322952710495963,0.8322952710495963,0.9820069204152249}{0.05s} & \unk & \unk & \unk \\
$^4$K5 \teamprint{4} & -0.030 & 7801 & 60 & 15 & 13.71s & \cbox{0.7843137254901961,0.7843137254901961,1.0}{1.59s} & \cbox{0.8778162245290273,0.8778162245290273,0.9649365628604383}{6.34s} & \unk & \unk & \unk \\
\mymidrulegray
$^3$L133 \teamprint{3} & 0.215 & 12688 & 9 & 6 & 0.49s & \cbox{0.7843137254901961,0.7843137254901961,1.0}{0.02s} & \cbox{0.8556708958093041,0.8556708958093041,0.9732410611303345}{0.05s} & 0.02s & \cbox{1.0,0.7843137254901961,0.7843137254901961}{24.89s} & \cbox{1.0,0.7843137254901961,0.7843137254901961}{45.96s} \\
$^3$L143 \teamprint{3} & 0.107 & 40409 & 16 & 8 & 1.39s & \cbox{0.7843137254901961,0.7843137254901961,1.0}{0.10s} & \cbox{0.8876585928489042,0.8876585928489042,0.9612456747404844}{0.48s} & 0.05s & \cbox{1.0,0.7843137254901961,0.7843137254901961}{2m 4s} & \cbox{1.0,0.7843137254901961,0.7843137254901961}{6m 3s} \\
$^3$L151 \teamprint{3} & -0.019 & 19981 & 20 & 10 & 1.54s & \cbox{0.7843137254901961,0.7843137254901961,1.0}{0.18s} & \cbox{0.8519800076893502,0.8519800076893502,0.9746251441753172}{0.50s} & 0.04s & \cbox{0.9543252595155709,0.9061130334486736,0.9061130334486736}{3.06s} & \cbox{0.9930795847750865,0.8027681660899654,0.8027681660899654}{13.98s} \\
$^3$L153 \teamprint{3} & 0.024 & 98606 & 25 & 10 & 16.03s & \cbox{0.7843137254901961,0.7843137254901961,1.0}{1.24s} & \cbox{0.8778162245290273,0.8778162245290273,0.9649365628604383}{4.94s} & 0.12s & \cbox{1.0,0.7843137254901961,0.7843137254901961}{7m 23s} & \cbox{1.0,0.7843137254901961,0.7843137254901961}{28m 13s} \\
$^3$L223 \teamprint{3} & 0.516 & 15659 & 4 & 4 & 0.13s & \cbox{0.7843137254901961,0.7843137254901961,1.0}{0.03s} & \cbox{0.8470588235294118,0.8470588235294118,0.9764705882352941}{0.08s} & 0.05s & \cbox{1.0,0.7843137254901961,0.7843137254901961}{13.48s} & \cbox{1.0,0.7843137254901961,0.7843137254901961}{18.53s} \\
$^3$L523 \teamprint{3} & 0.953 & 1299005 & 4 & 4 & 18.02s & \cbox{0.7843137254901961,0.7843137254901961,1.0}{11.26s} & \cbox{0.8384467512495194,0.8384467512495194,0.9797001153402537}{24.86s} & 6.83s & \cbox{1.0,0.7843137254901961,0.7843137254901961}{$>$ 6h} & \cbox{1.0,0.7843137254901961,0.7843137254901961}{$>$ 6h} \\
$^4$L133 \teamprint{3,4} & 0.147 & 159001 & 9 & 6 & 2.03s & \cbox{0.7843137254901961,0.7843137254901961,1.0}{0.21s} & \cbox{0.885198000768935,0.885198000768935,0.9621683967704728}{0.92s} & \unk & \unk & \unk \\
\mymidrulegray
$^3$D3 \teamprint{3} & 0.284 & 27622 & 9 & 6 & 0.80s & \cbox{0.7843137254901961,0.7843137254901961,1.0}{0.11s} & \cbox{0.8716647443291041,0.8716647443291041,0.9672433679354094}{0.40s} & 0.09s & \cbox{1.0,0.7843137254901961,0.7843137254901961}{11.05s} & \cbox{1.0,0.7843137254901961,0.7843137254901961}{11.05s} \\
$^3$D4 \teamprint{3} & 0.284 & 524225 & 16 & 8 & 1m 3s & \cbox{0.7843137254901961,0.7843137254901961,1.0}{22.54s} & \cbox{0.8778162245290273,0.8778162245290273,0.9649365628604383}{1m 28s} & 1.57s & \cbox{1.0,0.7843137254901961,0.7843137254901961}{3h 19m} & \cbox{1.0,0.7843137254901961,0.7843137254901961}{3h 19m} \\
$^4$D3 \teamprint{2,4} & 0.200 & 663472 & 9 & 6 & 27.05s & \cbox{0.7843137254901961,0.7843137254901961,1.0}{2.31s} & \cbox{0.8322952710495963,0.8322952710495963,0.9820069204152249}{4.70s} & \unk & \unk & \unk \\
$^6$D2 \teamprint{2,4,6} & 0.072 & 524225 & 8 & 6 & 10.74s & \cbox{0.7843137254901961,0.7843137254901961,1.0}{1.72s} & \cbox{0.8458285274894272,0.8458285274894272,0.9769319492502884}{4.26s} & \unk & \unk & \unk \\
$^6$D2 \teamprint{4,6} & 0.265 & 524225 & 16 & 8 & 16.55s & \cbox{0.7843137254901961,0.7843137254901961,1.0}{3.80s} & \cbox{0.8569011918492887,0.8569011918492887,0.9727797001153402}{11.09s} & \unk & \unk & \unk \\
$^6$D2 \teamprint{6} & 0.333 & 524225 & 32 & 10 & 31.00s & \cbox{0.7843137254901961,0.7843137254901961,1.0}{30.20s} & \cbox{0.8421376393694733,0.8421376393694733,0.9783160322952711}{1m 11s} & \unk & \unk & \unk\\
\bottomrule
\end{tabular}
}

    \caption{Runtime of our CFR-based algorithm (column `\textbf{This paper}') using the TB-DAG, compared to the prior state-of-the-art algorithms based on linear programming and column generation by \citet{Zhang22:Optimal} (`\textbf{ZFCS22}'), on several standard parametric benchmark games. See \Cref{sec:experiments} for a description of the games. \changed{Column {\em ``$u^*$''} is the TMECor value of the game.} Column \emph{``Init''} represents the time needed to construct the structures needed for solving the games. This corresponds to fully exploring the TB-DAG and computing its full representation in memory in the TB-DAG case. \changed{Columns ``$\eps=10^{-3}$'' and ``$\eps=10^{-4}$'' list the amount of time required for that algorithm to compute an $\eps$-approximate TMECor of the game.} Missing or unknown values are denoted by`\unk'. For each row, the background color of each runtime column is set proportionally to the ratio with the best runtime for the row, according to the logarithmic color scale \thecolorscale\!\!\!\!.
    Runtimes that are more than two orders of magnitude larger than the best runtime for the row (i.e., for which $R > 10^2$) are colored as if $R = 10^2$.\\
    }
    \label{tab:running_time}
\end{table}

This section investigates the empirical benefits brought about by applying the TB-DAG when computing mixed-Nash equilibria.
As highlighted in \Cref{sec:introduction,sec:related_works}, the literature on team games has been the one most concerned with the efficient computation of mixed Nash, with different works establishing benchmarks and proposing algorithms. We will, therefore, focus on comparing our approach against those previous related works.
Our main results are reported in \Cref{tab:game_size}, which reports the size of the original games and our derived representations, and in \Cref{tab:running_time}, which reports the time required to solve those instances up to an approximation factor.

\subsection{Postprocessing techniques}\label{sec:discussion:trivial_reductions}
In practice, \ref{al:tb-dag} is suboptimal in several ways. Here, we state some straightforward postprocessing techniques that can be used to shrink the size of the TB-DAG. These do not affect the theoretical statements as the primary focus of those is isolating the dependency on our parameters of interest, but they can significantly affect the practical performance, so we apply them in the experiments.
\begin{enumerate}
    \item If two terminal nodes $z, z'$ have the same sequence, we remove one of them (say, $z'$) from our DAG because it is redundant, and alias $\vx[\{z'\}]$ to $\vx[\{z\}]$. If this removal causes a section of the DAG to no longer contain any terminal descendants, we also remove that section.
    \item If a decision point in the TB-DAG has (at most) one parent and (at most) one child, we remove the decision point and directly connect the parent observation node to the grandchild decision points.
\end{enumerate}
In particular, if the team has perfect recall, the above two optimizations are sufficient for the TB-DAG to coincide with the sequence form.

\subsection{Experimental setting}

First, we give a complete description of the experimental setting in which the different algorithms are tested.

\paragraph{Game instances} We run experiments on commonly adopted parametric benchmarks in the team games literature. %
The following is the naming convention adopted for the instances considered:
\begin{itemize}
    \item {\bf $^n$K$r$}: $n$-player Kuhn poker with $r$ ranks~\citep{Kuhn50:Simplified}.
    \item {\bf $^n$L$brs$}: $n$-player Leduc poker with a $b$-bet maximum in each betting round, $r$ ranks, and $s$ suits~\citep{Southey05:Bayes}. 
    \item {\bf $^n$D$d$}: $n$-player Liar's Dice with one $d$-sided die for each player~\citep{Lisy15:Online}. 
\end{itemize}
The full description of these games can be found in \citet{Farina21:Connecting}.
For each game, the players belonging to team \tmin are represented along with the name. For example, $^4$L133 \teamprint{3,4} indicates a 4-player Leduc poker game with $1$ bet each round, $3$ ranks, $3$ suits, where players $3$ and $4$ belong to team \tmin and are therefore coordinated by player \pmin.

\paragraph{CFR Variant used}
We implemented the \textit{Predictive CFR$^+$ (PCFR$^+$)}~\citep{Farina21:Faster} state-of-the-art variant of CFR on the TB-DAG. PCFR$^+$ is a predictive regret minimization algorithm and uses quadratic averaging of iterates. At each time $t$, we use the previous utility vector for each time as the prediction for the next. 
We remark that applying the CFR algorithm on the belief game and on the TB-DAG leads to identical iterations since the two representations are structurally equivalent (as proven in \Cref{sec:tbdag}), and CFR is a deterministic algorithm. We therefore focus on the TB-DAG representation due to its efficiency. We also remark that the optimizations discussed in \Cref{se:bf reduction proof,sec:discussion:trivial_reductions} are applied during the experiments.

\paragraph{Baselines}
We use the column generation framework of \citet{Farina21:Connecting} and refined by \citet{Zhang22:Optimal} (henceforth ``\textbf{ZFCS22}'') as the prior state-of-the-algorithm to compare the performance of CFR on the team belief DAG. ZFCS22 belongs to the family of column generation approaches adopted in the past literature in team games and described in \Cref{sec:related_works}. ZFCS22 iteratively refines the strategy of each team by solving best-response problems using a tight integer program derived from the theory of extensive-form correlation~\citep{Stengel08:Extensive}. We used the original code by the authors, which was implemented for three-player games in which a team of two players faces an opponent.

\paragraph{Hardware used}
All experiments were run on a 64-core AMD EPYC 7282 processor. Each algorithm was allocated a maximum of 4 threads, 60GBs of RAM, and a time limit of 6 hours. ZFCS22 uses the commercial solver Gurobi to solve linear and integer linear programs. All CFR implementations are single-threaded, while we allowed Gurobi to use up to four threads.

\subsection{Discussion of the results}
We now discuss the empirical results obtained by our algorithms.

\paragraph{Representation vs Game size}
We analyze the size results from \Cref{tab:game_size}. The different orders of magnitude of the size of each representation and the original game highlight how the belief game construction increases the size of the game. Moreover, the striking difference between the two equivalent approaches of belief game and TB-DAG motivates the introduction of the latter: the direct construction of a decision problem and the more efficient representation brought by the DAG structure allow the construction of a substantially smaller representation. The benefits of the DAG imperfect-recall structure are especially beneficial in the case of Liar's Dice instances, which have a larger depth of the game tree.
Overall, this comparison confirms the results from the worst-case bounds from \Cref{sec:worst-case-size-aux,sec:worst_case_size_tbdag,sec:discussion:information_complexity}. The exponential factor of inefficiency between the two representations agrees with the results from the discussion in \Cref{sec:discussion:tree_vs_dag}.

There are also some minor remarks that are worth to be made. 
Whenever \pmin is a perfect-recall player (equivalently, when the team \tmin is composed of a single player), our constructions never increase the size of its decision problem. In the case of the belief game, we have that the adversary retains an identical number of information sets and sequences. In the case of the TB-DAG, the correspondence is $|\mathcal D_\pmin| = |\tilde{\mathcal I}_\pmin| + 1$ and $|\mathcal S_\pmin|= |\tilde{\Sigma}_\pmin|$

\paragraph{Running time}
We focus on the time performance of CFR applied to the games from \Cref{tab:running_time}. The main observation is that the TB-DAG approach combined with the CFR algorithm has good performance in most of the games traditionally employed in the team game literature. 
In particular, impressive performance is achieved in games where the information complexity is low. This is the case of Leduc and Liar's Dice benchmarks (whose number of infosets and sequences in the original game are reported in \Cref{tab:game_size}). On the other hand, the column generation approaches struggle since the dimension of the pure strategy space depends exponentially on the number of information sets.
The performance of our method depends crucially on having low information complexity. In fact, in games such as  $^3$K8 and $^3$K12 where the information complexity is high, we observe poor performance even though the game tree is small. %
On the other hand, column generation techniques avoid this cost by considering an incrementally larger action space. 

\section{Conclusion}\label{sec:conclusion}
This paper proposed a novel two-player zero-sum representation called the \emph{team-belief DAG} for the computation of mixed Nash equilibria in timeable two-player zero-sum imperfect-recall games and team max-min equilibria with correlation in adversarial team games. We proposed a conversion mechanism that can be interpreted from the point of view of a perfect-recall coordinator which manages all the player's strategic choices while not accessing any information destined to be imperfectly recalled. 
The behavior of such a coordinator is defined based on beliefs and observations, novel concepts that allow an intuitive yet effective characterization. 
We also introduced a DAG decision problem structure for the TB-DAG to characterize more efficiently our conversion, by avoiding the pitfalls of an extensive-form characterization of the equivalent game. %
We theoretically analyzed the efficiency of our method through worst-case bounding of the size of the converted game, and we experimentally tested it on a set of customary benchmark games against a state-of-the-art approach from the literature.
Our results are accompanied by novel complexity results that further characterize the hardness of computing equilibria in imperfect-recall games. In particular, we prove that computing a max-min strategy in behavioral strategies is $\Sigma_2^\P$-hard even when the number of players is constant and there is no chance. Similarly, we prove that computing a Nash equilibrium in mixed strategies is $\Delta_2^\P$-hard.

Many directions departing from this work can be interesting for further development of the literature on imperfect-recall and team games.
In particular, designing an algorithm able to exploit both the TB-DAG representation and the incrementality of column generation is an interesting approach to surpass the previous developments. Moreover, the TB-DAG construction may possibly be improved by preprocessing the game to reduce its information complexity, mitigating the exponential blowup due, while generalizing the notion of {\em triangle-free games}~\cite{Farina20:Polynomial} to DAG games may extend the class of games that can be solvable in polynomial time.
Another possible direction follows the more traditional two-player zero-sum literature. It aims to develop specific abstraction, dynamic pruning, and subgame-solving techniques tailored to our conversion's resulting two-player zero-sum games.
Finally, the question whether some of the results presented in the paper can be extended to the non-timeable or absent-minded imperfect-recall case is open.

\section*{Acknowledgements}
The work of Prof. Gatti’s research group is funded by the FAIR (Future Artificial Intelligence Research) project, funded by the NextGenerationEU program within the PNRR-PE-AI scheme (M4C2, Investment 1.3, Line on Artificial Intelligence). The work of Prof. Sandholm’s group is supported by the Vannevar Bush Faculty Fellowship ONR N00014-23-1-2876, National Science Foundation grants RI-2312342 and RI-1901403, ARO award W911NF2210266, and NIH award A240108S001. Brian Hu Zhang's work is also supported in part by the CMU Computer Science Department Hans Berliner PhD Student Fellowship.

\bibliographystyle{plainnat}
\bibliography{dairefs}

\appendix

\renewcommand{\thesection}{\Alph{section}} %
\makeatletter
\def\@seccntformat#1{\@ifundefined{#1@cntformat}%
   {\csname the#1\endcsname.\hspace{0.5em}}%
   {\csname #1@cntformat\endcsname}}%
\newcommand\section@cntformat{\appendixname\thesection.\hspace{0.5em}}
\makeatother

\section{Omitted Proofs from \Cref{sec:dag-decision-problems}}
\label{app:proof:dag-decision-problems}
\daggeneric*
\begin{proof}
Using the fact that $\vec D$ is a bijection, we have
    \begin{align}
    R^T_{\mc Q} &= \max_{\vx \in {\mc Q}} \sum_{t=1}^T \ip{\vu^t, \vx - \vx^t} 
    \\&= \max_{\hat\vx \in \hat{\mc Q}} \sum_{t=1}^T \ip{\vu^t, \vec D\vx - \vec D\hat\vx^t} 
    \\&= \max_{\hat\vx \in \hat{\mc Q}} \sum_{t=1}^T \ip{\vec D^\top \vu^t, \vx - \vx^t} = \hat R^T_{{\mc Q}}. \tag*{\qedhere}
    \end{align}
\end{proof}

\dagrm*
\begin{proof}
   Let $s \in \mc D$ be any decision point of $\mc Q$, and let $\hat s \in \hat{\mc D}$ be any decision point in $\hat{\mc Q}$ with $\delta(\hat s) = s$. It is enough to show that the sequence of utility vectors observed by $\mc R_{\hat s}$  when running \ref{al:dag-rm-generic} with \ref{al:cfr} as $\hat{\mc R}$ is the same as the sequence of utility vectors observed by $\mc R_s$ in \ref{al:dag-cfr}. We show this by induction on the decision points $\hat s$, leaves first.

   First, if $\hat s$ has no decision point descendants, then the claim is trivial because, by construction of $\vec D^\top$, we have $\hat\vu^t[\hat sa] = \vu^t[sa]$ for every $a \in A_s$. Now let $\hat s \in \hat{\mc D}$ be any internal node and $s = \delta(\hat s)$. By inductive hypothesis, for every decision point descendant $\hat s'$ of $\hat s$, at every timestep $t$, $\mc R_{\hat s'}$ receives the same utility vector as $\mc R_{s'}$ where $s' = \delta(\hat s')$, and thus produces the same behavioral strategy $\vec r_{s'}^t = \hat{\vec r}_{\hat s'}^t$. Thus, at any timestep $t$ the utility vector $\hat{\vec v}^t[\hat s*]$ that is passed to $\mc R_{\hat s}$ is given by
   \begin{align}
       \hat{\vec v}^t[\hat sa] &= \hat\vu^t[\hat sa] + \sum_{\hat s' : p_{\hat s'} = \hat sa} \ip{\hat{\vec r}^t_{\hat s'}, \hat{\vec v}^t[\hat s'*]}
       \\&= \vec u^t[sa] + \sum_{\hat s' : p_{\hat s'} = \hat sa} \ip{ {\vec r}^t_{s'}, {\vec v}^t[s'*]}
       \\&= \vec u^t[sa] + \sum_{s' : sa \in P_s} \ip{ {\vec r}^t_{s'}, {\vec v}^t[s'*]}
       \\&= \vec v^t[sa]
   \end{align}
where once again we use the notation $s' := \delta(\hat s')$, and the inductive hypothesis is used in the second equality on every term in the sum.
\end{proof}

\section{Omitted Proofs from \Cref{sec:tbdag}}
\label{app:proof:tbdag}

\fixedparameterhard*

\changed{
\begin{proof}
    We reduce from $k$-CLIQUE. Given a graph $G = (V, E)$, we construct the following game with a single team with two members. Nature selects two integers $j_1, j_2 \in [k]$ uniformly at random. Both team members observe their indices privately and select vertices $v_1(j_1), v_2(j_2) \in V$. If $j_1 = j_2$, the team scores $1$ if and only if $v_1 = v_2$; otherwise the team scores $0$. If $j_1 \ne j_2$, the team scores $1$ if and only if $(v_1, v_2)$ is an edge in the graph; otherwise the team scores $0$. This game has $2k$ information sets ($k$ for each player), so the information complexity is certainly bounded by $2k$.

    We claim that there is a strategy profile of the team that scores $1$ if and only if a $k$-clique exists in $G$. Indeed, if a $k$-clique exists in $G$, then the following profile scores $1$: let $(v(1), \dots, v(k))$ be a $k$-clique. Each player plays the pure strategy $v_i(j_i) = v(j_i)$.

    Conversely, a pure profile is a pair of functions $v_1, v_2 : [k] \to V$. If this profile scores $1$, then we must have $v_1 = v_2 =: v$ (because otherwise the team scores $0$ when both players are given an index $i$ for which $v_1(i) \ne v_2(i)$), and $(v(i), v(j))$ is an edge for all $(i, j)$ (otherwise the team scores $0$ when one player is given $i$ and the other player is given $j$). Thus $v$ is a $k$-clique.
\end{proof}
}

\section{Omitted Proofs from \Cref{sec:complexity}}
\label{app:proof:complexity}

\basichardness*

\begin{proof}
    We essentially follow the proof of \citet{Chu01:NP}, but we are more explicit about ordering the turns so that we can later bound the information complexity of our game.  Given a 3-CNF formula $\phi$ with $m$ clauses and $n$ variables, construct the following game with two team members on the maximizing team and no opponent. Nature picks a random clause $j \in [m]$ and a random variable $x_i$ that appears in clause $j$. P1 privately observes the variable and is asked to assign a value (either true or false) to it. P2 then privately observes the clause $j$ (but not the variable $i$, nor P1's assignment), and must pick one of the three variables in the clause. The team wins if P2 picks variable $i$ and P1 assigns the value to $i$ that makes the clause true.

    Note that the pure strategies of P1 are precisely the assignments $x \in \{0, 1\}^n$. We claim that, if P1 plays assignment $x$, the best value for the team is precisely $1/3$ of the maximum number of clauses satisfied by $x$. To see this, note that the team will lose with probability $2/3$ no matter what since P2 has only a $1/3$ chance of picking the same variable that was selected by Nature. Conversely, if P2 happens to pick the right variable, the team wins if and only if P1's assignment to that variable satisfies the clause $j$. Thus, P2's best strategy, given P1 strategy $x$, is always to pick the satisfying variable in the clause, if there is one. By construction, this achieves the desired value.

    Finally, by the PCP theorem~\cite{Haastad01:Some}, the above proof also shows that there exists an absolute constant $\eps$ such that computing the optimal value in a team game with no adversary to accuracy $\eps$ is \NP-hard. 
\end{proof}

\corkdependenceisoptimal*
\begin{proof}
The information complexity of the game used in the above construction is\footnote{Here we use the ordering of the players: namely, we have $k=n$ only because P1 plays before P2. If the order of the players were flipped, we would instead have $k=m$.} $k=n$, and the branching factor can be made an absolute constant by splitting the root chance node into $\Theta(\log m)$ layers. Finally, the size of the game is $O(mn)$. Thus, \Cref{th:main} implies a SAT-solving algorithm whose runtime is $2^{O(k)}$. Thus, in particular, if the $k$ were replaced by any $o(k)$ term, then SAT would have a $2^{o(n)}$-time algorithm, violating ETH.
\end{proof}

\tmesigmatwo*
\begin{proof}
Consider a behavioral max-min strategy represented by a distribution over the actions at each information set $I$. Let $\delta > 0$, and consider rounding each entry of the behavioral-form strategy by at most an additive $\delta$ so that the resulting strategy is rational. Let $\vx'$ be the correlation plan of the resulting strategy. Thus, for any given terminal node $s$, the resulting reach probability $x'[s]$ is perturbed by at most an additive $O(N \delta)$ where $N$ is the number of nodes in the game. Thus, $\norm{\vx' - \vx}_1 \le O(N^2 \delta)$. Thus, for any realization-form strategy $\vy$ for the opponent, we have $\abs{\ip{\vx' - \vx, \vec A \vy}} \le \norm{\vx' - \vx}_1 \norm{\vec A \vy}_\infty \le O(N^2 \delta)$, so $x'$ is $O(N^2 \delta)$-close to the optimal solution. Taking $\delta < O(\eps/N^2)$ thus concludes the proof.
\end{proof}

\tmesigmatwohard*
\begin{proof}
    We first give a reduction involving chance, then show how to relax this condition. We reduce from $\exists \forall$3-SAT, which is known to be $\Sigma_2^\P$-complete~\cite{Schaefer02:Completeness}. The $\exists \forall$3-SAT problem is to, given a $3$-DNF formula $\phi(X, Y)$, determine whether $\exists X\ \forall Y\ \phi(X, Y)$ holds.

    Given a 3-DNF formula $\phi$ with $m$ clauses, $n_1$ variables in $X$, and $n_2$ variables in $Y$, construct the following game between \tmax with $3$ players and \tmin with $3$ players. Nature chooses three variables $x_1, x_2, x_3$ from $X$ and three variables $y_1, y_2, y_3$ from $Y$. For each variable $x_i$ (respectively $y_i$), Player $i$ of \tmax (respectively \tmin) is asked for an assignment to the variable.

    If any two players of \tmax (respectively \tmin) have the same variable but differ in their assignment, \tmax gets value $-M$ (respectively $M$) where $M$ is a large value. In addition, \tmax gets value $1$ if at least one term in the 3-DNF $\phi$ is satisfied by the assignments of \tmax and \tmin.

    Let $n = \max(n_1, n_2)$. We complete the proof by showing that \tmax gets at least $1/n^3$ if and only if $\exists X\ \forall Y\ \phi(X, Y)$ holds; otherwise, their value is at most $0$. We first show that for large enough $M$, since players of \tmax cannot correlate, \tmax's pure strategies are dominant over non-pure ones.

    \begin{lemma}
        Let $x \in X$ be a variable and $p \leq 1/2$ be the probability that Player $i$ plays their less-likely action for $x$ in a behavioral strategy. If $p > 0$ and $M \geq n_1$, then this strategy is strictly dominated by the strategy under which Player $i$ only plays their more-likely action.
    \end{lemma}
    \begin{proof}
        Whenever variable $x$ is picked for Player $i$ and one of their teammate (probability strictly larger than $1/n_1^2$), the penalty incurred by the two players is strictly more than $(M/n_1^2)(p(1-q) + q(1-p)) \ge (M/n_1^2)(p + q(1-2p)) \ge (M/n_1^2)p$, where q is the probability that the teammate of Player $i$ picks Player $i$'s less-likely action for $x$. On the other hand, Player $i$ gains no more than $1$ by playing their less-likely action (probability $p/n_1$). Hence, if $M > n_1$, any strategy with $p > 0$ is strictly dominated by a pure strategy.
    \end{proof}
    The intuition behind this result is the following: the potential gain by deviating from a pure strategy is proportional to the probability of deviation ($p$ here), while the penalty is also proportional to this probability as well. Hence, non-zero deviation is strictly dominated if the penalty is large enough.

    The pure strategies of a player of \tmax (respectively \tmin) are precisely the assignments in $\{0,1\}^X$ (respectively in $\{0,1\}^Y$). By a similar argument, it is straightforward to show that it is a dominant strategy for \tmax (respectively for \tmin) to let all players pick the same assignment to avoid a large penalty.

    The hardness then follows from the following observation. If $\exists X\ \forall Y\ \phi(X, Y)$ holds, then \tmax can play the corresponding assignment to force a value of at least $1/n^3$: no matter what assignment \tmin picks, at least one term in $\phi$ is true, which is discovered with a probability of at least $1/n^3$ (whenever all the variables in such a term are picked by Nature). On the other hand, if $\exists X\ \forall Y\ \phi(X, Y)$ does not hold, then no matter what assignment \tmax picks, there is an assignment that \tmin can pick such that none of the terms is satisfied, which forces a value of $0$ for \tmax.

    \paragraph{Hardness for games of no chance} To show that the same hardness holds even when there is no chance, we use the following gadget to eliminate the need of Nature. Let us introduce a new \tmax-player called \tmax-Nature and a new \tmin-player called \tmin-Spoiler. The gadget will be such that \tmax-Nature can incur a large penalty whenever they do not mimic perfectly Nature's behavior. More concretely, as Nature in the construction above, \tmax-Nature picks $s = (x_1, x_2, x_3, y_1, y_2, y_3) \in X^3 \times Y^3$. \tmin-Spoiler then guesses \tmax-Nature's choice by picking $s' \in X^3 \times Y^3$. \tmax receives $-N(n_1^3 n_2^3 - 1)$ If $s = s'$, otherwise $N$, where $N$ is a large number. The game then continues as in the construction above.

    By a similar argument to the one used in the proof of the lemma above, \tmax-Nature's dominant strategy is to pick $s$ uniformly at random. Since \tmax cannot correlate, the game plays exactly like the construction above; \tmax can force a value of $1/n^3$ if and only if $\exists X\ \forall Y\ \phi(X, Y)$ holds.
\end{proof}

\tmecordeltatwo*
\begin{proof}
    Let $\mc X \subset \R^m, \mc Y \subset \R^n$ be the space of realization-form pure strategies of both players, and $\mc A$ be the payoff matrix. Then our goal is to decide whether the polytope
    \begin{align}
        \mc X^* := \qty{ \vx \in \R^m : \quad \mqty{
        \circled{1} & \vx \in \co \mc X,                            \\
        \circled{2} & \vy^\top \vec A \vx \le v\ \forall \vy \in \mc Y
            } }
    \end{align}
    is empty. We will show how to separate over $\mc X^*$ with a mixed-integer linear programming oracle, which suffices to complete the proof because such a separating oracle can be used to run the ellipsoid algorithm.

    Given a candidate solution $\vx$, we check both constraints. If $\circled 2$ is violated for some $\vy^* \in \mc Y$, then $\vec A \vy^*$ is a separating direction; such $\vy^*$ can be found by an integer programming oracle. If $\circled 1$ is violated, then a separating direction can be found because (strong) separation and optimization are equivalent for well-described polytopes~\cite{Grotschel93:Geometric}, and optimization over $\co \mc X$ is an integer program.
\end{proof}

\tmecordeltatwohard*
\begin{proof}
    We reduce from Last-SAT, which is known to be $\Delta_2^\P$-complete~\cite{Krentel88:Complexity}. The Last-SAT problem is to, given a 3-CNF formula $\phi(x)$, decide whether the lexicographically last satisfying assignment of $\phi$ has a $1$ in the least-significant bit.

    Given a 3-CNF formula $\phi$ with $m$ clauses over a set of $n$ variables $X = \{x_0, \ldots x_{n-1}\}$, we construct the following zero-sum game with $3$ players on each team. First, Nature chooses a triple $(y_1, y_2, y_3) \in X^3$ uniformly at random. Player~$i$ of \tmax (respectively of \tmin) is asked concurrently and independently to assign either true or false to the variable $y_i$. The payoff for \tmax is a sum of terms, determined by the following conditions.
    \begin{itemize}
        \item If any two players of \tmax (respectively of \tmin) assign different values to the same variable, \tmax receives $-N^2$ (respectively $+N^2$), where $N$ is a large number.
        \item If any clause in $\phi$ is rendered false by the assignment of \tmax (respectively \tmin), \tmax receives $-(n^2+N)$ (respectively $+N$).
        \item If $y_1 = x_k$ and player 1 of \tmax assigns true to this variable, then \tmax receives $+2^{n-k}$ (respectively $-2^{n-k}$).
        \item If $y_1 = x_{n-1}$ and player 1 of \tmax assigns false to this variable, then \tmax receives a penalty of $-1$.
    \end{itemize}

    We can see that when \tmin copies \tmax's mixed strategy:
    \begin{itemize}
        \item The first and the third term yield an expected payoff of $0$, since the positive payoff and the negative payoff cancels out exactly.
        \item The second term is non-positive, and strictly negative when at least one clause will be violated by an assignment in \tmax's mixed strategy.
        \item The fourth term is non-positive, and strictly negative when an assignment in \tmax's mixed strategy does not include the variable $x_{n-1}$.
    \end{itemize}

    If the maxmin value of the game is $0$, then there is a mixed strategy of \tmax such that the expected payoff is non-negative, even when \tmin copies this strategy. From the second term and the fourth term, we infer that this mixed strategy must be a satisfying assignment including $x_{n-1}$. If this assignment is not the lexicographically maximal one, then \tmax's expected payoff will be strictly negative when \tmin plays this maximal satisfying assignment, due to the third term of payoff. Hence, the satisfying assignment behind \tmax's mixed strategy must be lexicographically maximal; the Last-SAT instance is true.

    For the other direction, assume that the Last-SAT instance is true. Then \tmax can play the strategy corresponding to the lexicographically last satisfying assignment. When $N$ is large enough (e.g.\ $N = n^2 2^n$), the best \tmin can do is to play the same strategy. This means the expected payoff is at least $0$ for \tmax.

    Hence, the value of the game is $0$ if and only if the Last-SAT instance is true.

    \paragraph{Eliminate Nature} To eliminate Nature from the construction, we introduce a new \tmax-player called \tmax-Nature and a new \tmin-player called \tmin-Nature. Both new players are supposed to imitate Nature's uniform choice over the triples in $X^3$. The actual triple given to the 3 assignment-players of \tmax and the 3 assignment-players of \tmin is determined by the choice of \tmax-Nature and the one of \tmin-Nature.

    More concretely, the game proceeds as follows. First, \tmax-Nature picks a triple of indexes $I^+ = (i_1^+, i_2^+, i_3^+) \in \{0, \ldots, n-1\}^3$ without any other player observing the choice. Then \tmin-Nature also picks $I^- = (i_1^-, i_2^-, i_3^-) \in \{0, \ldots, n-1\}^3$ without any other player observing the choice. We write, for all $j$, $i_j = (i_j^+ + i_j^-) \mod n$. The game then continues as in the construction with a chance node, with $(y_1, y_2, y_3) = (x_{i_1}, x_{i_2}, x_{i_3}) \in X^3$.

    To see how this new construction with no chance node works, first notice that by implementing the uniform strategy (over the indexes), \tmax-Nature (resp. \tmin-Nature) can enforce a uniform distribution over $X^3$ for the triple $(x_{i_1}, x_{i_2}, x_{i_3})$, even when the distribution is conditioned on $I^-$ (resp. $I^+$). In particular, although the assignment-players of \tmax can correlate their assignment with the choice of \tmax-Nature (i.e., $I^+$), it does not do them any good when \tmin-Nature forces the conditional distribution on $I^+$ of $(y_1, y_2, y_3)$ to be uniform over $X^3$. The rest of the argument is similar to the construction with Nature above, with the consequence that the value of this game of no chance is $0$ if and only if the Last-SAT instance is true.
\end{proof}

\section{Discussion}\label{sec:discussion}

In this section, we discuss important details that may help the interested reader in clarifying some technical aspects of our contributions.

\subsection{Necessity of assumptions} \label{sec:assumptions}
\changed{
\paragraph{Timeability}
Timeability is used fundamentally in our TB-DAG construction, to allow the construction of the belief, which consists only of nodes at a single layer of the game tree. We think dealing with non-timeable games is a very interesting question for the future; doing so would require a new definition of information complexity $k$ that does not depend on the timing. Games with {\em absentmindedness} (which are a subset of non-timeable games) introduce other issues. In particular, mixtures of pure strategies are no longer sufficient to characterize optimal strategies (see, for example, \citet{Piccione97:Inperpretation}), so even our solution concepts break down. We refer the interested reader to \citet{Tewolde23:Computational} and papers cited therein for computational issues related to games with absentmindedness and without timeability.
\paragraph{Two-player zero-sum} The zero-sum assumption is fundamental to the efficient algorithms in the present paper: even two-player {\em normal-form} games are hard to solve if they are not zero sum~\cite{Rubinstein17:Settling}. The TB-DAG construction does apply to any imperfect-recall timeable game (even if it is not zero sum), allowing no-regret learning algorithms like CFR to be run on such games. However, CFR is not guaranteed to converge to a Nash equilibrium in such games, only to a {\em coarse-correlated equilibrium}.  
}

\subsection{Public states vs observations}\label{sec:discussion:obs}

\begin{figure*}[t]
    \input{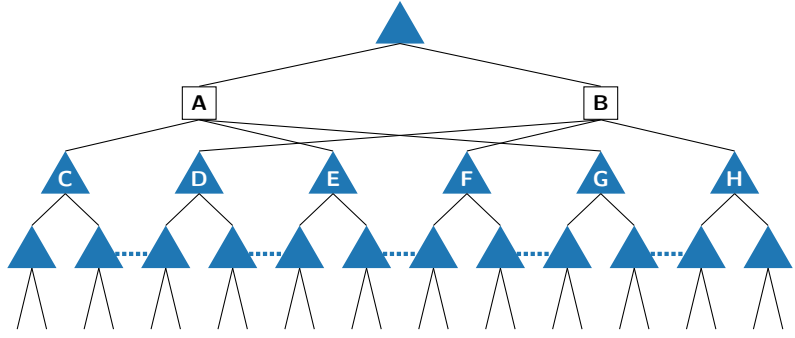}
  \begin{center}
  \begin{forest}
[\phantom X,p1
[A,name=A,nat
    [C,p1,parent=A
        [\phantom X,p1 [,terminal] [,terminal]] 
        [\phantom X,p1,name=1 [,terminal] [,terminal]]
    ]
    [D,p1,parent=B
        [\phantom X,p1,name=2 [,terminal] [,terminal]] 
        [\phantom X,p1,name=3 [,terminal] [,terminal]]
    ]
    [E,p1,parent=A
        [\phantom X,p1,name=4 [,terminal] [,terminal]] 
        [\phantom X,p1,name=5 [,terminal] [,terminal]]
    ]
]
[B,name=B,nat
    [F,p1,parent=B
        [\phantom X,p1,name=6 [,terminal] [,terminal]] 
        [\phantom X,p1,name=7 [,terminal] [,terminal]]
    ]
    [G,p1,parent=A
        [\phantom X,p1,name=8 [,terminal] [,terminal]] 
        [\phantom X,p1,name=9 [,terminal] [,terminal]]
    ]
    [H,p1,parent=B
        [\phantom X,p1,name=10 [,terminal] [,terminal]] 
        [\phantom X,p1,name=11 [,terminal] [,terminal]]
    ]
]
]
\draw[infoset1] (1) to (2);
\draw[infoset1] (3) to (4);
\draw[infoset1] (5) to (6);
\draw[infoset1] (7) to (8);
\draw[infoset1] (9) to (10);
  \end{forest}
  \end{center}
  \caption{
  A game showing that public state-based approaches do not subsume inflation.}\label{fi:inflation-counterexample}
\end{figure*}

In this section, we discuss in depth the difference between public {\em states} and public {\em observations}. Intuitively, the difference is that observations are {\em localized} to a particular node in the TB-DAG: if a fact is public to the team {\em conditional on the part of the team strategy that has been played to reach this point}, then it is an observation. On the other hand, public {\em states} only encode {\em unconditionally} public information. As we will see, using observations is strictly preferable to public states from both conceptual and theoretical perspectives.

\paragraph{Comparision to using public states}
We envision an alternative construction of the TB-DAG in which the team coordinator observes only the {\em public state} containing the current node. That is, the definition of \splitb{} is replaced by:
\begin{align}
    \label{def:splitbpub}
\splitb{i}^\text{pub}(H, h) \coloneqq H \cap P \text{ where } h \in P \in \P_i.
\end{align}
and $\splitb{i}^\text{pub}(H)$ defined analogously. Then, in \ref{al:tb-dag}, we replace $\splitb{i}(H)$ with $\splitb{i}^\text{pub}(H)$. We will call this new construction the {\em public-state TB-DAG} and spend the rest of this subsection contrasting it with the (observation) TB-DAG constructed by \ref{al:tb-dag}.

Our first result is that the TB-DAG can never be too much larger than the public state TB-DAG:

\begin{proposition}\label{pr:public state comparison}
    Let $N$ and $N'$ be the number of nodes in the TB-DAG and public state TB-DAG respectively. Then $N \le 2p N'$, where $p$ is the largest size (in number of nodes) of any belief in the public state TB-DAG.
\end{proposition}
\begin{proof}
Let $B$ be any belief in the public state TB-DAG. In the (non-public-state) TB-DAG, $B$ splits into disjoint beliefs $B_1, \dots, B_m$. Let $A_1, \dots, A_m$ be the sizes of the prescription spaces at $B_1, \dots, B_m$ respectively. Then $B$ has $A_1A_2 \dots A_m$ children, so $B$ induces $1 + A_1A_2 \dots A_m$ nodes in the public state TB-DAG. On the other hand, the beliefs $B_1, \dots, B_m$ in the TB-DAG will have $A_1, \dots, A_m$ children respectively, accounting for a total of $m + A_1 + \dots + A_m \le 2m A_1 \dots A_m$ nodes. Now, observing simply that $m \le p$ completes the proof.
\end{proof}

Thus, using observations is never much worse than using public states.

\paragraph{Comparision to using inflated public states}
{\em Complete inflation}~\cite{Kaneko95:Behavior}, which we simply call {\em inflation} for short, is an algorithm that splits an infoset $I$ into two infosets $I = I_1 \sqcup I_2$ if no pure strategy of the team can simultaneously play to a node in $I_1$ and a node in $I_2$, and repeats this process until no more such splits are possible. This preserves strategic equivalence. However, inflation can lead to the break-up of public states, which, in turn, reduces the size of public state TB-DAG. 

Indeed, consider the game in \Cref{fi:inflation-counterexample}. Due to the information sets marked in the last layer of the game tree, the connectivity graph contains a path C---D---E---...---H. Therefore, \{C, D, ..., H\} form a public state. Also, it is possible for the combinations CEG and DFH to be reached (if the player at the root plays left or right, respectively). Therefore, CEG and DFH are beliefs in the public-state TB-DAG. In the observation TB-DAG, consider, for example, what happens if the left action is played at the root so that C, E, and G are all reached. Note that there are no edges connecting C, E, and G---the path connecting C to E in the connectivity graph passes through D, which is not reached; therefore, C, E, and G are three different observations and hence three different beliefs, resulting in an exponentially-smaller TB-DAG. Inflation would remove the nontrivial information sets in the second black layer, which would ultimately have the same effect in this example as using observations. 

The number $3$ is not special in this construction;  it can be increased arbitrarily by simply increasing the number of children of {\sf A} and {\sf B}. Therefore, in particular, one can construct a family of games in which the public state TB-DAG (without inflation) has exponential size, while the (observation) TB-DAG has polynomial size. This is why \citet{Zhang22:Team_TreeDecomp} and \citet{Carminati21:Public} insist that inflation be done as a preprocessing step before beginning their constructions. 

The use of observations, however, removes the need for this step:

\begin{figure*}[t]
    \input{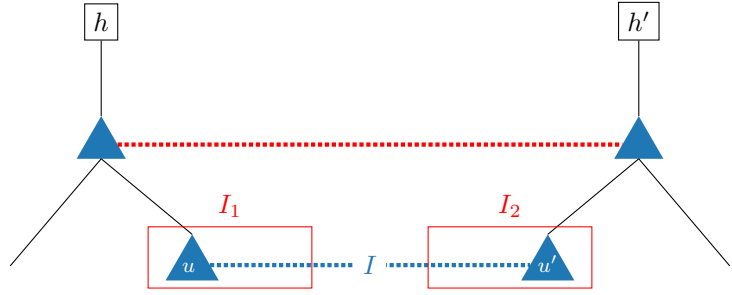}
  \begin{center}
  \begin{forest}
  for tree={l sep=1cm}
[,draw=none, s sep=4cm
    [$h$,no edge
        [\phantom X,name=x1,p1, s sep=2cm [,terminal] [ $u\phantom'$,name=u1,p1
        ]]
    ]
    [$h'$,no edge 
        [\phantom X,name=x2,p1, s sep=2cm [ $u'$,name=u2,p1
        ] [,terminal] ]
    ]
]
\draw[infoset1] (u1) -- (u2) node[midway,draw=none,fill=white]{$I$};
\draw[infoset1,draw=red] (x1) -- (x2);
\node[draw=none, fill=none, right=1cm of u1] (u1x) {};
\node[draw=none, fill=none, left=0cm of u1] (u1y) {};
\node[pubnode, fit=(u1y)(u1)(u1x), label=above:{\textcolor{red}{$I_1$}}] {};
\node[draw=none, fill=none, left=1cm of u2] (u2x) {};
\node[draw=none, fill=none, right=0cm of u2] (u2y) {};
\node[pubnode, fit=(u2y)(u2)(u2x), label=above:{\textcolor{red}{$I_2$}}] {};
  \end{forest}
  \end{center}
  \caption{
  A pictorial representation of the proof of \Cref{pr:inflate1}. Since $h$ and $h'$ can be played simultaneously but $u$ and $u'$ cannot, there must be an infoset like the red dotted one connecting a child of $h$ to a child of $h'$. Therefore, inflation cannot break existing edges between played nodes.}
\end{figure*}
\begin{proposition}\label{pr:inflate1}
    Given any team decision problem $\mathcal{T}$, the TB-DAG of $\mathcal{T}$ is the same no matter whether inflation is applied to $\mc T$ before the construction.
\end{proposition}
\begin{proof}
    Inflation operations affect the connectivity graph $\mc G_i$, thus changing the results of $\splitb{i}$ operations at a terminal node.
    Consider, therefore, any observation node $O$ in the TB-DAG, and let $h, h' \in O$, such that $I = I_1 \sqcup I_2$ is an inflatable infoset and $h \preceq u \in I_1$ and $h' \preceq u' \in I_2$. We need to show that inflating $I$ into $I_1$ and $I_2$ cannot remove the $(h, h')$ edge in $\mc G_i[O]$.

    Assume for contradiction that inflating would remove the $(h, h')$ edge and that therefore $\splitb{i}$ would split $h,h'$ into two different beliefs.
    We have that $O$ is a valid observation node, so it is possible for the player to play to both nodes $h$ and $h'$ simultaneously. But then there must be an infoset $I'$ connecting some node on the $h \to u$ path to some node on the $h' \to u'$ path---otherwise, it would be possible for the player to play to both $u$ and $u'$ simultaneously, which violates inflatability of $I$. But then there is still an $(h, h')$ edge in $\mc G_i[O]$, which is a contradiction.
\end{proof}

Although inflation {\em can} be performed efficiently, not requiring it as a preprocessing step simplifies the code and makes for a conceptually cleaner construction.
However, the benefits of observations go beyond making inflation unnecessary. In fact, even with inflation, there are still cases in which using observations instead represents an exponential improvement.

\begin{restatable}{proposition}{inflatetwo}\label{pr:inflate2}
    There exists a family of team decision problems in which the TB-DAG has a polynomial size, but the public state TB-DAG has exponential size, even if inflation is applied as a preprocessing step before building the latter.
\end{restatable}

\begin{proof}
\begin{figure}[t]
\noindent\adjustbox{max width=0.85\paperwidth,center}{
\input{figures/styles.tex}
\scalebox{0.8}{
\begin{forest}
for tree={anchor=center, s sep=0.4cm, l=1cm,}
[
    [\phantom1 ,p1,name=A1
        [,name=y3 [,name=y4 [,name=y5 [1,p1,name=y6
            [2,p1 [,terminal] [,terminal]] 
            [2,p1,name=1 [,terminal] [,terminal]]
        ]]]]
        [,terminal]
    ]
    [,name=x2 [\phantom1 ,p1,name=A2
        [[[1,p1
                [2,p1,name=2 [,terminal] [,terminal]] 
                [2,p1,name=3 [,terminal] [,terminal]]
        ]]]
        [,terminal]
    ]]
    [\phantom1 ,p1,name=B1
        [,terminal]
        [[\phantom1 ,p1,name=A3
            [[1,p1
                [2,p1,name=4 [,terminal] [,terminal]] 
                [2,p1,name=5 [,terminal] [,terminal]]
            ]]
            [,terminal]
        ]]
    ]
    [,name=x4 [\phantom1 ,p1,name=B2
        [,terminal]
        [[\phantom1 ,p1,name=A4
            [1,p1
                [2,p1,name=6 [,terminal] [,terminal]] 
                [2,p1,name=7 [,terminal] [,terminal]]
            ]
            [,terminal]
        ]]
    ]]
    [,name=x5 [[\phantom1 ,p1,name=B3
        [,terminal]
        [[1,p1
                [2,p1,name=8 [,terminal] [,terminal]] 
                [2,p1,name=9 [,terminal] [,terminal]]
        ]]
    ]]]
    [,name=x6 [[[\phantom1 ,p1,name=B4
        [,terminal]
        [1,p1
                [2,p1,name=10 [,terminal] [,terminal]] 
                [2,p1,name=11 [,terminal] [,terminal]]
        ]
    ]]]]
]
\draw[infoset1, bend right=20] (A1) to (B1);
\draw[infoset1, bend right=20] (A2) to (B2);
\draw[infoset1, bend right=20] (A3) to (B3);
\draw[infoset1, bend right=20] (A4) to (B4);
\draw[infoset1,] (1) to (2);
\draw[infoset1,] (3) to (4);
\draw[infoset1,] (5) to (6);
\draw[infoset1,] (7) to (8);
\draw[infoset1,] (9) to (10);
\node[above=1cm of A1.center,draw=none]{$c=1$};
\node[above=1cm of x2.center,draw=none]{$c=2$};
\node[above=1cm of B1.center,draw=none]{$c=3$};
\node[above=1cm of x4.center,draw=none]{$c=4$};
\node[above=1cm of x5.center,draw=none]{$c=5$};
\node[above=1cm of x6.center,draw=none]{$c=6$};
\node[left=1cm of y3,draw=none](t3label){$t=2$};
\node[above=0.5cm of t3label,draw=none](t2label){$t=1$};
\node[above=0.5cm of t2label,draw=none](t1label){$t=0$};
\node[left=1cm of y4,draw=none]{$t=3$};
\node[left=1cm of y5,draw=none]{$t=4$};
\node[left=1cm of y6,draw=none]{$t=5$};
\end{forest}
}
}
\caption{The counterexample for \Cref{pr:inflate2}, for $C=6$.}%
    \label{fi:counterexample}
\end{figure}

    The counterexample in \Cref{fi:inflation-counterexample} would work if it were not for the fact that all of the infosets in the last layer inflate. Therefore, we use a similar gadget at the bottom of the game to prove this result but ensure that inflation does nothing.

Consider the following family of games, parameterized by an integer $C > 1$. First, Nature picks an integer $c \in \{ 1, \dots, C\}$. Over the next $C-2$ layers $t=1, 2, \dots, C-2$, if $c \in \{ t, t+2 \}$, a player who cannot distinguish the two cases chooses an action $a \in \{ 0, 2 \}$. If $c = t+a$, then the game continues; otherwise, the game ends.

Finally, P1, who has perfect information about $c$ chooses between two actions numbered either $c$ or $c+1$. Then, player P2, observing P1's action number but {\em not} the value $c$, picks one of two options.

The resulting game is visualized in \Cref{fi:counterexample}. We observe the following things about it.

\begin{enumerate}
    \item No infoset inflates: all nontrivial infosets have size $2$, and it is easy to check that for all such infosets it is always possible to play to both nodes in them. This starkly contrasts the earlier counterexample, in which inflation was enough to achieve a small representation.
    \item Every P2-node in layer $C-1$ is in the same public state, and it is always possible to play to at least $C/2$ of them. Therefore, if using public-state-based beliefs, there will be a belief with $2^{C/2}$ prescriptions. Thus, the public-state-based team belief DAG, will have a size of at least $2^{C/2}$.
\end{enumerate}

We claim that layer $t \leq C-1$ does not have too many beliefs. Let $B \subseteq \H_t$ be a belief, and in the below discussion, let $[a..b]$ denote the set of integers $\{a, \dots, b\}$. 
\begin{enumerate}
    \item $B \subseteq [t+2..C]$. Since all nodes $j \ge t+2$ must be played to, we have $B = [t+2..C]$. 
    \item $B \not\subseteq [t+2..C]$. Then let $j = \max (B \setminus [t+2..C])$. Since $j \le t+1$, we have $j-2 \notin B$, since $j$ and $j-2$ are descendants of different actions taken at the infoset on layer $j-2 < t$. Thus $B$ does not contain any node $j' < j-2$ either, since in $\mathcal G[\H_t]$ such nodes $j'$ are only connected to $j$ through $j-2$. 

    Thus, $B \cap [1..t+1]$ is either $\{j\}$ or $\{j, j-1\}$. Further, since all nodes $j \ge t+2$ are played to and connected in $\mathcal G[\H_t]$, it follows that $B \cap [t+2..C]$ is either $[t+2..C]$ or empty. Thus, for each possible choice of $j \le t+1$, there are at most $4$ valid beliefs.
\end{enumerate} 
Thus, the number of beliefs in layer $t$ is at most $4(t+1)+1+2=4t+7$, where the $+2$ comes from counting the terminal beliefs, of which there are at most two. 

Further, at each belief $B$, we claim that the number of active information sets is, at most, a constant. For $t < C-1$ this is obvious since $\H_t$ contains only one information set (namely $\{t, t+2\}$). For $t=C+1$, by the above argument, we have $|B| \le 2$, so $B$ overlaps at most two information sets. 

Overall, we have that there are $O(C)$ beliefs at each of the $C$ layers of the game, and such beliefs never touch more than $O(1)$ different infosets. This is also true at the final layer because each infoset contains only two nodes at most.
It is therefore proven that the team belief DAG has a size of at most $O(C^2)$. 
\end{proof}

A practical experiment backs up these results. When $C = 16$, using observations generates a DAG with around 1000 edges; using public states generates a DAG with 30 million edges.

\subsection{Tree vs DAG representation}\label{sec:discussion:tree_vs_dag}

Here, we give an explicit example in which the TB-DAGs will be exponentially smaller than the game tree generated by \ref{al:auxgame}. This construction would work for most nontrivial adversarial team games, but for concreteness, consider the game $G$ depicted in \Cref{fig:example-atg}. Call the leftmost terminal node in that diagram $z$. Consider adding another copy of $G$ rooted in node $z$, and then repeating this process until $\ell$ copies of the game tree have been created, thus forming a game $G^\ell$. That is, $G^\ell$ is the game in which $G$ is played repeatedly until $\ell$ repetitions have been reached, or the terminal node reached is not $z$.

Note that when running \ref{al:auxgame} on $G$, multiple copies of node $z$ will appear. Thus, the number of nodes in the auxiliary game will be exponential in $\ell$. However, in the TB-DAG, after the $i$th repetition of the game finishes, the belief will always be $\{ z_i\}$ (where $z_i$ is the copy of $z$ in the $i$th repetition of the game). Thus, the size of the TB-DAG will scale linearly with $\ell$. Thus, as $\ell$ grows, the TB-DAG will be exponentially smaller than the auxiliary game, and in particular, the TB-DAG will have polynomial size while the auxiliary game will have exponential size.

\subsection{Definition of information complexity and comparison of bounds}
\label{sec:discussion:information_complexity}
We discuss the comparison between the bounds from \Cref{th:aux-size-upper} and \Cref{th:tb-dag-size} in more detail.

In \Cref{sec:beliefs_obs}, we defined the information complexity as the maximum number of {\em last-infosets} in any public state. This definition was made with \Cref{th:main} in mind, because it is the correct parameterization for that result. For \Cref{th:aux-size-upper}, however, we could have used a tighter parameterization. In particular, we could have defined a parameter $\kappa$ as the number of infosets (not last-infosets) in any public state. Then $O(b^{2\kappa d + d})$ would be a valid upper bound in \Cref{th:aux-size-upper}. One might ask how this new upper bound compares to that of \Cref{th:main}. To this end, we now compare the two bounds.

\begin{lemma}
    $k \le \kappa d$. 
\end{lemma}
\begin{proof}
    Every last-infoset at a public state $P$ will be an infoset intersecting some public state ancestor of $P$. Thus, there can be at most $\kappa d$ of these.
\end{proof}
Thus, the bound in \Cref{th:main} is at most
\begin{align}
    |\H| (b+1)^{k+1} \le |\H| (b+1)^{\kappa d} < |\H| b^{2 \kappa d} \le b^{2 \kappa d + d}
\end{align}
where we use the bounds $b \ge 2$ (which holds for every nontrivial game) and $|\H| \le b^d$. Thus, we conclude that the bound in \Cref{th:main} is always strictly tighter than the bound in \Cref{th:aux-size-upper}.

We also remark that in any case $\kappa \le |\H|$ is a loose bound that still ensures that the overall bound in \Cref{th:aux-tree-cfr} is polynomial in $|\H|$.

\subsection{Connection with tree decomposition}\label{sec:discussion:tree_decomp}

The public {\em state} TB-DAG can be viewed from the perspective of graphical models, specifically, using {\em tree decompositions}. Here, we review tree decompositions and show the tree decomposition-based perspective of the public-state TB-DAG.

\begin{definition}
    Given a (simple) graph $\mc G = (V, E)$, a {\em tree decomposition}\footnote{also known as a {\em clique tree} or {\em junction tree}} is a tree $\mc J$, with the following properties:
    \begin{enumerate}
        \item the nodes of $\mc J$ are subsets of $V$, called {\em bags};
        \item for each edge $(u, v) \in E$, there is a bag containing both $u$ and $v$; and
        \item for each vertex $u \in V$, the subset of nodes of $\mc J$ whose bags contain $u$ is connected.
    \end{enumerate}
    
\end{definition}

Consider an arbitrary set of the form 
\begin{align}
    \Pi = \{ \vx \in \{0, 1\}^n : g_k(\vx) = 0~~\forall k \in [m]\}
\end{align}
where the $g_k$s are arbitrary constraints, and as before let $\X = \co \Pi$. Each constraint $g_k$ has a {\em scope} $S_k \subseteq [n]$ of variables on which it depends. The {\em dependency graph} of $\Pi$ is the graph $\mc G_\Pi$ whose nodes are the integers $1, \dots, n$, and where there is an edge $(i, j)$ if there is a constraint whose scope $S_k$ contains both $i$ and $j$. For a subset $U \subseteq [n]$, a vector $\tilde\vx \in \{0, 1\}^U$ is {\em locally feasible} if $\tilde\vx = \vx_H$ for some $\vx \in \Pi$. We will use $\Pi_U$ to denote the set of all locally feasible vectors on $U$. Of course, $\Pi_{[n]} = \Pi$.  

The main result of interest to us is a corollary of the junction tree theorem (see, for example,~\citep{Wainwright08:Graphical}), which allows an arbitrary set $\X = \co \Pi$ to be described with a constraint system whose size is related to the sizes of tree decompositions of $\mc G_\Pi$.
\begin{theorem}[\citealp{Wainwright08:Graphical}]\label{th:tree-decomp}
    Let $\mc J$ be a tree decomposition of $\mc G_\Pi$. Then $\vx \in \X$ if and only if there are vectors $\vec\lambda_U \in \Delta(\Pi_H)$ for each bag $U$ of $\mc J$, such that:
    \begin{alignat}9
        \vx_U &= \sum_{\tilde\vx \in \Pi_U} \vec\lambda_U[\tilde \vx] \cdot \tilde\vx \qq{}&&\forall \textup{ bags $U$ in $\mc J$} \\
        \sum_{\substack{\tilde\vx \in \Pi_U \\ \tilde\vx_{U \cap V} = \tilde\vx^*}} \vec\lambda_U[\tilde\vx] &= \sum_{\substack{\tilde\vx \in \Pi_V \\ \tilde\vx_{U \cap V} = \tilde\vx^*}} \vec\lambda_V[\tilde\vx] &&\forall \textup{ edges $(U, V)$ of $\mc J$ and $\tilde\vx^* \in \Pi_{U \cap V}$}
    \end{alignat}
\end{theorem}
Intuitively, the first constraint says that every $\vx_U$ must be a convex combination of locally feasible $\tilde\vx \in \Pi_U$. This is of course a necessary condition. The second constraint says that marginal probabilities on edges $(U, V)$ must be consistent with each other. This is also clearly a necessary condition, so the difficulty of proving the above result lies in showing that these two constraints are {\em sufficient}. We will not prove the result here, but we will use it as a black box. 

In this section, we will work with a representation slightly different from the realization form. For a player $i$ in a coordinator game $G$ and a pure strategy of that player, the {\em history form} of the strategy as the vector $\vx \in \{0, 1\}^\H$ where $\vx[h]=1$ if and only if the team plays all actions on the $\Root \to h$ path. (Of course, the realization form is just the subvector of $\vx$ indexed by $\Z$.) As usual we will use $\Pi$ for the set of pure strategies in history form, and $\vx = \co\Pi$. The history form is the set of vectors $\vx  \in \{0, 1\}^\H$ satisfying the following constraint system. 
\begin{alignat}9
    \vx[\Root] &= 1 \\
    \vx[ha] &= \vx[h] &&\qif h \notin \H_i \\
    \vx[h] &= \sum_{a \in A_h} \vx[ha] &&\qif h \in \H_i \\
    \vx[ha] \vx[h'] &= \vx[h'a] \vx[h] &&\qif h, h' \in I \in \I_i; a \in A_h
\end{alignat}
This constraint system defines a dependency graph $\mc G_\Pi$, whose nodes are nodes of the tree, and in which there is an edge $(h, h')$ if either $h'$ is a child of $h$, or $h$ and $h'$ are in the same infoset of player $i$. 

Now consider the following tree decomposition of  $\mc J$  of $\mc G_\Pi$. For each public state $P$, the tree decomposition $\mc J$ has a bag $U_P$ that contains all nodes in $P$ and all children of nodes in $P$. The edges of $\mc J$ are the obvious edges, connecting each $U_P$ to $U_{P'}$ if $U_P \cap U_{P'} \ne \emptyset$. 

One can check that, up to trivial reformulations (that is, removal of redundant variables and constraints), the constraint system from \Cref{th:tree-decomp} associated with $\mc J$ is identical to the constraint system associated with the public state TB-DAG (via \eqref{eq:dag_strategy_root}). Thus, it is possible to interpret the public state TB-DAG entirely from the point of view of tree decompositions. We do not take this perspective in the rest of the paper because using beliefs is more interpretable and understandable from a game-theoretic perspective.

\end{document}